\documentclass[11pt,reqno,a4paper]{amsart}
\usepackage[latin1]{inputenc}
\usepackage{amsbsy}
\usepackage{amsmath}
\usepackage{amssymb}
\usepackage{amsfonts}
\usepackage{amsthm}
\usepackage{mathrsfs}
\usepackage{bm}
\usepackage{dsfont}
\usepackage{upgreek}
\usepackage{slashed}

\DeclareRobustCommand{\SkipTocEntry}[4]{} 
\newcommand{\1}{\mathds{1}}

\newcommand{\C}{\mathbb{C}}

\newcommand{\N}{\mathbb{N}}

\newcommand{\R}{\mathbb{R}}

\newcommand{\bA}{{\bm{A}}}
\newcommand{\bF}{{\bm{F}}}

\newcommand{\boB}{\mathcal{B}}
\newcommand{\boC}{\mathcal{C}}

\newcommand{\cF}{\mathcal{F}}
\newcommand{\boF}{\mathcal{F}}

\newcommand{\boH}{\mathcal{H}}
\newcommand{\boI}{\mathcal{I}}
\newcommand{\boJ}{\mathcal{J}}

\newcommand{\boQ}{\mathcal{Q}}
\newcommand{\boR}{\mathcal{R}}
\newcommand{\boS}{\mathcal{S}}
\newcommand{\boT}{\mathcal{T}}

\newcommand{\boW}{\mathcal{W}}

\newcommand{\bsalpha}{\bm{\alpha}}
\newcommand{\bsbeta}{\bm{\beta}}

\newcommand{\bssigma}{\bm{\sigma}}
\newcommand{\bsSigma}{\bm{\Sigma}}

\newcommand{\ga}{\mathfrak{a}}
\newcommand{\gm}{\mathfrak{m}}
\newcommand{\gv}{\mathfrak{v}}

\newcommand{\gS}{\mathfrak{S}}

\DeclareMathOperator{\curl}{{\rm curl}}
\newcommand{\E}{\mathscr{L}_{\rm PV}}
\DeclareMathOperator{\tr}{{\rm tr}}
\DeclareMathOperator{\sign}{{\rm sign}}

\renewcommand{\div}{\mathop{\mathrm{div}}\nolimits}

\newcommand\ii{{\ensuremath {\infty}}}
\newcommand\pscal[1]{{\ensuremath{\left\langle #1 \right\rangle}}}
\newcommand{\norm}[1]{ \left| \! \left| #1 \right| \! \right| }
\newcommand{\Hdiv}{\dot{H}^1_{\rm div}(\R^3)}

\newtheorem{cor}{Corollary}
\newtheorem{lemma}{Lemma}
\newtheorem{prop}{Proposition}

\newtheorem{theorem}{Theorem}
\theoremstyle{definition}
\newtheorem*{merci}{Acknowledgements}
\newtheorem{remark}{Remark}
\numberwithin{cor}{section}
\numberwithin{equation}{section}
\numberwithin{lemma}{section}
\numberwithin{prop}{section}
\numberwithin{remark}{section}
\numberwithin{theorem}{section}

\begin{document}

\title[Dirac's vacuum in electromagnetic fields]{Construction of the Pauli-Villars-regulated Dirac vacuum in electromagnetic fields}

\author[P. Gravejat]{Philippe GRAVEJAT}
\address{Centre de Math\'ematiques Laurent Schwartz (UMR 7640), \'Ecole Polytechnique, F-91128 Palaiseau Cedex, France.}
\email{gravejat@math.polytechnique.fr}

\author[C. Hainzl]{Christian HAINZL}
\address{Mathematisches Institut, Auf der Morgenstelle 10, D-72076 T\"ubingen, Germany.}
\email{christian.hainzl@uni-tuebingen.de}

\author[M. Lewin]{Mathieu LEWIN}
\address{CNRS \& Laboratoire de Math\'ematiques (UMR 8088), Universit\'e de Cergy-Pontoise, F-95000 Cergy-Pontoise, France.}
\email{mathieu.lewin@math.cnrs.fr}

\author[\'E. S\'er\'e]{\'Eric S\'ER\'E}
\address{Ceremade (UMR 7534), Universit\'e Paris-Dauphine, Place du Mar\'echal de Lattre de Tassigny, F-75775 Paris Cedex 16, France.}
\email{sere@ceremade.dauphine.fr}

\date{December 11, 2012. Final version to appear in \textit{Arch. Rat. Mech. Anal.}}

\begin{abstract}
Using the Pauli-Villars regularization and arguments from convex analysis, we construct solutions to the classical time-independent Maxwell equations in Dirac's vacuum, in the presence of small external electromagnetic sources. The vacuum is not an empty space, but rather a quantum fluctuating medium which behaves as a nonlinear polarizable material. Its behavior is described by a Dirac equation involving infinitely many particles. The quantum corrections to the usual Maxwell equations are nonlinear and nonlocal. Even if photons are described by a purely classical electromagnetic field, the resulting vacuum polarization coincides to first order with that of full Quantum Electrodynamics.
\end{abstract}

\maketitle

\tableofcontents

%%%%%%%%%%%%%%%%%%%%%%
%%%%%%%%%%%%%%%%%%%%%%
%%%%%%%%%%%%%%%%%%%%%%
\section{Introduction}
%%%%%%%%%%%%%%%%%%%%%%
%%%%%%%%%%%%%%%%%%%%%%
%%%%%%%%%%%%%%%%%%%%%%

In classical Physics, a time-independent external density of charge $\rho_{\rm ext}$ and a charge current $j_{\rm ext}$ induce a static electromagnetic field $(E_{\rm ext}=-\nabla V_{\rm ext},B_{\rm ext}=\curl A_{\rm ext})$, which solves Maxwell's equations in Coulomb gauge:
\begin{equation}
\begin{cases}
-\Delta V_{\rm ext}= 4\pi\,e\,\rho_{\rm ext},\\
-\Delta A_{\rm ext}= 4\pi\,e\,j_{\rm ext},\\
\div A_{\rm ext}=0,
\end{cases}
\label{eq:Maxwell}
\end{equation}
where $e$ is the elementary charge. It is convenient to gather the electrostatic and magnetic potentials in a unique vector $\bA_{\rm ext}=(V_{\rm ext},A_{\rm ext})$ called the \emph{four-potential}, which we will do in the whole paper. The electromagnetic potential $\bA_{\rm ext}$ solving~\eqref{eq:Maxwell} is the unique critical point of the (time-independent) Maxwell Lagrangian action functional
\begin{multline}
\mathscr{L}^{\rho_{\rm ext},j_{\rm ext}}(\bA)=\frac1{8\pi}\int_{\R^3}\Big(|\nabla V(x)|^2-|\curl A(x)|^2\Big)\,dx\\
-e\int_{\R^3}\rho_{\rm ext}(x)V(x)\,dx+e\int_{\R^3}j_{\rm ext}(x)\cdot A(x)\,dx,
\end{multline}
which is strictly convex with respect to $V$ and strictly concave with respect to $A$. In particular we can obtain $\bA_{\rm ext}$ by a min-max procedure:
$$\mathscr{L}^{\rho_{\rm ext},j_{\rm ext}}(\bA_{\rm ext})=\min_{V}\max_A\mathscr{L}^{\rho_{\rm ext},j_{\rm ext}}(V,A)=\max_A\min_{V}\mathscr{L}^{\rho_{\rm ext},j_{\rm ext}}(V,A)$$
where the constraint $\div A=0$ is always assumed.

The situation is much more complicated in Dirac's vacuum. It has been known for a long time that, in reality, the vacuum is not an empty space, but rather a quantum fluctuating medium which behaves as a nonlinear polarizable material~\cite{EulKoc-35,Euler-36,HeiEul-36,Schwinger-51a,GreSch-08}. In this medium,  virtual electron-positron pairs induce a polarization in response to external fields. The resulting electromagnetic field which is observed in experiments has to take into account the vacuum polarization effects. The corresponding four-potential $\bA_*$ solves coupled nonlinear Maxwell equations of the form
\begin{equation}
\begin{cases}
-\Delta V_*= 4\pi\,e\,\big(\rho_{\rm vac}(e\bA_*)+\rho_{\rm ext}\big),\\
-\Delta A_*= 4\pi\,e\,\big(j_{\rm vac}(e\bA_*)+j_{\rm ext}\big),\\
\div A_*=0.
\end{cases}
\label{eq:Maxwell-vp}
\end{equation}
Here $\rho_{\rm vac}(e\bA_*)$ and $j_{\rm vac}(e\bA_*)$ are respectively the charge density and the charge current induced in the vacuum. As we shall explain, they are nonlinear and nonlocal functions of $e\bA_*$.

The Dirac vacuum is described by Quantum Field Theory, that is, by a second-quantized fermion field. The charge and current densities 
$\rho_{\rm vac}(e\bA_*)$ and $j_{\rm vac}(e\bA_*)$ are obtained by minimizing the energy of this field in the presence of the given potential $e\bA_*$. In this model the interaction between the Dirac particles is mediated by the classical electromagnetic field which accounts for photons. This approach to vacuum polarization is usual in the Physics literature (see, e.g.,~\cite{Schwinger-51a,GreRei-08}).

The main idea behind the nonlinear Maxwell equations~\eqref{eq:Maxwell-vp} is that the vacuum behaves as a nonlinear medium, and the form of the equation is reminiscent of nonlinear optics. The nonlinear effects are in practice rather small since $e$ has a small physical value, but they become important in the presence of strong external sources, which can produce electron-positron pairs in the vacuum. Already in 1933, Dirac computed in~\cite{Dirac-34b} the first order term obtained by expanding $\rho_{\rm vac}(e\bA_*)$ in powers of $e$. The nonlinear equations~\eqref{eq:Maxwell-vp} was then studied by Euler, Heisenberg, Kockel and Weisskopf among others~\cite{EulKoc-35,Euler-36,HeiEul-36,Weisskopf-36}. In a celebrated paper, Schwinger~\cite{Schwinger-51a} used~\eqref{eq:Maxwell-vp} (and a time-dependent version of it) to derive the probability of pair creation by tunneling in a strong electrostatic field. For more recent works on the subject, see the references in~\cite{GreRei-08}. Several ongoing experiments aim at detecting some nonlinear effects of the vacuum in the laboratory~\cite{Zavattini-10}.

Like for the usual Maxwell equations~\eqref{eq:Maxwell}, the nonlinear equations~\eqref{eq:Maxwell-vp} in Dirac's vacuum arise from an effective Lagrangian action, which now includes the vacuum energy 
\begin{multline}
\mathscr{L}^{\rho_{\rm ext},j_{\rm ext}}(\bA)=\frac1{8\pi}\int_{\R^3}\Big(|\nabla V(x)|^2-|\curl A(x)|^2\Big)\,dx\\
-e\int_{\R^3}\rho_{\rm ext}(x)V(x)\,dx+e\int_{\R^3}j_{\rm ext}(x)\cdot A(x)\,dx-\boF_{\rm vac}(e\bA).
\label{eq:effective_Lagrangian_intro}
\end{multline}
Here $\boF_{\rm vac}(e\bA)$ is the ground state energy of Dirac's vacuum in the potential $e\bA$. The densities of the vacuum are then defined by
\begin{equation}
e\,\rho_{\rm vac}(\bA):=\frac{\partial}{\partial V}\boF_{\rm vac}(e\bA)\quad\text{and}\quad e\,j_{\rm vac}(\bA):=-\frac{\partial}{\partial A}\boF_{\rm vac}(e\bA). 
\label{eq:diff_PV_intro}
\end{equation}
We note that the vacuum correction $-\boF_{\rm vac}(e\bA)$ to Maxwell's Lagrangian has been computed to first order in the semi-classical approximation in~\cite{KarNeu-50,Schwinger-51a}.

It is not so easy to provide a rigorous definition of the vacuum energy $\boF_{\rm vac}(e\bA)$. It is well-known that this quantity is divergent in the high energy regime and an ultraviolet regularization has to be imposed. In this paper we use the famous Pauli-Villars regularization method~\cite{PauVil-49} to properly define the vacuum energy $\boF_{\rm vac}(e\bA)$ (Theorem~\ref{thm:defF} below). Then we are able to state our main result (Theorem~\ref{thm:polarized} below), which gives the existence of a critical point of the effective Lagrangian action~\eqref{eq:effective_Lagrangian_intro}, when the external sources $\rho_{\rm ext}$ and $j_{\rm ext}$ are not too large. As a corollary, we obtain solutions to the nonlinear Maxwell equations~\eqref{eq:Maxwell-vp}.

This article is the continuation of several works dealing with the Hartree-Fock approximation of Quantum ElectroDynamics (QED), some of them in collaboration with Solovej,~\cite{HaiLewSer-05a, HaiLewSer-05b, HaiLewSol-07, HaiLewSerSol-07, GraLewSer-09, GraLewSer-11}, and which originated from a seminal paper of Chaix and Iracane~\cite{ChaIra-89}. There, only the purely electrostatic case was considered. To our knowledge, the present work is the first dealing with electromagnetic fields in interaction with Dirac's vacuum.

\begin{merci}
M.L. and \'E.S. acknowledge support from the French Ministry of Research (Grant ANR-10-0101). M.L. acknowledges support from the European Research Council under the European Community's Seventh Framework Programme (FP7/2007-2013 Grant Agreement MNIQS 258023). We are also grateful to the referees for their careful reading of the manuscript.
\end{merci}

%%%%%%%%%%%%%%%%%%%%%%%%%
%%%%%%%%%%%%%%%%%%%%%%%%%
\section{Main results}

%%%%%%%%%%%%%%%%%%%%%%%%%%%%%%%%%%%%%%%%%%%%%%%%%%%%%%%%%%%%%%%%%%%%%%%%%%%%%%%%%
\subsection{Elementary properties of electromagnetic Dirac operators}
\label{sub:preliminary}
%%%%%%%%%%%%%%%%%%%%%%%%%%%%%%%%%%%%%%%%%%%%%%%%%%%%%%%%%%%%%%%%%%%%%%%%%%%%%%%%%

Before entering the main subject of this article, we recall some elementary spectral properties of the Dirac operator in the presence of electromagnetic fields~\cite[Chap. 4]{Thaller}. 

We work in a system of units such that the speed of light and Planck's reduced constant are both set to one, $c=\hbar=1$. We introduce the Dirac operator with mass $m$, elementary charge $e$ and electromagnetic four-potential $\bA = (V, A)$,
\begin{equation}
\label{def:DmA}
D_{m, e\bA} := \bsalpha \cdot \big( - i \nabla - e A(x) \big)+eV(x) + m \bsbeta
\end{equation}
which is an operator acting on $L^2(\R^3,\C^4)$. Here the four Dirac matrices $\bsalpha = (\bsalpha_1, \bsalpha_2 , \bsalpha_3)$ and $\bsbeta$ are equal to
$$\bsalpha_k := \begin{pmatrix} 0 & \bssigma_k \\ \bssigma_k & 0 \end{pmatrix} \quad {\rm and} \quad \bsbeta := \begin{pmatrix} I_2 & 0 \\ 0 & - I_2 \end{pmatrix},$$
the Pauli matrices $\bssigma_1$, $\bssigma_2$ and $\bssigma_3$ being defined by
$$\bssigma_1 := \begin{pmatrix} 0 & 1 \\ 1 & 0 \end{pmatrix},\quad \bssigma_2 := \begin{pmatrix} 0 & - i \\ i & 0 \end{pmatrix} \quad {\rm and} \quad \bssigma_3 := \begin{pmatrix} 1 & 0 \\ 0 & - 1 \end{pmatrix}.$$
The spectrum of the free Dirac operator is not semi-bounded~\cite{Thaller},
$$\sigma(D_{m,0})=(-\ii,-m]\cup[m,\ii).$$
As we will recall below, the unbounded negative spectrum of $D_{m,0}$ led Dirac to postulate the existence of the positron, and to assume that the vacuum is a much more complicated object than expected. The mathematical difficulties arising from the negative spectrum are reviewed for instance in~\cite{EstLewSer-08}.

In our setting, the natural space for the four-potential $\bA=(V,A)$ is the Coulomb-gauge homogeneous Sobolev space
\begin{multline}
\label{eq:def_homogeneous_Sobolev}
\Hdiv := \Big\{ \bA = (V, A) \in L^6(\R^3, \R^4) \ : \\ \div A = 0 \ {\rm and} \ \bF=(-\nabla V, \curl A) \in L^2(\R^3, \R^6) \Big\},
\end{multline}
endowed with its norm
\begin{equation}
\| \bA \|_{\Hdiv}^2 := \| \nabla V \|_{L^2(\R^3)}^2 + \| \curl A \|_{L^2(\R^3)}^2 = \| \bF \|_{L^2(\R^3)}^2.
\end{equation}
Here and everywhere, the equation $\div A=0$ is understood in the sense of distributions. The requirement $\bF\in L^2(\R^3)$ simply says that the electrostatic field $E=-\nabla V$ and the magnetic field $B=\curl A$ have a finite energy,
$$\int_{\R^3}|E^2|+|B|^2<\ii.$$

\begin{lemma}[Elementary spectral properties of $D_{m,\bA}$]
\label{lem:spectre}
Let $m > 0$.

\smallskip

\noindent $(i)$ Any four-potential $\bA \in \Hdiv$ is $D_{m, 0}$--compact. The operator $D_{m, \bA}$ is self-adjoint on $H^1(\R^3)$ and its essential spectrum is 
$$\sigma_{\rm ess}(D_{m, \bA}) = (- \infty, - m] \cup [m, \infty).$$

\smallskip

\noindent $(ii)$ The eigenvalues of $D_{m, \bA}$ in $(- m, m)$ are Lipschitz functions of $\bA$ in the norm $\| \bA \|_{\Hdiv}$.

\smallskip

\noindent $(iii)$ There exists a universal constant $C$ such that, if 
\begin{equation}
\label{eq:cond-petit-A}
\| \bA \|_{\Hdiv} \leq \eta \sqrt{m},
\end{equation}
for some number $\eta < 1/C$, then
$$\sigma(D_{m, \bA}) \cap (- m (1 - C \eta), (1 - C \eta) m) = \emptyset.$$

\smallskip

\noindent $(iv)$ Finally, if $V \equiv 0$, then $\sigma(D_{m, \bA}) \cap (- m, m) = \emptyset$.
\end{lemma}

\begin{proof}
Recall the Kato-Seiler-Simon inequality (see~\cite{SeiSim-75} and~\cite[Thm. 4.1]{Simon-79})
\begin{equation}
\label{eq:KSS}
\forall p \geq 2, \ \big\| f(x) g(-i \nabla) \big\|_{\gS_p} \leq \frac{1}{(2 \pi)^\frac{3}{p}} \big\| f \big\|_{L^p} \big\| g \big\|_{L^p},
\end{equation}
where $\gS_p$ is the usual Schatten class~\cite{Simon-79}. Applying~\eqref{eq:KSS} with $p = 6$ together with the Sobolev inequality, we obtain 
$$\Big\| V \frac{1}{D_{m,0}} \Big\|_{\gS_6} \leq \frac{C}{\sqrt{m}} \big\| V \big\|_{L^6} \leq \frac{C}{\sqrt{m}} \big\| \nabla V \big\|_{L^2},$$
and, similarly,
$$\Big\| \bsalpha \cdot A \frac{1}{D_{m,0}} \Big\|_{\gS_6} \leq \frac{C}{\sqrt{m}} \big\|A \big\|_{L^6} \leq \frac{C}{\sqrt{m}} \big\| \curl A \big\|_{L^2},$$
where we have used that $\div A = 0$. Since all the operators in $\gS_6$ are compact, statements $(i)$ and $(ii)$ follow from usual perturbation theory~\cite{Kato,ReeSim2}. Concerning $(iii)$, we notice that
$$D_{m, \bA}\,(D_{m, 0})^{-1} = \Big( I + \big( V - \bsalpha \cdot A \big) \frac{1}{D_{m, 0}} \Big) $$
so that, under condition~\eqref{eq:cond-petit-A},
$$\big| D_{m, \bA} \big| \geq \big( 1 - C \eta) \big| D_{m, 0} \big|.$$
Statement $(iii)$ then follows from $(i)$, whereas $(iv)$ is~\cite[Thm 7.1]{Thaller}.
\end{proof}

%%%%%%%%%%%%%%%%%%%%%%%
\subsection{The Pauli-Villars-regulated vacuum energy}
\label{sub:definition}

\subsubsection{Derivation}
We are now ready to define the energy of the vacuum, using Quantum Field Theory. We consider a fermionic second-quantized field, placed in a given electromagnetic potential $\bA$. Later, the potential $\bA$ which describes light and external sources will be optimized. In this section it is kept fixed and we shall look for the ground state energy of the Dirac field in the given $\bA$. Our second-quantized field only interact through the potential $\bA$, there is no instantaneous interaction between the fermions.

The Hamiltonian of the field  reads
\begin{equation}
\mathbb{H}^{e\bA}:=\frac12\int_{\R^3} \Big(\Psi^*(x)D_{m,e\bA} \Psi(x) - \Psi(x)D_{m,e\bA} \Psi^*(x)\Big) dx,
\label{eq:formal-Hamiltonian}
\end{equation}
where $\Psi(x)$ is the second-quantized field operator which annihilates an electron at $x$ and satisfies the anti-commutation relation
\begin{equation}
\label{CAR}
\Psi^*(x)_\sigma \Psi(y)_\nu + \Psi(y)_\nu \Psi^*(x)_\sigma = 2 \delta_{\sigma, \nu} \delta(x - y).
\end{equation}
Here $1\leq\sigma,\nu\leq4$ are the spin variables and $\Psi(x)_\sigma$ is an operator-valued distribution. The Hamiltonian $\mathbb{H}^{e\bA}$ formally acts on the fermionic Fock space $\cF=\C\oplus\bigoplus_{N\geq1}\bigwedge_1^NL^2(\R^3,\C^4)$.
The proper interpretation of the expression in parenthesis in~\eqref{eq:formal-Hamiltonian} is
\begin{multline*}
\Psi^*(x)D_{m,e\bA} \Psi(x) - \Psi(x)D_{m,e\bA} \Psi^*(x)\\:=\sum_{\mu,\nu=1}^4\Psi^*(x)_\mu\big(D_{m,e\bA}\big)_{\mu\nu}\Psi(x)_\nu-\Psi(x)_\mu\big(D_{m,e\bA}\big)_{\mu\nu}\Psi^*(x)_\nu.
\end{multline*}
This choice is made to impose charge-conjugation invariance~\cite{HaiLewSol-07},
$$\mathscr{C}\mathbb{H}^{e\bA}\mathscr{C}^{-1}=\mathbb{H}^{-e\bA}$$ 
where $\mathscr{C}$ is the charge-conjugation operator in Fock space~\cite{Thaller}.

For any fixed $\bA$, the Hamiltonian~\eqref{eq:formal-Hamiltonian} is a quadratic polynomial in the creation and annihilation operators $\Psi^*(x)$ and $\Psi(x)$. It is therefore very well understood. The expectation value in any state in Fock space can be expressed as
\begin{equation}
\pscal{\mathbb{H}^{e\bA}}=\tr D_{m,e\bA}\left(\gamma-\frac12\right)
\label{eq:HF-energy} 
\end{equation}
where $\gamma$ is the one-particle density matrix of the chosen state, namely
$$\gamma(x, y)_{\sigma,\nu} = \langle \Psi^*(x)_\sigma \Psi(y)_\nu \rangle.$$
The subtraction of half the identity comes from charge-conjugation invariance. The details of this calculation can be found in~\cite{HaiLewSol-07}. Because electrons are fermions, it is known that $\gamma$ must satisfy the constraint $0\leq\gamma\leq 1$ on $L^2(\R^3,\C^4)$, which is called the Pauli principle. Conversely, any operator $\gamma$ such that $0\leq\gamma\leq1$, arises (at least formally) from one state in Fock space. Since the energy only depends on the quantum state of the electrons through the operator $\gamma$, we can refrain from formulating our model in Fock space and only use the simpler operator $\gamma$ and the corresponding energy~\eqref{eq:HF-energy}. 

We note that the energy~\eqref{eq:HF-energy} is gauge invariant. Namely, it does not change if we replace $A$ by $A+\nabla\chi$ and $\gamma$ by $e^{ie\chi}\gamma e^{-ie\chi}$ for any function $\chi$.

\begin{remark}
A preferred state among those having $\gamma$ as one-particle density matrix is the unique associated quasi-free state~\cite{BacLieSol-94}, also called (generalized) Hartree-Fock state. So the main simplification with the model of this paper is that, when the electromagnetic field is purely classical, the ground state of the Hamiltonian is always a quasi-free state. Hartree-Fock theory is exact here, it is not an approximation. This simplification does not occur when the photon field is quantized.
\end{remark}

We are interested in finding the ground state of the vacuum, which corresponds to minimizing~\eqref{eq:HF-energy} with respect to $\gamma$. 
For atoms and molecules, we would impose a charge constraint of the form
$$\tr \Big( \gamma - \frac{1}{2} \Big) = N.$$
In this paper we restrict ourselves to the vacuum case for simplicity, and we thus do not have any other constraint than $0\leq\gamma\leq1$. The formal minimizer of the energy~\eqref{eq:HF-energy} is the negative spectral projector
$$\gamma = \1_{(- \infty, 0)} \Big( D_{m,e\bA}\Big).$$
The interpretation is that the polarized vacuum consists of particles filling all the negative energies of the Dirac operator $D_{m,e\bA}$, in accordance with the original ideas of Dirac~\cite{Dirac-33,Dirac-34a,Dirac-34b}.\footnote{For atoms and molecules, the vacuum projector $\1_{(- \infty, 0)}(\cdots)$ has to be replaced by a spectral projector of the form $\1_{(- \infty, \mu)}(\cdots)$, for some chemical potential $\mu$ which is chosen to ensure the correct number $N$ of electrons in the gap (more precisely the correct total charge of the system). Except from this change of chemical potential, the equations take exactly the same form.}

The corresponding ground state energy is
\begin{equation}
\label{eq:def_total_energy}
\min_{0\leq\gamma\leq1} \tr D_{m,e\bA}\left(\gamma-\frac12\right) = - \frac{1}{2} \tr \big| D_{m, e\bA} \big|.
\end{equation}
Of course this energy is infinite, except if our model is settled in a box with an ultraviolet cut-off~\cite{HaiLewSol-07}.

Here and everywhere in the paper, the absolute value of an operator is defined by the functional calculus
$$|C| := \sqrt{C^* C}.$$
It is in general not a scalar operator, that is, it may still depend on the spin. In the special case of $D_{m, 0}$, it does not depend on the spin, however, since it is the scalar pseudo-differential operator
$$\big| D_{m, 0} \big| = \sqrt{- \Delta + m^2}.$$

In order to give a clear mathematical meaning to~\eqref{eq:def_total_energy}, we argue as follows. First, we can subtract the (infinite) energy of the free Dirac sea and define the relative energy as
\begin{equation}
\label{eq:def_relative_energy}
\boF_{\rm rel}(e\bA):=\frac{1}{2} \tr \Big( \big| D_{m, 0} \big| - \big| D_{m, e \bA}\big| \Big).
\end{equation}
Since we have removed an (infinite) constant, we formally do not change the variational problem in which we are interested, hence we also do not change the resulting equations.

Unfortunately, the functional~\eqref{eq:def_relative_energy} is not yet well-defined, because the model is known to have ultraviolet divergences. Indeed, the operator $|D_{m, 0}| - |D_{m, e \bA}|$ is not trace-class when $A \neq 0$, which is reminiscent of the fact that the difference of the two corresponding negative projectors is never Hilbert-Schmidt~\cite{NenSch-78}. This can be seen by formally expanding the trace in a power series of $e\bA$. The first order term vanishes and the second order term is computed in Section~\ref{sec:2ndorder} below. It is infinite if no high energy cut-off is imposed.

In order to remove these divergences, an ultraviolet cut-off has to be imposed. The choice of this regularization is extremely important. When the trace of $|D_{m, 0}| - |D_{m, e \bA}|$ is expanded as a power series of $e\bA$, several terms which look divergent actually vanish because of gauge invariance. In addition to obvious physical motivations, it is necessary to keep the gauge symmetry for this reason. Some simple choices in the spirit of what we have done in the purely electrostatic case (see, e.g.~\cite{GraLewSer-09} for two different choices) would not work here, because of their lack of gauge symmetry.

In 1949, Pauli and Villars~\cite{PauVil-49} have proposed a very clever way to regularize QED, while keeping the appropriate invariances. It is this technique that we will use in this paper (but there are other choices, like the famous dimensional regularization~\cite{Leibbrandt-75}). The Pauli-Villars method consists in introducing $J$ fictitious particles in the model, with very high masses $m_1,...,m_J$ playing the role of ultraviolet cut-offs.\footnote{Since in our units $c=\hbar=1$, $m$ has the dimension of a momentum.} These particles have no physical significance and their role is only to regularize the model at high energies. Because of their large masses, they do not participate much in the low energy regime where everyday Physics takes place. 

In our language, the Pauli-Villars method consists in introducing the following energy functional
\begin{equation}
\label{eq:def_total_energy_intro}
\boxed{\boF_{\rm PV}(e\bA) := \frac{1}{2} \tr \Bigg( \sum_{j = 0}^J c_j \, \Big( \big| D_{m_j, 0} \big| - \big| D_{m_j, e \bA} \big| \Big) \Bigg),}
\end{equation}
see Remark~\ref{rmk:PV} below for some details.
Here $m_0 = m$ and $c_0 = 1$, and the coefficients $c_j$ and $m_j$ are chosen such that 
\begin{equation}
\label{cond:PV}
\sum_{j = 0}^J c_j = \sum_{j = 0}^J c_j \, m_j^2 = 0. 
\end{equation}

It is well-known in the Physics literature~\cite{PauVil-49, GreRei-08, BjorkenDrell-65} that only two auxiliary fields are necessary to fulfill these conditions, hence we shall take $J=2$ in the rest of the paper. In this case, the condition~\eqref{cond:PV} is equivalent to
\begin{equation}
\label{cond:coef}
c_1 = \frac{m_0^2 - m_2^2}{m_2^2 - m_1^2} \quad {\rm and} \quad c_2 = \frac{m_1^2 - m_0^2}{m_2^2 - m_1^2}.
\end{equation}
We will always assume that $m_0 < m_1 < m_2$, which implies that $c_1 < 0$ and $c_2 > 0$.

The role of the constraint~\eqref{cond:PV} is to remove the worst (linear) ultraviolet divergences. In the limit $m_1, m_2 \to \infty$, the regularization does not prevent a logarithmic divergence, which is best understood in terms of the averaged ultraviolet cut-off $\Lambda$ defined as
\begin{equation}
\label{def:Lambda}
\boxed{\log(\Lambda^2) := - \sum_{j = 0}^2 c_j \log (m_j^2).}
\end{equation}
The value of $\Lambda$ does not determine $m_1$ and $m_2$ uniquely. In practice, the latter are chosen as functions of $\Lambda$ such that $c_1$ and $c_2$ remain bounded when $\Lambda$ goes to infinity.

That the model is still logarithmically divergent in the averaged cut-off $\Lambda$ can again be seen by looking at the second order term in the expansion, given by \eqref{eq:formF2} and~\eqref{est:M} below. Removing this last divergence requires a \emph{renormalization} of the elementary charge $e$. This can be done following the method that we used in the purely electrostatic case with a sharp cut-off in~\cite{GraLewSer-11}, but it is not our goal in this paper. 

\begin{remark}\label{rmk:PV}
In our language the fictitious particles of the Pauli-Villars scheme are described by density matrices $\gamma_j$, with $\gamma_0 = \gamma$ and the divergent energy of the vacuum is chosen in the form
$$\sum_{j = 0}^J c_j \, \tr D_{m_j, e \bA}(\gamma_j - 1/2)$$
instead of~\eqref{eq:HF-energy}. When optimizing the energy subject to the Pauli principles $0 \leq \gamma_j \leq 1$, one has to minimize over the matrices $\gamma_j$ such that $c_j > 0$ and maximize over those such that $c_j < 0$. Adding the infinite constant $\sum_{j = 0}^J c_j \, \tr |D_{m_j, 0}|/2$ gives Formula~\eqref{eq:def_total_energy_intro}.
\end{remark}

\subsubsection{Rigorous definition}
Our main result below says that the energy functional $\bA\mapsto \boF_{\rm PV}\big(\bA\big)$ is well defined for a general four-potential $\bA = (V, A)$ in the energy space $\Hdiv$ (and which therefore satisfies the Coulomb gauge condition $\div A = 0$). 

\begin{theorem}[Proper definition of $\boF_{\rm PV}$]
\label{thm:defF}
Assume that $c_j$ and $m_j$ satisfy
\begin{equation}
\label{cond:c_i_thm}
c_0 = 1, \quad m_2 > m_1 > m_0 > 0 \quad {\rm and} \quad \sum_{j = 0}^2 c_j = \sum_{j = 0}^2 c_j m_j^2 = 0.
\end{equation}

\noindent {\it (i)} Let
\begin{equation}
\label{def:TA}
T_\bA := \frac{1}{2} \sum_{j = 0}^2 c_j \Big( \big| D_{m_j, 0} \big| - \big| D_{m_j, \bA} \big| \Big).
\end{equation}
For any $\bA \in L^1(\R^3, \R^4) \cap \Hdiv$, the operator $\tr_{\C^4} T_\bA$ is trace-class on $L^2(\R^3, \C)$. In particular, $\boF_{\rm PV}(\bA)$ is well-defined in this case, by
\begin{equation}
\label{eq:proper-def-F}
\boxed{\boF_{\rm PV}(\bA) = \tr \big( \tr_{\C^4} T_{\bA} \big).}
\end{equation}

\smallskip

\noindent {\it (ii)} The functional $\boF_{\rm PV}$ can be uniquely extended to a continuous mapping on $\Hdiv$.

\smallskip

\noindent {\it (iii)} Let $\bA \in \Hdiv$. We have
\begin{equation}
\label{eq:devF}
\boxed{\boF_{\rm PV}(\bA) = \boF_2(\bF) + \boR(\bA),}
\end{equation}
where $\bF := (E, B)$, with $E = - \nabla V$ and $B = \curl A$. The functional $\boR$ is continuous on $\Hdiv$ and satisfies
\begin{equation}
\label{est:R}
|\boR(\bA)| \leq K \Bigg( \bigg( \sum_{j = 0}^2 \frac{|c_j|}{m_j} \bigg) \big\| \bF \big\|_{L^2}^4 + \bigg( \sum_{j = 0}^2 \frac{|c_j|}{m_j^2} \bigg) \big\| \bF \big\|_{L^2}^6 \Bigg),
\end{equation}
for a universal constant $K$.

\smallskip

\noindent {\it (iv)} The functional $\boF_2$ is the non-negative and bounded quadratic form on $L^2(\R^3, \R^4)$ given by 
\begin{equation}
\label{eq:formF2}
\boF_2(\bF) = \frac{1}{8 \pi} \int_{\R^3} M(k) \Big( \big| \widehat{B}(k) \big|^2 - \big| \widehat{E}(k) \big|^2 \Big) \, dk, 
\end{equation}
where
\begin{equation}
\label{def:M}
M(k) := - \frac{2}{\pi} \sum_{j = 0}^2 c_j \int_0^1 u (1 - u) \log \big( m_j^2 + u (1 - u) |k|^2 \big) \, du.
\end{equation}
The function $M$ is positive and satisfies the uniform estimate
\begin{equation}
\label{est:M}
0 < M(k) \leq M(0) = \frac{2 \log(\Lambda)}{3 \pi},
\end{equation}
where $\Lambda$ was defined previously in~\eqref{def:Lambda}.
\end{theorem}

Let us emphasize the presence of the $\C^4$--trace in statement $(i)$ about the trace-class property of $\tr_{\C^4} T_\bA$. We do not believe that the operator is trace-class without taking first the $\C^4$--trace, except when $V \equiv 0$. If we are allowed to take more fictitious particles by increasing the numbers of auxiliary masses, it is possible to obtain a trace-class operator under the additional conditions 
$$\sum_j c_j \, m_j = \sum_j c_j \, m_j^3 = 0.$$
At least four auxiliary masses are then necessary. The terms which are not trace-class when only two fictitious particles are used, actually do \emph{not} contribute to the final value of the energy functional $\boF_{\rm PV}$ (their trace formally vanishes). For this reason, we have found more convenient to first take the $\C^4$--trace (which is enough to discard the problematic terms) and limit ourselves to two fictitious particles, as is usually done in the Physics literature. This suffices to provide a clear meaning to the energy.

The function $M$ describes the linear response of the Dirac sea. It is well-known in the Physics literature~\cite[Eq. (5.39)]{GreRei-08}. We will see below that
\begin{equation}
\label{eq:limitM}
\lim_{\Lambda \to \infty} \Big( \frac{2 \log \Lambda}{3 \pi} - M(k) \Big) = U(k) := \frac{|k|^2}{4 \pi} \int_0^1 \frac{z^2 - z^4/3}{1 + |k|^2 (1 - z^2)/4} \, dz.
\end{equation}
The function in the right-hand side of~\eqref{eq:limitM} was first computed by Serber~\cite{Serber-35} and Ueh\-ling~\cite{Uehling-35}. The same function $U$ already appeared in our previous works dealing with pure electrostatic potentials~\cite{HaiSie-03, HaiLewSer-05b, GraLewSer-11}. This is a consequence of the gauge and relativistic invariances of full QED.

Our proof of Theorem~\ref{thm:defF} in Sections~\ref{sec:atrace} and \ref{sec:estimates} below, consists in expanding the energy $\boF_{\rm PV}(\bA)$ in powers of the four-potential $\bA$. We use the resolvent expansion at a high but fixed order and therefore our main result is valid for all $\bA$'s, not only for small ones. All the odd order terms vanish (by charge-conjugation invariance). The second order term is given by the explicit formula~\eqref{eq:formF2} and it is responsible of the logarithmic ultraviolet divergences. It will be important for our existence proof that this term be strictly convex in $A$ and strictly concave in $V$. We also have to deal with the fourth order term in some detail. The latter was computed in the Physics literature in~\cite{KarNeu-50} and our task will be to estimate it. The higher order terms are then bounded in a rather crude way, following techniques of~\cite{HaiLewSer-05a}. The main difficulty in our work is to verify that the Pauli-Villars conditions~\eqref{cond:PV} induce the appropriate cancellations in the few first order terms, and to estimate them using the $L^2$--norm of the electromagnetic fields and nothing else.

In spite of its widespread use in quantum electrodynamics, the Pauli-Villars scheme~\cite{PauVil-49} has not attracted a lot of attention on the mathematical side so far (see~\cite{FraLic-64, Slavnov-73, Slavnov-74b, Slavnov-75, Finster-11} for some previous results). Theorem~\ref{thm:defF} seems to be among the first in this direction.

%%%%%%ù
\subsubsection{Differentiability}

After having properly defined the functional $\boF_{\rm PV}$, we need some of its differentiability properties, in order to be able to define the charge and current densities of the vacuum by~\eqref{eq:diff_PV_intro}. In this direction, we can prove the following

\begin{theorem}[Differentiability of $\boF_{\rm PV}$]
\label{thm:differentiability}
Assume that $c_j$ and $m_j$ satisfy conditions~\eqref{cond:c_i_thm}.

\smallskip

\noindent $(i)$ Let $\bA \in \Hdiv$ be such that $0$ is not an eigenvalue of the operators $D_{m_j, \bA}$ for $j = 0, 1, 2$. Then the functional $\boF_{\rm PV}$ is $C^\infty$ in a neighborhood of $\bA$. 

\smallskip

\noindent $(ii)$ The first derivative of $\boF_{\rm PV}$ is given by
\begin{equation}
\label{eq:diff-boF}
\langle {\rm d}\boF_{\rm PV}(\bA), (\gv,\ga) \rangle = \int_{\R^3} \big\langle (\rho_\bA,-j_\bA, ), (\gv, \ga) \big\rangle_{\R^4},
\end{equation}
for all $(\gv,\ga) \in \Hdiv$, where the density $\rho_\bA$ and the current $j_\bA$ are defined as
\begin{equation}
\label{eq:def-jA-rhoA}
\rho_\bA(x) := \big[\tr_{\C^4} Q_\bA\big](x, x) \quad {\rm and} \quad j_\bA(x) := \big[\tr_{\C^4} \bsalpha \, Q_\bA\big](x, x),
\end{equation}
and with $Q_\bA$ refering to the kernel of the operator
$$Q_\bA := \sum_{j = 0}^2 c_j\, \1_{(- \infty, 0)} \big( D_{m_j, \bA} \big).$$
The operators $\tr_{\C^4} Q_\bA$ and $\tr_{\C^4} \alpha_k Q_\bA$ for $k = 1, 2, 3$ are locally trace-class on $L^2(\R^3, \C^4)$, and $\rho_\bA$ and $j_\bA$ are well-defined functions in $L^1_{\rm loc}(\R^3) \cap \boC$, where $\boC$ is the Coulomb space
\begin{equation}
\label{eq:def_Coulomb}
\boC := \Big\{ f : \R^3 \to \C : \int_{\R^3} \frac{|\widehat{f}(k)|^2}{|k|^2} \, dk < \infty \Big\} = \dot{H}^{-1}(\R^3).
\end{equation}

\smallskip

\noindent $(iii)$ There exists a universal constant $\eta > 0$ such that the second derivative of $\boF_{\rm PV}$ satisfies the estimate
\begin{equation}
\label{eq:estim_Hessian}
\Bigg\| {\rm d}^2 \boF_{\rm PV}(\bA) - \frac{1}{4 \pi} \begin{pmatrix} - M & 0 \\ 0 & M \end{pmatrix} \Bigg\| \leq 2 K \bigg( \sum_{j = 0}^2 \frac{|c_j|}{m_j} \bigg) \| \bA \|_{\Hdiv}^2,
\end{equation}
for all $\bA$ such that $\| \bA \|_{\Hdiv} \leq \eta \sqrt{m_0}=\eta \sqrt{m}$.
\end{theorem}

Our estimate~\eqref{eq:estim_Hessian} means more precisely that
\begin{align*}
\bigg| \big\langle \bA', {\rm d}^2 \boF_{\rm PV}(\bA)\, \bA' \big\rangle & - \frac{1}{4 \pi} \int_{\R^3} M(k) \Big( \big| \widehat{B'}(k) \big|^2 - \big| \widehat{E'}(k) \big|^2 \Big) \, dk \bigg|\\
& \leq 2 K \bigg( \sum_{j = 0}^2 \frac{|c_j|}{m_j} \bigg) \big\| \bA \big\|_{\Hdiv}^2 \big\| \bA' \big\|_{\Hdiv}^2,
\end{align*}
when $\bA$ is small enough in $\Hdiv$. 

As a consequence of Lemma~\ref{lem:spectre} and Theorem~\ref{thm:differentiability}, we obtain

\begin{cor}[Regularity in a neighborhood of $0$]
\label{cor:regularity}
There exists a positive radius $\eta$ such that the functional $\boF_{\rm PV}$ is $\boC^\infty$ on the ball 
$\boB(\eta) := \big\{ \bA \in \Hdiv : \| \bA \|_{\Hdiv} < \eta \sqrt{m_0} \big\}$. On this ball, the differential ${\rm d} \boF_{\rm PV}$ is given by~\eqref{eq:diff-boF}, whereas ${\rm d}^2 \boF_{\rm PV}$ satisfies estimate~\eqref{eq:estim_Hessian}.
\end{cor}

\begin{proof}
We fix $\eta$ such that
$$C \eta < 1,$$
where $C$ is the constant in statement $(iii)$ of Lemma~\ref{lem:spectre}. For this choice, given any four-potential $\bA$ in the ball $\boB(\eta)$, $0$ is not an eigenvalue of each of the operators $D_{m_j, \bA}$. Corollary~\ref{cor:regularity} then follows from Theorem~\ref{thm:differentiability}.
\end{proof}

\subsection{Solutions to the Maxwell equations in Dirac's vacuum}

In this section, we explain how to use Theorems~\ref{thm:defF} and~\ref{thm:differentiability} in order to get the stability of the free Dirac vacuum, and to construct solutions to the nonlinear Maxwell equations. The proofs of Theorems~\ref{thm:defF} and~\ref{thm:differentiability} will be detailed later.

Let $e > 0$ be the (bare) charge of the electron. Assume that $c_0 = 1$, and that $c_j$ and $m_j$ satisfy~\eqref{cond:PV}. We work under the condition that $e \leq \bar{e}$ for some fixed constant $\bar{e}$ ($e$ is not allowed to be too large, but it can be arbitrarily small). All our constants will depend on $\bar{e}$, but not on $e$. Note that $e$ is dimensionless here because we have already set the speed of light equal to $1$. 

Using Theorem~\ref{thm:defF}, we can properly define the effective electromagnetic Lagrangian action by
\begin{multline}
\E^{\rho_{\rm ext},j_{\rm ext}}(\bA) := -\boF_{\rm PV} \big( e \bA \big) + \frac{1}{8 \pi} \int_{\R^3} |\nabla V|^2 - |\curl A|^2\\
 -e\int_{\R^3}\rho_{\rm ext} V+e\int_{\R^3}j_{\rm ext} \cdot A, 
\label{eq:effective_Lagrangian}
\end{multline}
for all $\bA\in\Hdiv$, the Coulomb-gauge homogeneous Sobolev space. Our purpose will be to construct a critical point of this Lagrangian, which will in the end solve the nonlinear equations~\eqref{eq:Maxwell-vp}.

\begin{remark}
In this paper we have considered a second-quantized Dirac field which only interacts with a classical electromagnetic field. There is another way to arrive at exactly the same Lagrangian action~\eqref{eq:effective_Lagrangian}, which is closer to our previous works~\cite{HaiLewSer-05a,HaiLewSer-05b,HaiLewSol-07,HaiLewSerSol-07,GraLewSer-09,GraLewSer-11}. We start with Coulomb-gauge QED with quantized transverse photons. Then, we restrict our attention to special states in Fock space of the form $\Omega = \Omega_{\rm HF} \otimes \Omega_{\rm Coh}$, where $\Omega_{\rm HF}$ is an electronic (generalized) Hartree-Fock state characterized by its one-particle density matrix $0 \leq \gamma \leq 1$, and $\Omega_{\rm Coh}$ is a coherent state for the photons, characterized by its magnetic potential $A(x)$ (a given classical magnetic potential in $\R^3$). We thereby get a Hartree-Fock model coupled to a classical magnetic field. Because of the instantaneous part of the Coulomb interaction, the model contains an exchange term. When this term is neglected, one obtains the exact same theory as in this paper. In relativistic density functional theory~\cite{Engel-02}, the exchange term is approximated by a local function of $\rho_{\gamma - 1/2}$ and $j_{\gamma - 1/2}$ only.  
\end{remark}

%%%%%%%%%%%%%%%%%%%%%%%%%%%%%%%%%
\subsubsection{Stability of the free Dirac vacuum}
\label{sub:free}
%%%%%%%%%%%%%%%%%%%%%%%%%%%%%%%%%%

In the vaccum case $\rho_{\rm ext}=j_{\rm ext}=0$, we have the obvious solution $\bA=0$. The following theorem says that $0$ is stable in the sense that it is a saddle point of the effective Lagrangian action, with the same Morse index as for the classical Maxwell Lagrangian action. This can be interpreted as the stability of the free vacuum under its own electromagnetic excitations.

\begin{theorem}[Stability of the free Dirac vacuum]
\label{thm:free}
Assume that $c_j$ and $m_j$ satisfy~\eqref{cond:c_i_thm}.
The four-potential $\bA \equiv 0$ is a saddle point of $\E^0$. It is the unique solution to the min-max problem
\begin{equation}
\label{eq:min_max_free} 
\E^0(0, 0) = \min_{\| \nabla V \|_{L^2} < \frac{r \sqrt{m_0}}{e}} \, \E^0(V, 0) = \max_{\| \curl A \|_{L^2} < \frac{r \sqrt{m_0}}{e}} \, \E^0(0, A),
\end{equation}
or, equivalently,
\begin{equation}
\label{eq:minmax_is_maxmin_free}
\begin{split}
\E^0(0, 0) & = \min_{\| \nabla V \|_{L^2} < \frac{r \sqrt{m_0}}{e}} \quad \sup_{\| \curl A \|_{L^2} < \frac{r \sqrt{m_0}}{e}} \, \E^0(V, A)\\
& = \max_{\| \curl A \|_{L^2} < \frac{r \sqrt{m_0}}{e}} \quad \inf_{\| \nabla V \|_{L^2} < \frac{r \sqrt{m_0}}{e}} \, \E^0(V, A), 
\end{split}
\end{equation}
for some positive radius $r$ which only depends on $\sum_{j = 0}^2 |c_j| (m_0/m_j)$ and $\bar{e}$ (the largest possible value of $e$).
\end{theorem}

The result is a direct consequence of the properties of the functional $\boF_2$ defined in~\eqref{eq:formF2}, as well as of the regularity properties of $\boF_{\rm PV}$.

As we have seen we can let the cut-off $\Lambda$ go to $\infty$ (which implies that $m_0/m_j \to 0$ for $j = 1, 2$), while keeping $c_1$ and $c_2$ bounded. We therefore see that the radius $r$ of the ball of stability of the free vacuum does not go to zero in the limit $\Lambda\to\ii$ if the bare parameters $e$ and $m_0$ are kept bounded.

For $A\equiv0$, the electrostatic stability of the free Dirac vacuum was pointed out first by Chaix, Iracane and Lions~\cite{ChaIra-89, ChaIraLio-89} and proved later in full generality in~\cite{BacBarHelSie-99, HaiLewSer-05a, HaiLewSer-05b}. It is possible to include the exchange term and even establish the \emph{global} stability of the free Dirac vacuum~\cite{HaiLewSer-05a, HaiLewSer-05b, HaiLewSerSol-07}. Dealing with magnetic fields is more complicated and, so far, we are only able to prove \emph{local} stability, using the Pauli-Villars regularization. Because of lack of gauge symmetry, it is not clear whether the free Dirac sea is still stable under magnetic excitations when a sharp ultraviolet cut-off is used.

\begin{proof}
We choose $r > 0$ such that 
\begin{equation}
\label{eq:cond_r_free}
r \leq \eta/\sqrt{2} \quad {\rm and} \quad 2 K \bigg( \sum_{j = 0}^2 \frac{|c_j|}{m_j} \bigg) m_0 \big( r^2 + 2 m_0 r^4 \big) \leq \frac{1}{8 \pi\bar{e}^2},
\end{equation}
where $K$ is the constant appearing in~\eqref{est:R}, and where $\eta$ is the constant in statement $(iii)$ of Theorem~\ref{thm:differentiability}. We recall that $e \leq \bar{e}$. Consider now any $\bA$ such that $\| \nabla V \|_{L^2} \leq r \sqrt{m_0}/e$ and $\| \curl A \|_{L^2} \leq r \sqrt{m_0}/e$ (which implies $\| \bA \|_{\Hdiv} = \| \bF \|_{L^2} \leq \sqrt{2 m_0} r/e$). By~\eqref{est:R}, we have
\begin{equation}
\label{eq:reduce_F_2}
\begin{split}
\big| \boF_{\rm PV}(e \bA) - \boF_2(e \bF) \big| & \leq K \Bigg( \bigg( \sum_{j = 0}^2 \frac{|c_j|}{m_j} \bigg) e^4 \| \bF \|_{L^2}^4 + \bigg( \sum_{j = 0}^2 \frac{|c_j|}{m_j^2} \bigg) e^6 \| \bF \|_{L^2}^6 \Bigg)\\
& \leq 2 K \bigg( \sum_{j = 0}^2 \frac{|c_j|}{m_j} \bigg) m_0 \big( r^2 + 2 m_0r^4 \big) e^2 \| \bF \|_{L^2}^2\\
& \leq \frac{1}{8 \pi} \big\| \bF \big\|_{L^2}^2.
\end{split}
\end{equation}
Using Formula~\eqref{eq:formF2} for $\boF_2$, we get 
$$\E^0( V, 0) \geq \frac{e^2}{8 \pi} \int_{\R^3} M(k) |\widehat{E}(k)|^2 \, dk \geq 0,$$
with equality if and only $V \equiv 0$, since $M>0$. Similarly,
$$\E^0(0, A) \leq -\frac{e^2}{8 \pi} \int_{\R^3} M(k) |\widehat{B}(k)|^2 \, dk \leq 0,$$
with equality if and only $A\equiv0$. Thus we have shown~\eqref{eq:min_max_free}. The equivalence between~\eqref{eq:min_max_free} and~\eqref{eq:minmax_is_maxmin_free} is a classical fact of convex analysis, see~\cite[Prop. 1.2, Chap. VI]{EkelandTemam}. 

Finally, since we can deduce from~\eqref{eq:estim_Hessian} that
\begin{align*}
\bigg\| {\rm d}^2 \E^0(\bA) - \frac{1}{4 \pi} \begin{pmatrix} 1 + e^2 M & 0 \\ 0 & -1 - e^2 M \end{pmatrix} \bigg\| & \leq 2 K e^2 \bigg( \sum_{j = 0}^2 \frac{|c_j|}{m_j} \bigg) \| \bF \|_{L^2}^2\\
& \leq 4 K m_0 \bigg( \sum_{j = 0}^2 \frac{|c_j|}{m_j} \bigg) r^2,
\end{align*}
for $e \| \bA \|_{\Hdiv} \leq r \sqrt{2 m_0} \leq \eta \sqrt{m_0}$, we deduce that $\E^0$ is strictly convex with respect to $V$ and strictly concave with respect to $A$, provided that $r$ satisfies the additional condition 
\begin{equation}
4 K m_0 \bigg( \sum_{j = 0}^2 \frac{|c_j|}{m_j} \bigg) r^2 < \frac{1}{4 \pi}.
\label{eq:add_cond_r} 
\end{equation}
This implies uniqueness of the saddle point by~\cite[Prop 1.5, Chap. VI]{EkelandTemam}.
\end{proof}

%%%%%%%%%%%%%%%%%%%%%%%%%%%%%%%%%%%%%%%%%%%%%%%%%%%%%%%%%%%%%%%%%%%%%%
\subsubsection{Solution with external sources}
\label{sub:external}
%%%%%%%%%%%%%%%%%%%%%%%%%%%%%%%%%%%%%%%%%%%%%%%%%%%%%%%%%%%%%%%%%%%%%%

We now come back to our initial problem and consider an external density $\rho_{\rm ext}$ and an external current $j_{\rm ext}$. It will be convenient to express our result in terms of the size of
$$V_{\rm ext}:=e\,\rho_{\rm ext}\ast\frac{1}{|x|}\quad\text{and}\quad A_{\rm ext}:=e\,j_{\rm ext}\ast\frac{1}{|x|},$$
which are the corresponding potentials when the vacuum does not react. We look for a saddle point of the Lagrangian action $\mathscr{L}^{\rho_{\rm ext},j_{\rm ext}}_{\rm PV}$ defined above in~\eqref{eq:effective_Lagrangian}.

\begin{theorem}[Nonlinear Maxwell equations in small external sources]
\label{thm:polarized}
Assume that $c_j$ and $m_j$ satisfy~\eqref{cond:c_i_thm}. Let $r$ be the same constant as in Theorem~\ref{thm:free}.

\smallskip

\noindent {\it (i)} For any 
\begin{equation}
\label{eq:estim_A_ext}
e \| \bA_{\rm ext} \|_{\Hdiv} < \frac{r \sqrt{m_0}}{9},
\end{equation}
there exists a unique solution $\bA_* = (V_*, A_*) \in \Hdiv$ to the min-max problem
\begin{equation}
\begin{split}
\mathscr{L}^{\rho_{\rm ext},j_{\rm ext}}_{\rm PV}(\bA_*) &= \min_{\| \nabla V \|_{L^2} < \frac{r \sqrt{m_0}}{3 e}} \, \mathscr{L}^{\rho_{\rm ext},j_{\rm ext}}_{\rm PV}(V, A_*) \\
&= \max_{\| \curl A \|_{L^2} < \frac{r \sqrt{m_0}}{3 e}} \, \mathscr{L}^{\rho_{\rm ext},j_{\rm ext}}_{\rm PV}(V_*, A), 
\end{split}
\end{equation}
or, equivalently, to 
\begin{equation}
\label{eq:maxmin_minmax}
\begin{split}
\mathscr{L}^{\rho_{\rm ext},j_{\rm ext}}_{\rm PV}(\bA_*) & = \min_{\| \nabla V \|_{L^2} < \frac{r \sqrt{m_0}}{3e}} \quad \sup_{\| \curl A \|_{L^2} < \frac{r \sqrt{m_0}}{3 e}} \, \mathscr{L}^{\rho_{\rm ext},j_{\rm ext}}_{\rm PV}(\bA)\\
& = \max_{\| \curl A \|_{L^2} < \frac{r \sqrt{m_0}}{3 e}} \quad \inf_{\| \nabla V \|_{L^2} < \frac{r \sqrt{m_0}}{3 e}} \, \mathscr{L}^{\rho_{\rm ext},j_{\rm ext}}_{\rm PV}(\bA).
\end{split}
\end{equation}

\smallskip

\noindent {\it (ii)} The four-potential $\bA_*$ is a solution to the nonlinear equations
\begin{equation}
\label{eq:SCF}
\Bigg\{ \begin{array}{ll} - \Delta V_* = 4 \pi e\,\big(\rho_{e\bA_*}+\rho_{\rm ext}\big),\\
- \Delta A_* = 4 \pi e\, \big(j_{e\bA_*}+j_{\rm ext}\big), \end{array}
\end{equation}
where $\rho_{e\bA_*}$ and $j_{e\bA_*}$ refer to the charge and current densities of the Pauli-Villars-regulated vacuum 
\begin{equation}
\label{eq:SCF_op}
Q_*=\sum_{j = 0}^2 c_j \, \1_{(- \infty, 0)} \big( D_{m_j, e \bA_*} \big),
\end{equation}
defined in Theorem~\ref{thm:differentiability}.
\end{theorem}

Equations~\eqref{eq:SCF} and~\eqref{eq:SCF_op} are well known in the Physics literature (see, e.g., \cite[Eq. (62)--(64)]{Engel-02}). 
Solutions have been rigorously constructed in the previous works~\cite{HaiLewSer-05a, HaiLewSer-05b, HaiLewSol-07}, with a sharp ultraviolet cut-off, but in the purely electrostatic case $A_{\rm ext}=A_*=0$. In this special case it is possible to obtain the polarized vacuum as a \emph{global} minimizer. The method of~\cite{HaiLewSer-05a, HaiLewSer-05b, HaiLewSol-07} does not seem to be applicable with magnetic fields, however. To our knowledge, Theorem~\ref{thm:polarized} is the first result dealing with optimized magnetic fields in interaction with Dirac's vacuum.

The proof of Theorem~\ref{thm:polarized} is based on tools of convex analysis, using that $\mathscr{L}^{\rho_{\rm ext},j_{\rm ext}}_{\rm PV}$ has the local saddle point geometry by Theorem~\ref{thm:defF}.

\begin{proof}
Let us define the balls
$$\boB_{\rm V}(r) := \big\{ V \in L^6(\R^3, \R) \ : \ e \| \nabla V \|_{L^2} \leq r \sqrt{m_0} \big\},$$
and
$$\boB_{\rm A}(r) := \big\{ A \in L^6(\R^3, \R^3) \ : \ e \| \curl A \|_{L^2} \leq r \sqrt{m_0} \big\}.$$
As we have already shown in the proof of Theorem~\ref{thm:free}, when $r$ satisfies condition~\eqref{eq:cond_r_free}, the function $\bA \mapsto \E^0(\bA)$ is strictly convex with respect to $V$ and strictly concave with respect to $A$ on $\boB_{\rm V}(r) \times \boB_{\rm A}(r)$. We deduce that 
$$\bA\mapsto \E^0(\bA)+e\int_{\R^3}j_{\rm ext}\cdot A-\rho_{\rm ext}V$$
satisfies the same property.

We now assume that the external field $\bA_{\rm ext} \in \boB_{\rm V}(\varepsilon r) \times \boB_{\rm A}(\varepsilon r)$ for some $\varepsilon\leq 1/3$ to be chosen later, and we look for a saddle point in $\boB_{\rm V}(3\varepsilon r) \times \boB_{\rm A}(3\varepsilon r)$. Since $\E^{\rho_{\rm ext},j_{\rm ext}}$ is strongly continuous on $\boB_{\rm V}(3\varepsilon r) \times \boB_{\rm A}(3\varepsilon r)$ by Theorem~\ref{thm:defF}, a classical result from convex analysis implies that $\E^{\rho_{\rm ext},j_{\rm ext}}$ possesses at least one saddle point $\bA_*=(V_*, A_*) \in \boB_{\rm V}(3\varepsilon r) \times \boB_{\rm A}(3\varepsilon r)$, solving
$$\E^{\rho_{\rm ext},j_{\rm ext}}(\bA_*) = \min_{V \in \boB_V(3\varepsilon r)} \, \E^{\rho_{\rm ext},j_{\rm ext}}(V, A_*) = \max_{A \in \boB_A(3\varepsilon r)} \, \E^{\rho_{\rm ext},j_{\rm ext}}(V_*, A).$$
See for instance~\cite[Prop. 2.1, Chap. VI]{EkelandTemam}. Uniqueness follows from the strict concavity and convexity, by~\cite[Prop 1.5, Chap. VI]{EkelandTemam}.

It only remains to verify that $\bA_*$ does not lie on the boundary of $\boB_{\rm V}(3\varepsilon r) \times \boB_{\rm A}(3\varepsilon r)$. Similarly as in~\eqref{eq:reduce_F_2}, we first compute
\begin{equation}
\label{eq:reduce_F_2_ext}
\begin{split}
\big| \boF_{\rm PV}(e \bA') - \boF_2(e \bF') \big| & \leq K \bigg( \sum_{j = 0}^2 \frac{|c_j|}{m_j} \bigg) \bigg( e^4 \| \bF' \|_{L^2}^4 + \frac{e^6}{m_0} \| \bF' \|_{L^2}^6 \bigg)\\
& \leq \frac{(3\varepsilon)^2}{8 \pi} \| \bF' \|_{L^2}^2,
\end{split}
\end{equation}
for all $\bA' \in \boB_{\rm V}(3\varepsilon r) \times \boB_{\rm A}(3\varepsilon r)$, when $r$ satisfies~\eqref{eq:cond_r_free} and $\varepsilon \leq 1/3$. We obtain with $B_{\rm ext}:=\curl A_{\rm ext}$
\begin{align*}
& \E^{\rho_{\rm ext},j_{\rm ext}}(V, A) - \E^{\rho_{\rm ext},j_{\rm ext}}(V, 0)\\
& \ \leq -e^2 \boF_2(0, B) - \frac{1}{8 \pi} \int_{\R^3} B^2+\frac{1}{4\pi}\int_{\R^3}B_{\rm ext}\cdot B + \frac{(3\varepsilon)^2}{8\pi}\int_{\R^3}2E^2+B^2.
\end{align*}
When $A$ belongs to the boundary of $\boB_A(3\varepsilon r)$, we obtain
$$\norm{B_{\rm ext}}_{L^2}\leq \frac13 \norm{B}_{L^2}\quad\text{and}\quad \norm{E}_{L^2}\leq \norm{B}_{L^2}$$
and therefore
\begin{equation*}
\E^{\rho_{\rm ext},j_{\rm ext}}(V, A) - \E^{\rho_{\rm ext},j_{\rm ext}}(V, 0)
\ \leq -e^2 \boF_2(0, B) - \frac{1-2/3-3(3\varepsilon)^2}{8 \pi} \int_{\R^3} B^2. 
\end{equation*}
Choosing $\varepsilon = 1/9$, the right-hand side is $\leq-e^2 \boF_2(0, B)<0$ since $B\neq0$. So we have shown that when $A$ belongs to the boundary of $\boB_{\rm V}(3\varepsilon r)$, $\E^{\rho_{\rm ext},j_{\rm ext}}(V, A) < \E^{\rho_{\rm ext},j_{\rm ext}}(V,0)$. Since $A_*$ maximizes $A \mapsto \E^{\rho_{\rm ext},j_{\rm ext}}(V_*, A)$, it cannot have an energy smaller than $\E^{\rho_{\rm ext},j_{\rm ext}}(V_*,0)$ and we conclude that 
$$e^2 \int_{\R^3} |\curl A_*|^2 < (3\varepsilon)^2 r^2 m_0.$$
Similarly, we can show that $V_*$ does not belong to the boundary of $\boB_{\rm V}(3\varepsilon r)$.

The unique saddle point $\bA_*=(V_*,A_*)$ being in the interior of the set $\boB_{\rm V}(3 \varepsilon r) \times \boB_{\rm A}(3 \varepsilon r)$, the derivative of $\E^{\rho_{\rm ext},j_{\rm ext}}$ must vanish at this point. Using the value~\eqref{eq:diff-boF} of the derivative of $\boF_{\rm PV}$ computed in Theorem~\ref{thm:differentiability}, we find that 
$$\begin{cases}
-\Delta V_*=4\pi\,e\, \big(\rho_{e\bA_*}+\rho_{\rm ext}\big)\\
-\Delta A_*=4\pi\,e\,\big(j_{e\bA_*}+j_{\rm ext}\big)
  \end{cases}
$$
where $\rho_{e\bA_*}$ and $j_{e\bA_*}$ are given by~\eqref{eq:def-jA-rhoA} in Theorem~\ref{thm:differentiability}. This is exactly the self-consistent equation~\eqref{eq:SCF}.
\end{proof}

The rest of the paper is devoted to the proofs of Theorems~\ref{thm:defF} and~\ref{thm:differentiability}. Our strategy is as follows. First, in Section~\ref{sec:atrace}, we show that the functional $\cF_{\rm PV}$ is well-defined for four-potentials $\bA$ with an appropriate decay in $x$-space (the integrability of $\bA$ on $\R^3$ is enough). Then, we compute things more precisely in Section~\ref{sec:estimates}, and we exhibit the cancellations which show that this functional can be uniquely extended by continuity to $\Hdiv$.

%%%%%%%%%%%%%%%%%%%%%%%%%%%%%%%%%%%%%%%%%%%%%%%%%%%%%%%%%%%%%%%%%%%%%%%%%%%%%%%%%%%%%%
%%%%%%%%%%%%%%%%%%%%%%%%%%%%%%%%%%%%%%%%%%%%%%%%%%%%%%%%%%%%%%%%%%%%%%%%%%%%%%%%%%%%%%
%%%%%%%%%%%%%%%%%%%%%%%%%%%%%%%%%%%%%%%%%%%%%%%%%%%%%%%%%%%%%%%%%%%%%%%%%%%%%%%%%%%%%%
\section{The Pauli-Villars functional for integrable potentials}
\label{sec:atrace}
%%%%%%%%%%%%%%%%%%%%%%%%%%%%%%%%%%%%%%%%%%%%%%%%%%%%%%%%%%%%%%%%%%%%%%%%%%%%%%%%%%%%%%
%%%%%%%%%%%%%%%%%%%%%%%%%%%%%%%%%%%%%%%%%%%%%%%%%%%%%%%%%%%%%%%%%%%%%%%%%%%%%%%%%%%%%%
%%%%%%%%%%%%%%%%%%%%%%%%%%%%%%%%%%%%%%%%%%%%%%%%%%%%%%%%%%%%%%%%%%%%%%%%%%%%%%%%%%%%%%

The purpose of this section is to prove that the operator
\begin{equation}
\label{def:trTA}
\tr_{\C^4} T_\bA := \frac{1}{2} \sum_{j = 0}^2 c_j \tr_{\C^4} \big( |D_{m_j, 0}| - |D_{m_j, \bA}| \big),
\end{equation}
is trace-class, when the four-potential $\bA := (V, A)$ decays sufficiently fast. The proof relies on an expansion of $\boF_{\rm PV}(\bA)$ with respect to the four-potential $\bA$ using the resolvent formula, but for which we actually do \emph{not} need that $\bA$ is small. Our precise statement is the following

\begin{prop}[$\tr_{\C^4} T_\bA$ is in $\gS_1$]
\label{prop:trace-class}
Assume that $c_j$ and $m_j$ satisfy conditions~\eqref{cond:c_i_thm}. Then, the operator $\tr_{\C^4} T_{\bA}$ is trace-class whenever $\bA \in L^1(\R^3, \R^4) \cap H^1(\R^3, \R^4)$.
\end{prop}

\begin{remark}
For this result, it is not important that $\div A = 0$, hence we do not require that $A \in \Hdiv$.
\end{remark}

The rest of this section is devoted to the proof of Proposition~\ref{prop:trace-class}.

\begin{proof}
Our starting point is the integral formula
\begin{equation}
\label{eq:absval}
|x| = \frac{1}{\pi} \int_{\R} \frac{x^2}{x^2 + \omega^2} \, d\omega = \frac{1}{2 \pi} \int_\R \Big( 2 - \frac{i \omega}{x + i \omega} + \frac{i \omega}{x - i \omega} \Big) \, d\omega.
\end{equation}
When $T$ is a self-adjoint operator on $L^2(\R^3, \R^4)$, with domain $D(T)$, it follows from~\eqref{eq:absval} using standard functional calculus (see e.g.~\cite{ReeSim1}), that the absolute value $|T|$ of $T$ is given by
\begin{equation}
\label{eq:intvalabs}
|T| = \frac{1}{2 \pi} \int_\R \Big( 2 - \frac{i \omega}{T + i \omega} + \frac{i \omega}{T - i \omega} \Big) \, d\omega.
\end{equation}
Let us remark that this integral is convergent when seen as an operator from $D(T^2)$ to the ambient Hilbert space. In particular, 
$$\bigg\| \frac{T^2}{T^2 + \omega^2} \bigg\|_{D(T^2)\to L^2(\R^3, \C^4)} \leq \min \Big\{ 1, \omega^{- 2} \big\| T^2 \big\|_{D(T^2)\to L^2(\R^3, \C^4)} \Big\}.$$

Since the domains of $D_{m_j, 0}^2$ and $D_{m_j, \bA}^2$ are both equal to $H^2(\R^3, \C^4)$, we deduce that we can write
\begin{equation}
\label{eq:develT}
\begin{split}
T_\bA = \frac{1}{4 \pi} \int_\R \sum_{j = 0}^2 c_j \, \Big( \frac{i \omega}{D_{m_j, \bA} + i \omega} & - \frac{i \omega}{D_{m_j, \bA} - i \omega}\\
& - \frac{i \omega}{D_{m_j, 0} + i \omega} + \frac{i \omega}{D_{m_j, 0} - i \omega} \Big) \, d\omega
\end{split}
\end{equation}
on $H^2(\R^3, \C^4)$. Here and everywhere else it is not a problem if $D_{m_j,\bA}$ has $0$ as an eigenvalue. The operator $D_{m_j, \bA} + i \omega$ is invertible for $\omega \neq 0$, and $(i \omega)(D_{m_j, \bA} + i \omega)^{-1}$ stays uniformly bounded in the limit $\omega \to 0$.

In order to establish Proposition~\ref{prop:trace-class}, we will prove that the $\C^4$--trace of the integral in the right-hand side of~\eqref{eq:develT} defines a trace-class operator according to the inequality
\begin{multline}
\label{int:omega}
\int_\R \bigg\| \sum_{j = 0}^2 c_j \, \tr_{\C^4} \Big( \frac{i \omega}{D_{m_j, \bA} + i \omega}  - \frac{i \omega}{D_{m_j, \bA} -i \omega}\\
 - \frac{i \omega}{D_{m_j, 0} + i \omega} + \frac{i \omega}{D_{m_j, 0} - i \omega} \Big) \bigg\|_{\gS_1} \, d\omega < \infty,
\end{multline}
which we can establish when $\bA \in L^1(\R^3, \R^4) \cap H^1(\R^3, \R^4)$. This will complete the proof of Proposition~\ref{prop:trace-class}.

As a consequence, our task reduces to derive estimates in Schatten spaces on the integrand operator
\begin{align*}
\boR(\omega, \bA) := \sum_{j = 0}^2 c_j \, \tr_{\C^4} \Big( \frac{i \omega}{D_{m_j, \bA} + i \omega} - & \frac{i \omega}{D_{m_j, \bA} - i \omega}\\
- & \frac{i \omega}{D_{m_j, 0} + i \omega} + \frac{i \omega}{D_{m_j, 0} - i \omega} \Big),
\end{align*}
which we can integrate with respect to $\omega$. To this end, we use the resolvent expansion 
$$\frac{i\omega}{D_{m_j, \bA} + i \omega} - \frac{i\omega}{D_{m_j, 0} + i \omega} = \frac{i\omega}{D_{m_j, \bA} + i \omega} \big( \bsalpha \cdot A - V \big) \, \frac{1}{D_{m_j, 0} + i \omega},$$
iterated six times and obtain
\begin{equation}
\label{eq:resolvent_expansion}
\begin{split}
\frac{i \omega}{D_{m_j, \bA} + i \omega} - \frac{i \omega}{D_{m_j, 0} + i \omega} = & \sum_{n = 1}^5 \frac{i \omega}{D_{m_j, 0} + i \omega} \Big( \big( \bsalpha \cdot A - V \big) \, \frac{1}{D_{m_j, 0} + i \omega} \Big)^n\\
& + \frac{i \omega}{D_{m_j, \bA} + i \omega} \Big( \big( \bsalpha \cdot A - V \big) \, \frac{1}{D_{m_j, 0} + i \omega} \Big)^6,
\end{split}
\end{equation}
together with the similar expression for the term with $- i \omega$ instead of $+ i \omega$. Again, we insist on the fact that this expansion makes perfect sense for $\omega \neq 0$, even if the spectrum of $D_{m_j, \bA}$ contains $0$. This allows us to write
\begin{equation}
\label{eq:def_orders}
\begin{split}
\boR(\omega, \bA) = \sum_{n = 1}^5 \tr_{\C^4} \Big( R_n(\omega, \bA) & + R_n(- \omega, \bA) \Big)\\
& + \tr_{\C^4} \big( R_6'(\omega, \bA) + R_6'(- \omega, \bA) \big),
\end{split}
\end{equation}
with
\begin{equation}
R_n(\omega, \bA) := \sum_{j = 0}^2 c_j \, \frac{i \omega}{D_{m_j, 0} + i \omega} \Big( \big( \bsalpha \cdot A - V \big) \frac{1}{D_{m_j, 0} + i \omega} \Big)^n,
\label{eq:def_R_n} 
\end{equation}
and
\begin{equation}
R_6'(\omega, \bA) := \sum_{j = 0}^2 c_j \, \frac{i \omega}{D_{m_j, \bA} + i \omega} \Big( \big( \bsalpha \cdot A - V \big) \frac{1}{D_{m_j, 0} + i \omega} \Big)^6.
\label{eq:def_R_6} 
\end{equation}
Our purpose is to prove that
\begin{multline}
\label{eq:thekey}
\int_\R \bigg( \sum_{n = 1}^5 \big\| \tr_{\C^4} \big( R_n(\omega, \bA) + R_n(- \omega, \bA) \big) \big\|_{\gS_1}\\
 + \big\| \tr_{\C^4} \big( R_6'(\omega, \bA) + R_6'(- \omega, \bA) \big) \big\|_{\gS_1} \bigg) \, d\omega < \infty.
\end{multline}

%%%%%%%%%%%%%%%%%%%%%%%%%%%%%%%%%%%%%%%%%%%%%%
\addtocontents{toc}{\SkipTocEntry}
\subsection*{Estimate on the sixth order term}
%%%%%%%%%%%%%%%%%%%%%%%%%%%%%%%%%%%%%%%%%%%%%%

We first estimate the sixth order term $R_6'(\omega, \bA)$ in~\eqref{eq:def_orders} which is the simplest one. The $\C^4$--trace is not going to be helpful for us here. First we use the inequality
$$\Big\| \frac{i \omega}{D_{m_j, \bA} + i \omega} \Big\| \leq 1,$$
which, in particular, takes care of the possibility of having $0$ in the spectrum of $D_{m_j,\bA}$.
Combining with H\"older's inequality in Schatten spaces, $\norm{AB}_{\gS_1}\leq \norm{A}_{\gS_p}\norm{B}_{\gS_{p'}}$, we obtain
\begin{equation}
\label{eq:Holder_6th}
\begin{split}
\Big\| \frac{i \omega}{D_{m_j, \bA} + i \omega} \Big( \big( \bsalpha \cdot A - V \big) & \frac{1}{D_{m_j, 0} + i \omega} \Big)^6 \Big\|_{\gS_1}\\
\leq & \Big\| \big( \bsalpha \cdot A - V \big) \frac{1}{D_{m_j, 0} + i \omega} \Big\|_{\gS_6}^6.
\end{split}
\end{equation}
We next use the Kato-Seiler-Simon inequality~\eqref{eq:KSS}, similarly as in the proof of Lemma~\ref{lem:spectre}, which gives us
$$\forall p > 3, \ \Big\| \big( \bsalpha \cdot A - V \big) \frac{1}{D_{m_j, 0} + i \omega} \Big\|_{\gS_p} \leq (I_p)^\frac{1}{p} (m_j^2 + \omega^2)^{\frac{3}{2 p} - \frac{1}{2}} \big\| \bA \big\|_{L^p},$$
where
$$I_p := \frac{1}{2 \pi^2} \int_0^\infty \frac{r^2 \, dr}{(1 + r^2)^{\tfrac{p}{2}}}.$$
For $p = 6$, we can use the Sobolev inequalities
\begin{equation}
\label{eq:Sobolev}
\| V \|_{L^6} \leq S \| \nabla V \|_{L^2} \quad {\rm and} \quad \| A \|_{L^6} \leq S \| \nabla A \|_{L^2},
\end{equation}
to obtain an estimate in terms of the gradient $\nabla \bA$ by
$$\Big\| \big( \bsalpha \cdot A - V \big) \frac{1}{D_{m_j, 0} + i \omega} \Big\|_{\gS_6} \leq \frac{(I_6)^\frac{1}{6}S}{(m_j^2 + \omega^2)^\frac{1}{4}} \big\| \nabla \bA \big\|_{L^2}.$$
Inserting in~\eqref{eq:Holder_6th}, we have
\begin{equation}
\label{eq:R6atrace}
\big\| R_6'(\omega, \bA) \big\|_{\gS_1} \leq \sum_{j = 0}^2 |c_j| \, \frac{S^6 I_6}{(m_j^2 + \omega^2)^\frac{3}{2}} \, \big\| \nabla \bA \big\|_{L^2}^6,
\end{equation}
so that
\begin{equation}
\label{eq:6th-final}
\int_\R \| R_6'(\omega, \bA) \|_{\gS_1} \, d\omega \leq S^6 I_6 \bigg( \sum_{j = 0}^2 \frac{|c_j|}{m_j^2} \bigg) \| \nabla \bA \|_{L^2}^6 \int_\R \frac{d\omega}{(1 + \omega^2)^\frac{3}{2}}.
\end{equation}
The term with $+ i \omega$ replaced by $- i \omega$ is treated similarly. 

%%%%%%%%%%%%%%%%%%%%%%%%%%%%%%%%%%%%%%%%%%%%%%
\addtocontents{toc}{\SkipTocEntry}
\subsection*{Estimate on the fifth order term}
%%%%%%%%%%%%%%%%%%%%%%%%%%%%%%%%%%%%%%%%%%%%%%

The method that we have used for the sixth order term of~\eqref{eq:resolvent_expansion} can be applied in a similar fashion to the fifth order term, leading to the estimate
\begin{equation}
\label{eq:5th-final}
\int_\R \| R_5(\pm \omega, \bA) \|_{\gS_1} d\omega \leq I_5 \bigg( \sum_{j = 0}^2 \frac{|c_j|}{m_j} \bigg) \| \bA \|_{L^5}^5 \int_\R \frac{|\omega| \, d\omega}{(1 + \omega^2)^\frac{3}{2}}.
\end{equation}
None of these estimates use simplifications coming from the $\C^4$--trace. The latter is only useful for lower order terms.

%%%%%%%%%%%%%%%%%%%%%%%%%%%%%%%%%%%%%%%%%%%%%%%
\addtocontents{toc}{\SkipTocEntry}
\subsection*{Estimate on the fourth order term}
%%%%%%%%%%%%%%%%%%%%%%%%%%%%%%%%%%%%%%%%%%%%%%%
 
For the other terms in~\eqref{eq:resolvent_expansion}, we need more precise estimates based on conditions~\eqref{cond:coef} satisfied by the coefficients $c_j$ and the masses $m_j$. We start by considering the fourth order term, for which we use the identity $c_0 + c_1 + c_2 = 0$ to write
\begin{align*}
R_4(\omega, \bA) = & \sum_{j = 0}^2 c_j \sum_{k = 0}^4 \Big( \frac{1}{D_{m_0, 0} + i \omega} \big( \bsalpha \cdot A - V \big) \Big)^k \times\\
& \times \Big( \frac{i \omega}{D_{m_j, 0} + i \omega} - \frac{i \omega}{D_{m_0, 0} + i \omega} \Big) \Big( \big( \bsalpha \cdot A - V \big) \frac{1}{D_{m_j, 0} + i \omega} \Big)^{4 - k}.
\end{align*}
Next we use that
\begin{equation}
\label{eq:norm-diff-Dj}
\begin{split}
\Big\| \frac{i \omega}{D_{m_j, 0} + i \omega} - \frac{i \omega}{D_{m_0, 0} + i \omega} \Big\| = & \Big\| \frac{m_0 - m_j}{D_{m_j, 0} + i \omega} \bsbeta \frac{i \omega}{D_{m_0, 0} + i \omega} \Big\|\\
\leq & \big( m_j - m_0 \big) \frac{|\omega|}{m_0^2 + \omega^2},
\end{split}
\end{equation}
since $m_j \geq m_0$, and we argue as before, using this time $\bA \in L^4(\R^3, \R^4)$. We obtain
\begin{equation}
\label{eq:R4atrace}
\big\| R_4(\pm \omega, \bA) \big\|_{\gS_1} \leq \frac{5 I_4 |\omega|}{(m_0^2 + \omega^2)^{\tfrac{3}{2}}} \sum_{j = 0}^2 |c_j| \, \big( m_j - m_0 \big) \, \| \bA \|_{L^4}^4,
\end{equation}
hence
\begin{equation}
\label{eq:4th-final}
\int_\R \big\| R_4(\pm \omega, \bA) \big\|_{\gS_1} \, d\omega \leq 5 I_4 \| \bA \|_{L^4}^4 \sum_{j = 0}^2 |c_j| \, \frac{m_j - m_0}{m_0} \int_\R \frac{|\omega| \, d\omega}{(1 + \omega^2)^{\tfrac{3}{2}}}.
\end{equation}
Notice again that we have not used the $\C^4$--trace in our estimate of the fourth order term.

%%%%%%%%%%%%%%%%%%%%%%%%%%%%%%%%%%%%%%%%%%%%%%
\addtocontents{toc}{\SkipTocEntry}
\subsection*{Estimate on the first order term}
%%%%%%%%%%%%%%%%%%%%%%%%%%%%%%%%%%%%%%%%%%%%%%

In order to deal with the first, second and third order terms, we need to use more cancellations. We start by considering the first order term for which we can write
\begin{align*}
& \frac{i \omega}{D_{m_j, 0} + i \omega} \big( \bsalpha \cdot A - V \big) \frac{1}{D_{m_j, 0} + i \omega} - \frac{i \omega}{D_{m_j, 0} -i \omega} \big( \bsalpha \cdot A - V \big) \frac{1}{D_{m_j, 0} - i \omega}\\
= & \frac{2 \omega^2}{D_{m_j, 0}^2 + \omega^2} \big( \bsalpha \cdot A - V \big) \frac{1}{D_{m_j, 0} + i \omega} + \frac{1}{D_{m_j, 0} - i \omega} \big( \bsalpha \cdot A - V \big) \frac{2 \omega^2}{D_{m_j, 0}^2 + \omega^2}.
\end{align*}
Inserting
\begin{equation}
\label{eq:Delta}
\frac{1}{D_{m_j, 0} \pm i \omega} = \frac{D_{m_j, 0} \mp i \omega}{D_{m_j, 0}^2 + \omega^2},
\end{equation}
we obtain
\begin{align*}
\frac{i \omega}{D_{m_j, 0} + i \omega} \big( \bsalpha \cdot A - V \big) \frac{1}{D_{m_j, 0} + i \omega} - \frac{i \omega}{D_{m_j, 0} -i \omega} \big( \bsalpha \cdot A - V \big) & \frac{1}{D_{m_j, 0} - i \omega}\\
= \frac{2 \omega^2}{D_{m_j, 0}^2 + \omega^2} \big\{ \bsalpha \cdot A - V, D_{m_j, 0} \big\} & \frac{1}{D_{m_j, 0}^2 + \omega^2},
\end{align*}
where the notation $\{ T_1, T_2 \}$ refers to the anti-commutator operator
$$\{ T_1, T_2 \} := T_1 T_2 + T_2 T_1.$$
At this stage, we recall that
\begin{equation}
\label{eq:prodalpha}
\big( \bsalpha \cdot X \big) \, \big( \bsalpha \cdot Y \big) = X \cdot Y + i (X \times Y) \cdot \bsSigma,
\end{equation}
for all $(X, Y) \in (\R^3)^2$. In this formula, $X \times Y$ is the cross product of the vectors $X$ and $Y$, whereas the notation $\bsSigma = (\bsSigma_1, \bsSigma_2, \bsSigma_3)$ refers to the matrices
\begin{equation}
\label{def:Sigma}
\bsSigma_j := \begin{pmatrix} \bssigma_j & 0 \\ 0 & \bssigma_j \end{pmatrix}.
\end{equation}
As a consequence, we obtain
$$\big\{ \bsalpha \cdot p, \bsalpha \cdot A \big\} = \big\{ p, A \big\}_{\R^3} + i \big( p \times A + A \times p \big) \cdot \bsSigma = \big\{ p, A \big\}_{\R^3} + B \cdot \bsSigma,$$
where $\{ \cdot, \cdot \}_{\R^3}$ is a notation for 
$$\big\{ S, T \big\}_{\R^3} := S \cdot T + T \cdot S$$
and where $p=-i\nabla$ (a simplifying notation that will be used henceforth). Since $\bsbeta \bsalpha_k + \bsalpha_k \bsbeta = 0$, we deduce that
\begin{equation}
\label{eq:antiDA}
\big\{ \bsalpha \cdot A - V, D_{m_j, 0} \big\} = \big\{ p , A - V \bsalpha \big\}_{\R^3} + B \cdot \bsSigma - 2 m_j V \bsbeta.
\end{equation}
This finally gives us
\begin{equation}
\label{def:boR1}
R_1(\omega, \bA) + R_1(- \omega, \bA) = 2 \omega^2 \big( \boR_{1, 1} + \boR_{1,2} \big),
\end{equation}
where
\begin{equation}
\label{def:boR11}
\boR_{1, 1} := \sum_{j = 0}^2 c_j \, \frac{1}{p^2 + m_j^2 + \omega^2} \Big( \big\{ p, A - V \bsalpha \big\}_{\R^3} + B \cdot \bsSigma \Big) \frac{1}{p^2 + m_j^2 + \omega^2},
\end{equation}
and
\begin{equation}
\label{def:boR12}
\boR_{1, 2} := - 2 \sum_{j = 0}^2 c_j m_j \, \frac{1}{p^2 + m_j^2 + \omega^2} \, V \bsbeta \, \frac{1}{p^2 + m_j^2 + \omega^2}.
\end{equation}
Concerning the operator $\boR_{1, 1}$, the last step consists in using identities~\eqref{cond:PV} and the two expansions
\begin{equation}
\label{eq:decomp_carres}
\begin{split}
\frac{1}{p^2 + m_j^2 + \omega^2} = \frac{1}{p^2 + m_0^2 + \omega^2} & + \frac{m_0^2 - m_j^2}{(p^2 + m_0^2 + \omega^2) (p^2 + m_j^2 + \omega^2)}\\
= \frac{1}{p^2 + m_0^2 + \omega^2} & + \frac{m_0^2 - m_j^2}{(p^2 + m_0^2 + \omega^2)^2}\\
& + \frac{(m_0^2 - m_j^2)^2}{(p^2 + m_j^2 + \omega^2) (p^2 + m_0^2 + \omega^2)^2}.
\end{split}
\end{equation}
This gives
\begin{align*}
& \boR_{1, 1} = \sum_{j = 0}^2 c_j (m_0^2 - m_j^2)^2 \bigg(\\
& \frac{1}{(p^2 + m_0^2 + \omega^2)^2} (\{ p, A - V \bsalpha \}_{\R^3} + B \cdot \bsSigma) \frac{1}{(p^2 + m_0^2 + \omega^2) (p^2 + m_j^2 + \omega^2)}\\
& + \frac{1}{(p^2 + m_j^2 + \omega^2) (p^2 + m_0^2 + \omega^2)^2} (\{ p, A - V \bsalpha \}_{\R^3}+ B \cdot \bsSigma) \frac{1}{p^2 + m_j^2 + \omega^2}\\
& + \frac{1}{p^2 + m_0^2 + \omega^2} (\{ p, A - V \bsalpha \}_{\R^3} + B \cdot \bsSigma) \frac{1}{(p^2 + m_j^2 + \omega^2) (p^2 + m_0^2 + \omega^2)^2} \bigg).
\end{align*}
We now use the fact that $\bA \in L^1(\R^3, \R^4)$, $B = i (p \times A + A \times p)$, as well as the Kato-Seiler-Simon inequality~\eqref{eq:KSS} to get
\begin{equation}
\label{eq:R11atrace}
\big\| \boR_{1, 1} \big\|_{\gS_1} \leq 18 \sum_{j = 0}^2 |c_j| (m_0^2 - m_j^2)^2 \frac{I_7}{(m_0^2 + \omega^2)^2} \big\| \bA \big\|_{L^1}.
\end{equation}

The analysis of the operator $\boR_{1, 2}$ is more involved. Under conditions~\eqref{cond:PV}, we are not able to prove that $\boR_{1, 2}$ is trace-class. However we can compute first the $\C^4$--trace before taking the operator trace. We obtain
\begin{equation}
\label{eq:trboR12}
\tr_{\C^4} \boR_{1, 2} = 0,
\end{equation}
since $\tr_{\C^4} \bsbeta = 0$. 

\begin{remark}
By this argument, we do not prove that $\boR_{1, 2}$ is trace-class. 
Under the additional conditions
$$\sum_j c_j \, m_j = \sum_j c_j \, m_j^3 = 0,$$
the operator $\boR_{1, 2}$ becomes a trace-class operator, and its trace is equal to $0$. This strategy however requires to introduce additional fictitious particles in our model. Introducing more fictitious particles in order to justify the computation of a term which is anyway $0$ does not seem very reasonable from a physical point of view. This explains why we prefer here to first take the $\C^4$--trace.
\end{remark}

As a consequence, we can conclude our estimate of the first order term by combining~\eqref{eq:R11atrace} and~\eqref{eq:trboR12} in order to obtain
\begin{multline}
\label{eq:1st-final}
\int_\R \Big\| \tr_{\C^4} \big( R_1(\omega, \bA) + R_1(- \omega, \bA) \big) \Big\|_{\gS_1} \, d\omega\\
\leq 36\, I_7 \sum_{j = 0}^2 |c_j| \, \frac{(m_0^2 - m_j^2)^2 }{m_0} \, \| \bA \|_{L^1} \int_\R \frac{\omega^2 \, d\omega}{(1 + \omega^2)^2}.
\end{multline}

\medskip

The second and third order terms are treated following the same method, except that the algebra is a little more tedious. We start by writing that
\begin{align*}
R_n(\omega, \bA) + R_n(- \omega, \bA) = 2 & \omega^2 \sum_{j = 0}^2 c_j \sum_{k = 0}^n \Big( \frac{1}{D_{m_j, 0} - i \omega} \big( \bsalpha \cdot A - V \big) \Big)^k \times\\
& \times \frac{1}{D_{m_j, 0}^2 + \omega^2} \Big( \big( \bsalpha \cdot A - V \big) \frac{1}{D_{m_j, 0} + i \omega} \Big)^{n - k}.
\end{align*}
We next expand as before using~\eqref{eq:Delta}. 

%%%%%%%%%%%%%%%%%%%%%%%%%%%%%%%%%%%%%%%%%%%%%%%
\addtocontents{toc}{\SkipTocEntry}
\subsection*{Estimate on the second order term}
%%%%%%%%%%%%%%%%%%%%%%%%%%%%%%%%%%%%%%%%%%%%%%%

For the second order term, we are left with
\begin{align*}
R_2(\omega, \bA) + & R_2(- \omega, \bA) = - 2 \omega^2 \sum_{j = 0}^2 c_j \times\\
& \times \bigg( \frac{\omega^2}{D_{m_j, 0}^2 + \omega^2} \Big( \big( \bsalpha \cdot A - V \big) \frac{1}{D_{m_j, 0}^2 + \omega^2} \Big)^2 - \sum_{k = 0}^2 \Big( \frac{D_{m_j, 0}}{D_{m_j, 0}^2 + \omega^2} \times\\
& \times \big( \bsalpha \cdot A - V \big) \Big)^k \frac{1}{D_{m_j, 0}^2 + \omega^2} \Big( \big( \bsalpha \cdot A - V \big) \frac{D_{m_j, 0}}{D_{m_j, 0}^2 + \omega^2} \Big)^{2 - k} \bigg),
\end{align*}
which may also be written as
\begin{equation}
\label{eq:R2inter}
\begin{split}
R_2(\omega, \bA) + R_2(- \omega, \bA) & = - 2 \omega^2 \sum_{j = 0}^2 c_j \bigg( \frac{1}{D_{m_j, 0}^2 + \omega^2} \big( \bsalpha \cdot A - V \big)^2 \frac{1}{D_{m_j, 0}^2 + \omega^2}\\
& - \frac{1}{D_{m_j, 0}^2 + \omega^2} \Big( \big\{ \bsalpha \cdot A - V, D_{m_j, 0} \big\} \frac{1}{D_{m_j, 0}^2 + \omega^2} \Big)^2 \bigg).
\end{split}
\end{equation}
Inserting
$$(\bsalpha \cdot A - V)^2 = |A|^2 + V^2 - 2 \, \bsalpha \cdot A \, V,$$
and~\eqref{eq:antiDA} into~\eqref{eq:R2inter}, we are led to
\begin{equation}
\label{eq:R2end}
R_2(\omega, \bA) + R_2(- \omega, \bA) = - 2 \omega^2 \big( \boR_{2, 1} + \boR_{2, 2} \big),
\end{equation}
where
\begin{align*}
\boR_{2, 1} & := \sum_{j = 0}^2 c_j \, \frac{1}{p^2 + m_j^2 + \omega^2} \bigg( \Big( |A|^2 + V^2 - 2 \, \bsalpha \cdot A \, V \Big) \frac{1}{p^2 + m_j^2 + \omega^2}\\
- 4 & m_j^2 \Big( V \, \frac{1}{p^2 + m_j^2 + \omega^2} \Big)^2 -\Big( \big( \big\{ p, A - V \bsalpha \big\}_{\R^3} + B \cdot \bsSigma \big) \frac{1}{p^2 + m_j^2 + \omega^2} \Big)^2 \bigg),
\end{align*}
and
\begin{align*}
\boR_{2, 2} := 2 \sum_{j = 0}^2 & c_j\, m_j\, \frac{1}{p^2 + m_j^2 + \omega^2} \times\\
\times &\bigg\{ V \bsbeta \, \frac{1}{p^2 + m_j^2 + \omega^2}, \Big( \big\{ p, A - V \bsalpha \big\}_{\R^3} + B \cdot \bsSigma \Big) \frac{1}{p^2 + m_j^2 + \omega^2} \bigg\}.
\end{align*}
The proof that $\boR_{2, 1}$ is trace-class is similar to the first order case, using~\eqref{eq:decomp_carres}. The final estimate is
\begin{align*}
\big\| \boR_{2, 1} \big\|_{\gS_1} \leq & \sum_{j = 0}^2 |c_j| \bigg( 8\, I_7 \frac{m_j^2 - m_0^2}{(m_0^2 + \omega^2)^2} \big\| \bA \big\|_{L^2} \big\| B \big\|_{L^2} + I_8 \frac{m_j^2 - m_0^2}{(m_0^2 + \omega^2)^\frac{5}{2}} \times\\
& \times \Big( 4 m_j^2 \big\| V \big\|_{L^2}^2 + \big\| B \big\|_{L^2}^2 + 8(m_j^2 - m_0^2) \big\| \bA \big\|_{L^2}^2 \Big) \bigg).
\end{align*}
Since
\begin{equation}
\label{eq:trboR22}
\tr_{\C^4} \boR_{2, 2} = 0,
\end{equation}
as for the first-order term, our final estimate is
\begin{multline}
\label{eq:2nd-final}
\int_\R \Big\| \tr_{\C^4} \big( R_2(\omega, A) + R_2(- \omega, A) \big) \Big\|_{\gS_1} \, d\omega\\
\leq \sum_{j = 0}^2 |c_j| \bigg( 8\, I_7 \, \frac{m_j^2 - m_0^2}{m_0} \, \big\| \bA \big\|_{L^2} \, \big\| B \big\|_{L^2} \, \int_\R \frac{\omega^2 \, d\omega}{(1 + \omega^2)^2} + I_8 \, \frac{m_j^2 - m_0^2}{m_0^2} \times\\
\times \Big( 4 m_j^2 \big\| V \big\|_{L^2}^2 + \big\| B \big\|_{L^2}^2 + 8 (m_j^2 - m_0^2) \big\| \bA \big\|_{L^2}^2 \Big) \int_\R \frac{\omega^2 \, d\omega}{(1 + \omega^2)^\frac{5}{2}} \bigg).
\end{multline}

%%%%%%%%%%%%%%%%%%%%%%%%%%%%%%%%%%%%%%%%%%%%%%
\addtocontents{toc}{\SkipTocEntry}
\subsection*{Estimate on the third order term}
%%%%%%%%%%%%%%%%%%%%%%%%%%%%%%%%%%%%%%%%%%%%%%

Similar computations give for the third order term
\begin{equation}
\label{eq:R3end}
\begin{split}
R_3(& \omega, \bA) + R_3(- \omega, \bA) = \sum_{j = 0}^2 c_j \bigg( \frac{2 \omega^2}{p^2 + m_j^2 + \omega^2} \Big( \big( \big\{ p, A - V \bsalpha \big\}_{\R^3} + B \cdot \bsSigma\\
& - 2 m_j V \bsbeta \big) \frac{1}{p^2 + m_j^2 + \omega^2} \Big)^3 - \frac{2 \omega^2}{p^2 + m_j^2 + \omega^2} \bigg\{ \big( |A|^2 + V^2 - 2 \, \bsalpha \cdot A \, V \big)\\
& \times \frac{1}{p^2 + m_j^2 + \omega^2}, \big( \big\{ p, A - V \bsalpha \big\}_{\R^3} + B \cdot \bsSigma - 2 m_j V \bsbeta \big) \frac{1}{p^2 + m_j^2 + \omega^2} \bigg\} \bigg).
\end{split}
\end{equation}
Using once again~\eqref{eq:decomp_carres}, we deduce
\begin{align*}
& \int_\R \Big\| \tr_{\C^4} \big( R_3(\omega, \bA) + R_3(- \omega, \bA) \big) \Big\|_{\gS_1} \, d\omega \leq K \sum_{j = 0}^2 |c_j| \, \int_\R \frac{\omega^2 \, d\omega}{(1 + \omega^2)^\frac{13}{8}} \times\\
& \times \bigg( m_j I_8 \| V \|_{L^3}^3 + I_8 \| V \|_{L^4}^2 \| B \|_{L^2} + \frac{I_6 (I_8)^\frac{1}{4}}{m_j^\frac{1}{4}} \| B \|_{L^2}^2 \| V \|_{L^4} + \frac{(I_{16/3})^\frac{3}{2}}{m_j^\frac{1}{2}} \| B \|_{L^2}^3\\
& + m_j I_7 \| \bA \|_{L^3} \| V \|_{L^3}^2 + \frac{(I_4)^\frac{1}{4} I_6}{m_j^\frac{1}{4}} \| B \|_{L^2}^2 \| \bA \|_{L^4} + \frac{I_8}{m_j} \| \bA \|_{L^3}^2 \| V \|_{L^3}\\
& + (m_j^2 - m_0^2) \Big( \frac{I_8}{m_j^2} \| \bA \|_{L^4}^2 \| B \|_{L^2} + \frac{I_7}{m_j} \| \bA \|_{L^3}^3 \Big) \bigg),
\end{align*}
for some universal constant $K$. 

\medskip

Combining with~\eqref{eq:6th-final},~\eqref{eq:5th-final},~\eqref{eq:4th-final},~\eqref{eq:1st-final} and~\eqref{eq:2nd-final}, we obtain~\eqref{eq:thekey}, provided that $\bA$ is in $L^1(\R^3, \R^3) \cap H^1(\R^3, \R^3)$. This concludes the proof of Proposition~\ref{prop:trace-class}.
\end{proof}

%%%%%%%%%%%%%%%%%%%%%%%%%%%%%%%%%%%%%%%%%%%%%%%%%%%%%%%%%%%%%%%%%%%%%%%
%%%%%%%%%%%%%%%%%%%%%%%%%%%%%%%%%%%%%%%%%%%%%%%%%%%%%%%%%%%%%%%%%%%%%%%
\section{Estimates involving the field energy}
\label{sec:estimates}
%%%%%%%%%%%%%%%%%%%%%%%%%%%%%%%%%%%%%%%%%%%%%%%%%%%%%%%%%%%%%%%%%%%%%%%
%%%%%%%%%%%%%%%%%%%%%%%%%%%%%%%%%%%%%%%%%%%%%%%%%%%%%%%%%%%%%%%%%%%%%%%

In Proposition~\ref{prop:trace-class} above we have shown that the operator 
$$\tr_{\C^4} T_{\bA} := \frac{1}{2} \sum_{j = 0}^2 c_j \, \tr_{\C^4} \big( |D_{m_j, 0}| - |D_{m_j, \bA}| \big),$$
is trace-class when $\bA$ decays fast enough. More precisely, in the proof of Proposition~\ref{prop:trace-class}, we have written
\begin{equation}
\label{eq:dev-TA}
\begin{split}
T_{\bA} = \sum_{n = 1}^5 T_n(\bA) + T_6'(\bA) := & \frac{1}{4 \pi} \sum_{n = 1}^5 \int_\R \big( R_n(\omega, \bA) + R_n(- \omega, \bA) \big) \, d\omega\\
& + \frac{1}{4 \pi} \int_\R \big( R_6'(\omega, \bA) + R_6'(- \omega, \bA) \big) \, d\omega,
\end{split}
\end{equation}
with $R_n$ and $R_6'$ given by~\eqref{eq:def_R_n} and~\eqref{eq:def_R_6}, and we have proved that the operators $\tr_{\C^4} T_n(\bA)$ and $\tr_{\C^4} T_6'(\bA)$ are trace-class. However our estimates involve non gauge-invariant quantities (some $L^p$ norms of $\bA$) and they require that $\bA$ decays fast enough at infinity. 

In this section, we establish better bounds on these different terms. We are interested in having estimates which only involve the field $\bF = (- \nabla V, \curl A)$ through the norms $\| \nabla V \|_{L^2}$ and $\| \curl A \|_{L^2}$. Our simple estimate~\eqref{eq:6th-final} on the sixth order only depends on the field $\bF$. But we will also need to know that the sixth order is continuous, which will require some more work. For the other terms, we have to get the exact cancellations.

With these estimates at hand, it will be easy to show that $\boF_{\rm PV}$ can be uniquely extended to a continuous function on the Coulomb-gauge homogeneous Sobolev space $\Hdiv$, as stated in Theorem~\ref{thm:defF}, and which we do in the next section.

\begin{remark}
In the estimates of the previous section, it was not important that $\div A = 0$. We have to use this property now.
\end{remark}

%%%%%%%%%%%%%%%%%%%%%%%%%%%%%%%%%%%%%%%%%%%%
\subsection{The odd orders vanish}
\label{sec:odd_orders}
%%%%%%%%%%%%%%%%%%%%%%%%%%%%%%%%%%%%%%%%%%%%

The following lemma says that the trace of the odd order operators $\tr_{\C^4} T_1(\bA)$, $\tr_{\C^4} T_3(\bA)$ and $\tr_{\C^4} T_5(\bA)$ vanish. This well-known consequence of the charge-conjugation invariance is sometimes called Furry's theorem, see~\cite{Furry-37} and~\cite[Sec.4.1]{GreRei-08}.

\begin{lemma}[The odd orders vanish]
\label{lem:odd-zero}
For $\bA\in \Hdiv\cap L^1(\R^3,\R^4)$ and $n = 1, 3, 5$, we have
\begin{equation}
\label{eq:oddnul}
\tr \big( \tr_{\C^4} T_n(\bA) \big) = \frac{1}{4 \pi} \int_\R \tr \Big( \tr_{\C^4} \big( R_n(\omega, \bA) + R_n(- \omega, \bA) \big) \Big) \, d\omega = 0.
\end{equation}
\end{lemma}

\begin{proof}
Let $\boC \psi := i \beta \alpha_2 \overline{\psi}$ be the (anti-unitary) charge-conjugation operator. Since $\boC \, D_{m_j, 0} \, \boC^{- 1} = - D_{m_j, 0}$, we have
$$\boC \, \big( D_{m_j, 0} \pm i \omega \big)^{- 1} \, \boC^{- 1} = - \big( D_{m_j, 0} \pm i \omega \big)^{- 1}.$$
Similarly, since $A$ and $V$ are real-valued, we can write
$$\boC \, \bsalpha \cdot A \, \boC^{- 1} = \bsalpha \cdot A \quad {\rm and} \quad \boC \, V \, \boC^{- 1} = V,$$
so that
\begin{equation}
\label{lopez}
\boC \, R_n(\pm \omega, \bA) \, \boC^{- 1} = (- 1)^n R_n(\pm \omega, \bA).
\end{equation}

At this stage, we can compute
\begin{equation}
\label{gomis}
\tr_{\C^4} \big( \boC T \boC^{- 1} \big) = \tr_{\C^4} \overline{T},
\end{equation}
for any operator $T$ on $L^2(\R^3, \C^4)$. Here, $\overline{T}$ refers to the operator defined as
$$\overline{T}(f) := \overline{T(\overline{f})}.$$
When $\tr_{\C^4} T$ is trace-class, so is the operator $\tr_{\C^4} \overline{T}$, and its trace is equal to
\begin{equation}
\label{briand}
\tr \big( \tr_{\C^4} \overline{T} \big) = \overline{\tr(\tr_{\C^4} T)}.
\end{equation}
As a consequence, the operator $\tr_{\C^4}(\boC T \boC^{- 1})$ is trace-class, as soon as $T$ is trace-class, and its trace is the complex conjugate of the trace of $T$.

Finally, recall that we have established in the proof of Proposition~\ref{prop:trace-class} that the operators $\tr_{\C^4} (R_n(\omega, \bA) + R_n(- \omega, \bA))$ are trace-class for $n = 1, 3, 5$. Combining~\eqref{lopez} with~\eqref{gomis} and~\eqref{briand}, we obtain
\begin{align*}
\tr \Big( \tr_{\C^4} \big( R_n(\omega, \bA) + R_n(- & \omega, \bA) \big) \Big)\\
& = (- 1 )^n \overline{\tr \Big( \tr_{\C^4} \big( R_n(\omega, \bA) + R_n(- \omega, \bA) \big) \Big)}.
\end{align*}
We deduce that the quantity $\tr \big( \tr_{\C^4} (R_n(\omega, \bA) + R_n(- \omega, \bA)) \big)$ is purely imaginary when $n$ is odd, so that the trace of $\tr_{\C^4} T_n(\bA)$ is purely imaginary. Since the operator $\tr_{\C^4} T_n(\bA)$ is self-adjoint, its trace is necessarily equal to $0$. This gives Formula~\eqref{eq:oddnul}.
\end{proof}

%%%%%%%%%%%%%%%%%%%%%%%%%%%%%%%%%%
\subsection{The second order term}
\label{sec:2ndorder}
%%%%%%%%%%%%%%%%%%%%%%%%%%%%%%%%%%

We now compute exactly the second order term $T_2(\bA)$ appearing in the decomposition of $T_{\bA}$, assuming that $\bA$ belongs to $H^1(\R^3, \R^4)$ and $\div A = 0$. We will verify that it only depends on the electromagnetic fields $E := - \nabla V$ and $B := \curl A$. 

\begin{lemma}[Formula for the second order term]
\label{lem:2ndform}
For $\bA \in \Hdiv\cap L^2(\R^3, \R^4)$, we have
\begin{equation}
\label{eq:trT2A}
\tr \big( \tr_{\C^4} T_2(\bA) \big) = \frac{1}{8 \pi} \int_{\R^3} M(k) \big( |\widehat{B}(k)|^2 - |\widehat{E}(k)|^2 \big) \, dk:=\boF_2(\bF),
\end{equation}
where $M$ is the function defined in~\eqref{def:M} and $\bF=(E,B)$.
\end{lemma}

The formula~\eqref{def:M} for the function $M$ which is nothing but the dielectric response of Dirac's vacuum is well-known in the physical literature. The following proof is inspired of the calculations in~\cite[p. 280--282]{GreRei-08}.

\begin{proof}
In the course of the proof of Proposition~\ref{prop:trace-class}, we have shown that the operator $\tr_{\C^4} T_2(\bA)$ is trace-class when $\bA \in H^1(\R^3, \R^3)$ (see inequality~\eqref{eq:2nd-final}). As a consequence, its trace is well-defined and given by
\begin{equation}
\label{eq:tr}
\tr \big( \tr_{\C^4} T_2(\bA) \big) = \int_{\R^3} \widehat{\big( \tr_{\C^4} T_2(\bA) \big)}(p, p) \, dp.
\end{equation}
Here, $\widehat{\tr_{\C^4} T_2(\bA)}$ refers to the Fourier transform of the trace-class operator $\tr_{\C^4} T_2(\bA)$. Our convention for the Fourier transform of a trace-class operator $T$ is the following
$$\widehat{T}(p, q) := \frac{1}{(2 \pi)^3} \int_{\R^6} T(x, y) e^{- i p \cdot x} e^{i q \cdot y} \, dx \, dy.$$
In view of~\eqref{eq:R2end}, the operator $\tr_{\C^4} T_2(\bA)$ is given by
\begin{equation}
\label{eq:T2}
\begin{split}
\tr_{\C^4} & T_2(\bA)(p,p) = \frac{2}{\pi} \int_\R \sum_{j = 0}^2 c_j \frac{1}{p^2 + m_j^2 + \omega^2} \bigg( - \big( |A|^2 + V^2 \big) \frac{1}{p^2 + m_j^2 + \omega^2}\\
& + \Big( \big\{ p, A \big\}_{\R^3} \frac{1}{p^2 + m_j^2 + \omega^2} \Big)^2 + \sum_{k = 1}^3 B_k \frac{1}{p^2 + m_j^2 + \omega^2} B_k \frac{1}{p^2 + m_j^2 + \omega^2}\\
& + \sum_{k = 1}^3 \big\{ p_k, V \big\} \frac{1}{p^2 + m_j^2 + \omega^2} \big\{ p_k, V \big\} \frac{1}{p^2 + m_j^2 + \omega^2}\\
& + 4 m_j^2 \Big( V \frac{1}{p^2 + m_j^2 + \omega^2} \Big)^2 \bigg) \omega^2 \, d\omega.
\end{split}
\end{equation}
Since $\div A=0$, we deduce after a lengthy calculation that 
\begin{equation}
\label{eq:trT2}
\tr \big( \tr_{\C^4} T_2(\bA) \big) = \sum_{k = 1}^3 \boT_{2, k},
\end{equation}
where
\begin{equation}
\label{def:boT1}
\boT_{2, 1} := - \frac{1}{\sqrt{2} \pi^\frac{5}{2}} \int_{\R^4} \sum_{j = 0}^2 c_j \frac{\omega^2 \, d\omega \, dp}{(p^2 + m_j^2 + \omega^2)^2} \, \big( \widehat{|A|^2}(0) + \widehat{V^2}(0) \big),
\end{equation}
\begin{equation}
\label{def:boT2}
\begin{split}
\boT_{2, 2} &:= \frac{1}{\pi^4} \int_{\R^7} \sum_{j = 0}^2 c_j \frac{dk \, \omega^2 \, d\omega \, dp}{(p^2 + m_j^2 + \omega^2)^2 ((p - k)^2 + m_j^2 + \omega^2)} \times\\
& \qquad\qquad\qquad\times \Big( \big( p \cdot \widehat{A}(k) \big) \big( p \cdot \widehat{A}(- k) \big) + \big( p^2 + m_j^2 \big) |\widehat{V}(k)|^2 \Big)\\
& := \boT_{2, 2}(A) + \boT_{2, 2}(V),
\end{split}
\end{equation}
and
\begin{equation}
\label{def:boT3}
\boT_{2, 3} := \frac{1}{4 \pi^4} \int_{\R^7} \sum_{j = 0}^2 c_j \frac{k^2 |\widehat{A}(k)|^2 + (k^2 - 4 p \cdot k) |\widehat{V}(k)|^2}{(p^2 + m_j^2 + \omega^2)^2 ((p - k)^2 + m_j^2 + \omega^2)} \,dk \, \omega^2 \, d\omega \, dp.
\end{equation}
We next use the following \emph{Ward identities}~\cite{Ward-50}
\begin{equation}
\label{eq:Ward}
\begin{split}
\int_{\R^3} \sum_{j = 0}^2 c_j & \frac{p_m p_n \, dp}{(p^2 + m_j^2 + \omega^2)^2 ((p - k)^2 + m_j^2 + \omega^2)}\\
& = \int_{\R^3} \sum_{j = 0}^2 c_j \frac{(k_m - q_m) (k_n - q_n) \, dq}{((q - k)^2 + m_j^2 + \omega^2)^2 (q^2 + m_j^2 + \omega^2)},
\end{split}
\end{equation}
for all $(m, n) \in \{ 1, 2, 3 \}^2$ and all $k \in \R^3$. This equation is nothing else than a change of variables $p = k - q$, which makes perfect sense thanks to conditions~\eqref{cond:PV} which guarantee the convergence of the integral. Its importance is well-known in the Physics literature, see, e.g.,~\cite[Sec. 7.4]{PesSch-95}. Since $\div A = 0$, we infer that
\begin{align*}
\boT_{2, 2}(A) = - \frac{1}{4 \pi^4} & \int_{\R^4} \sum_{m = 1}^3 \sum_{n = 1}^3 \widehat{A_m}(k) \widehat{A_n}(- k) \, dk \, \omega^2 \, d\omega \times\\
\times \int_{\R^3} & \sum_{j = 0}^2 c_j p_m \partial_{p_n} \Big( \frac{1}{(p^2 + m_j^2 + \omega^2) ((p - k)^2 + m_j^2 + \omega^2)} \Big) \, dp.
\end{align*}
Integrating by parts, we are led to
\begin{equation}
\label{eq:boT2A}
\boT_{2, 2}(A) = \frac{1}{4 \pi^4} \int_{\R^7} \sum_{j = 0}^2 c_j \frac{|\widehat{A}(k)|^2 \, dk \, \omega^2 \, d\omega \, dp}{(p^2 + m_j^2 + \omega^2) ((p - k)^2 + m_j^2 + \omega^2)}.
\end{equation}
Similarly, we can compute
\begin{align*}
\boT_{2, 2}(V) = \frac{1}{\pi^4} \int_{\R^7} \sum_{j = 0}^2 c_j & \frac{|\widehat{V}(k)|^2 \, dk \, \omega^2 \, d\omega \, dp}{(p^2 + m_j^2 + \omega^2) ((p - k)^2 + m_j^2 + \omega^2)}\\
- & \frac{1}{\pi^4} \int_{\R^7} \sum_{j = 0}^2 c_j \frac{|\widehat{V}(k)|^2 \, dk \, \omega^4 \, d\omega \, dp}{(p^2 + m_j^2 + \omega^2)^2 ((p - k)^2 + m_j^2 + \omega^2)}.
\end{align*}
Integrating by parts with respect to $\omega$, one can check that
\begin{align*}
\int_{\R^4} \sum_{j = 0}^2 c_j & \frac{\omega^4 \, d\omega \, dp}{(p^2 + m_j^2 + \omega^2)^2 ((p - k)^2 + m_j^2 + \omega^2)}\\
& = \frac{3}{4} \int_{\R^4} \sum_{j = 0}^2 c_j \frac{\omega^2 \, d\omega \, dp}{(p^2 + m_j^2 + \omega^2) ((p - k)^2 + m_j^2 + \omega^2)}, 
\end{align*}
so that
\begin{equation}
\label{eq:boT2V}
\boT_{2, 2}(V) = \frac{1}{4 \pi^4} \int_{\R^7} \sum_{j = 0}^2 c_j \frac{|\widehat{V}(k)|^2 \, dk \, \omega^2 \, d\omega \, dp}{(p^2 + m_j^2 + \omega^2) ((p - k)^2 + m_j^2 + \omega^2)}.
\end{equation}
On the other hand, since $A$ and $V$ are real-valued, we have
$$\widehat{|A|^2}(0) + \widehat{V^2}(0) = \frac{1}{(2 \pi)^\frac{3}{2}} \int_{\R^3} \big( |\widehat{A}(k)|^2 + |\widehat{V}(k)|^2 \big) \, dk,$$
hence
$$\boT_{2, 1} = - \frac{1}{4 \pi^4} \int_{\R^7} \sum_{j = 0}^2 c_j \frac{|\widehat{A}(k)|^2 + |\widehat{V}(k)|^2}{(p^2 + m_j^2 + \omega^2)^2} \, \omega^2 \, d\omega \, dk \, dp.$$
Combining with~\eqref{eq:boT2A} and~\eqref{eq:boT2V}, we arrive at
$$\boT_{2, 1} + \boT_{2, 2} = \frac{1}{4 \pi^4} \int_{\R^7} \sum_{j = 0}^2 c_j \frac{\big( 2 p \cdot k - |k|^2 \big) \big( |\widehat{A}(k)|^2 + |\widehat{V}(k)|^2 \big)}{(p^2 + m_j^2 + \omega^2)^2 ((p - k)^2 + m_j^2 + \omega^2)} \, \omega^2 \, d\omega \, dk \, dp.$$
In view of~\eqref{eq:trT2} and~\eqref{def:boT3}, this provides
\begin{equation}
\label{def:G}
\tr \big( \tr_{\C^4} T_2(A) \big) = \int_{\R^3} G(k) \big( |\widehat{A}(k)|^2 - |\widehat{V}(k)|^2 \big) \, dk,
\end{equation}
where
$$G(k) := \frac{1}{2 \pi^4} \int_{\R^4} \sum_{j = 0}^2 c_j \frac{p \cdot k \, \omega^2 \, d\omega \, dp}{(p^2 + m_j^2 + \omega^2)^2 ((p - k)^2 + m_j^2 + \omega^2)}.$$
We next use the identity
$$\frac{1}{a^2 b} = \int_0^1 \bigg( \int_0^\infty s^2 e^{- s(u a + (1 - u) b)} \, ds \bigg) u \, du,$$
see~\cite[Chap. 5]{GreRei-08}, to rewrite
\begin{equation}
\label{eq:G1}
\begin{split}
G(k) = \frac{1}{2 \pi^4} & \int_{\R^4} \sum_{j = 0}^2 c_j \, \omega^2 \, d\omega \, p \cdot k \, dp \times\\
& \times \bigg( \int_0^1 \int_0^\infty e^{- s (p^2 + m_j^2 + \omega^2) - s(1 - u)(k^2 - 2 p \cdot k)} \, s^2 \, ds \, u \, du \bigg).
\end{split}
\end{equation}
Using conditions~\eqref{cond:PV}, we can invoke Fubini's theorem to recombine the integrals in~\eqref{eq:G1} as
\begin{align*}
G(k) = \frac{1}{2 \pi^4} \int_0^1 \int_0^\infty \sum_{j = 0}^2 c_j & \, e^{- s (m_j^2 + (1 - u) k^2)} s^2 \, ds \, u \, du \times\\
& \times \int_{\R^3} p \cdot k \, e^{- s (p^2 - 2 (1 - u) p \cdot k)} \bigg( \int_\R e^{- s \omega^2} \omega^2 \, d\omega \bigg) dp.
\end{align*}
Since
$$\int_\R e^{- s \omega^2} \omega^2 \, d\omega = \frac{\sqrt{\pi}}{2 s^\frac{3}{2}},$$
and
\begin{align*}
\int_{\R^3} p \cdot k \, e^{- s (p^2 - 2 (1 - u) p \cdot k)} \, dp = & k \cdot \nabla \bigg( \int_{\R^3} e^{p \cdot x - s p^2 - 2 s (1 - u) p \cdot k} \, dp \bigg)_{| x = 0}\\
= & \Big( \frac{\pi}{s} \Big)^\frac{3}{2} (1 - u) k^2 e^{s (1 - u)^2 k^2},
\end{align*}
we deduce that
\begin{equation}
\label{eq:G2}
G(k) = \frac{k^2}{4 \pi^2} \int_0^1 \int_0^\infty \sum_{j = 0}^2 c_j \, e^{- s (m_j^2 + u (1 - u) k^2)} s^{- 1} \, ds \, u (1 -u) \, du.
\end{equation}
Integrating by parts, we now compute
\begin{align*}
\int_0^\infty \sum_{j = 0}^2 c_j \, & e^{- s (m_j^2 + u (1 - u) k^2)} s^{- 1} \, ds\\
& = \int_0^\infty \sum_{j = 0}^2 c_j \, \log(s) e^{- s (m_j^2 + u (1 - u) k^2)} \, (m_j^2 + u (1 - u) k^2) \, ds,
\end{align*}
which is justified again thanks to conditions~\eqref{cond:PV}. Letting $\sigma = s (m_j^2 + u (1 - u) k^2)$, we infer again from~\eqref{cond:PV} that
\begin{align*}
\int_0^\infty \sum_{j = 0}^2 c_j e^{- s (m_j^2 + u (1 - u) k^2)} s^{- 1} \, ds = & \int_0^\infty \sum_{j = 0}^2 c_j \log \Big( \frac{\sigma}{m_j^2 + u (1 - u) k^2} \Big) e^{- \sigma} \, d\sigma\\
= & - \sum_{j = 0}^2 c_j \, \log(m_j^2 + u (1 - u) k^2).
\end{align*}
Inserting into~\eqref{eq:G2}, we get
$$G(k) = - \frac{k^2}{4 \pi^2} \int_0^1 \sum_{j = 0}^2 c_j \, u (1 -u) \, \log(m_j^2 + u (1 - u) k^2) \, du = \frac{k^2}{8 \pi} M(k).$$
Combining with~\eqref{def:G}, we obtain Formula~\eqref{eq:trT2A}.
\end{proof}

We complete our analysis of the second order term by giving the main properties of the function $M$.

\begin{lemma}[Main properties of $M$]
\label{lem:prop_M}
Assume that $c_j$ and $m_j$ satisfy~\eqref{cond:c_i_thm}. The function $M$ given by~\eqref{def:M} is well-defined and positive on $\R^3$, and satisfies
$$0 < M(k) \leq M(0) = \frac{2 \, \log(\Lambda)}{3 \pi},$$
where $\Lambda$ is defined by~\eqref{def:Lambda}. Moreover,
$$\frac{2 \, \log(\Lambda)}{3 \pi} - M(k) \to \frac{|k|^2}{4 \pi} \int_0^1 \frac{z^2 - z^4/3}{1 + |k|^2 (1 - z^2)/4} \, dz,$$
when $m_1 \to \infty$ and $m_2 \to \infty$.
\end{lemma}

Most of the above properties of $M$ are well-known in the Physics literature, see for instance~\cite[Sec. 5.2]{GreRei-08}. The positivity of $M(k)$ for all $k$, which is crucial for our study of the nonlinear Lagrangian action, does not seem to have been remarked before.

\begin{proof}
In view of~\eqref{def:M}, the function $M$ is well-defined on $\R^3$. Concerning its positivity, we set
$$\Phi(t) := \sum_{j = 0}^2 c_j \, \log(m_j^2 + t),$$
for all $t \geq 0$. Using~\eqref{cond:coef}, we compute
$$\Phi'(t) = \sum_{j = 0}^2 \frac{c_j}{m_j^2 + t} = \frac{(m_1^2 - m_0^2) (m_2^2 - m_0^2)}{(m_0^2 + t)(m_1^2 + t)(m_2^2 + t)} > 0.$$
Since $\Phi(0) = - 2 \log \Lambda < 0$ and
$$\Phi(t) = \sum_{j = 0}^2 c_j \, \log \Big(1 + \frac{m_j^2}{t} \Big) \to 0, \ {\rm as} \ t \to \infty,$$
by~\eqref{cond:PV}, we deduce that
$$- 2 \log \Lambda < \Phi(t) < 0,$$
for all $t > 0$. Inserting into~\eqref{def:M}, we obtain~\eqref{est:M}.

As for~\eqref{eq:limitM}, we first write
\begin{align*}
M(k) = & - \frac{2}{\pi} \int_0^1 \sum_{j = 0}^2 c_j \, u (1 -u) \, \Big( \log(m_j^2) + \log \Big( 1 + \frac{u (1 - u) k^2}{m_j^2} \Big) \Big) \, du\\
= & \frac{2 \, \log \Lambda}{3 \pi} - \frac{2}{\pi} \int_0^1 \sum_{j = 0}^2 c_j \, u (1 -u) \, \log \Big( 1 + \frac{u (1 - u) k^2}{m_j^2} \Big) \, du.
\end{align*}
When $m_1 \to \infty$ and $m_2 \to \infty$, we infer that
$$\frac{2 \, \log \Lambda}{3 \pi} - M(k) \to \frac{2}{\pi} \int_0^1 u (1 -u) \, \log \big( 1 + u (1 - u) k^2 \big) \, du.$$
Integrating by parts, we compute
$$\int_0^1 u (1 -u) \, \log \big( 1 + u (1 - u) k^2 \big) \, du = - \int_0^1 \frac{(\frac{u^2}{2} - \frac{u^3}{3})(1 - 2 u) \, du}{1 + u(1 - u) k^2},$$
so that it only remains to set $z = 1 - 2 u$ to derive~\eqref{eq:limitM}. This concludes the proof of Lemma~\ref{lem:prop_M}.
\end{proof}

%%%%%%%%%%%%%%%%%%%%%%%%%%%%%%%%%%
\subsection{The fourth order term}
\label{sec:4thorder}
%%%%%%%%%%%%%%%%%%%%%%%%%%%%%%%%%%

Our goal is now to provide an estimate on the fourth order term $T_4(\bA)$. We have estimated this term in~\eqref{eq:4th-final}, and we know that $T_4(\bA)$ is trace-class when $\bA \in L^4(\R^3, \R^4)$. Here, we want to get an estimate involving only the norm of $\bA$ in $\Hdiv$.

\begin{lemma}[Estimate for the fourth order term]
\label{lem:4thterm}
Let $\bA = (A, V) \in L^4(\R^3, \R^4) \cap \Hdiv$ and set $B := \curl A$ and $E := - \nabla V$. There exists a universal constant $K$ such that
\begin{equation}
\label{eq:trT4A}
\big| \tr \big( \tr_{\C^4} T_4(\bA) \big) \big| = \big| \tr T_4(\bA) \big| \leq K \bigg( \sum_{j = 0}^2 \frac{|c_j|}{m_j} \bigg) \big( \| B \|_{L^2} + \| E \|_{L^2} \big)^4.
\end{equation}
\end{lemma}

\begin{proof}
Arguing as in the proof of Proposition~\ref{prop:trace-class} (see the proof of Formulas~\eqref{eq:R2end} and~\eqref{eq:R3end}), we decompose $T_4(\bA)$ as
\begin{equation}
\label{eq:T4}
T_4(\bA) = T_{4, 1}(\bA) - T_{4, 2}(\bA) + T_{4, 3}(\bA),
\end{equation}
where
\begin{equation}
\label{def:T41}
T_{4, 1}(\bA) := \frac{1}{2 \pi} \int_\R \sum_{j = 0}^2 c_j \, \omega^2 \, d\omega \, \frac{1}{p^2 + m_j^2 + \omega^2} \bigg( \boW_2 \, \frac{1}{p^2 + m_j^2 + \omega^2} \bigg)^2,
\end{equation}
\begin{equation}
\label{def:T42}
\begin{split}
T_{4, 2}(\bA) & := \frac{1}{2 \pi} \int_\R \sum_{j = 0}^2 c_j \, \omega^2 \, d\omega \, \frac{1}{p^2 + m_j^2 + \omega^2} \bigg( \boW_2 \, \frac{1}{p^2 + m_j^2 + \omega^2} \times\\
& \times \Big( \boW_1 \, \frac{1}{p^2 + m_j^2 + \omega^2} \Big)^2 + \Big( \boW_1 \, \frac{1}{p^2 + m_j^2 + \omega^2} \Big)^2 \boW_2 \, \frac{1}{p^2 + m_j^2 + \omega^2}\\
& + \boW_1 \, \frac{1}{p^2 + m_j^2 + \omega^2} \, \boW_2 \, \frac{1}{p^2 + m_j^2 + \omega^2} \, \boW_1 \, \frac{1}{p^2 + m_j^2 + \omega^2} \bigg), 
\end{split}
\end{equation}
and
\begin{equation}
\label{def:T43}
T_{4, 3}(\bA) := \frac{1}{2 \pi} \int_\R \sum_{j = 0}^2 c_j \, \omega^2 \, d\omega \, \frac{1}{p^2 + m_j^2 + \omega^2} \bigg( \boW_1 \, \frac{1}{p^2 + m_j^2 + \omega^2} \bigg)^4.
\end{equation}
Here, we have set for shortness,
$$\boW_1 := \big\{ p, A - V \bsalpha \big\}_{\R^3} + B \cdot \bsSigma - 2 m_j V \bsbeta,$$
and
$$\boW_2 := |A|^2 + V^2 - 2 \big( \bsalpha \cdot A \big) V.$$

Let us now explain our method to establish~\eqref{eq:trT4A}. When looking at $\boT_{4, k}$ with $k = 1, 2, 3$, we are worried about several terms. First the function $\boW_2$ does not decay too fast, it is only in $L^3(\R^3)$ if we only want to use the $L^6$ norm of $\bA$. Furthermore, it involves quantities which are not gauge invariant. Similarly, the term involving $p$ in $\boW_1$ involves non-gauge invariant quantities. On the other hand, the term involving $B$ is in $L^2(\R^3)$ and it is gauge invariant. The term involving $V$ alone is also not gauge invariant but it has the matrix $\bsbeta$ which will help us, and it has no $p$. Since the result should be gauge invariant, these terms cannot be a problem. They should not contribute to the total (fourth order) energy.

In order to see this, we use the following technique. In Formulas~\eqref{def:T41}--\eqref{def:T43}, we commute all the operators involving $p$ in order to place them either completely on the left or completely on the right. We have to commute the terms $(p^2 + m_j^2 + \omega^2)^{-1}$ as well as the $p$ appearing in $\boW_1$. We think that it does not matter how many terms we put on the left and on the right. The main point is to have some functions of $p$ on both sides (to get a trace-class operator under suitable assumptions on $\bA$). All the commutators obtained by these manipulations are better behaved and they will be estimated using the Kato-Seiler-Simon inequality~\eqref{eq:KSS}, only in terms of $\| \bA \|_{\Hdiv}$. 

In the end of the process, we will be left with a sum of terms of the form
$$\frac{|p|^c}{(p^2 + m_j^2 + \omega^2)^a} \, f(x) \, \frac{|p|^d}{(p^2 + m_j^2 + \omega^2)^b},$$
where $f(x)$ is $\boW_2^2$ or a product of $\boW_2$ with some of the functions appearing in $\boW_1$, or only these functions. For instance, when we take the trace, the worst term involving only $V$ is
\begin{align*}
\bigg( \int_{\R^3} V^4 \bigg) \bigg( \int_{\R^3} \frac{dp}{(p^2 + m_j^2 + \omega^2)^3} + & 3 \int_{\R^3} \frac{|p|^2 \, dp}{(p^2 + m_j^2 + \omega^2)^4}\\
+ & \int_{\R^3} \frac{|p|^4 \, dp}{(p^2 + m_j^2 + \omega^2)^5} \bigg).
\end{align*}
Here the integrals over $p$ come respectively from $\boT_{4, 1}$, $\boT_{4, 2}$ and $\boT_{4, 3}$ and they behave exactly like $(\omega^2 + m_j^2)^{- 5/2}$. So we run into problems when we want to multiply by $\omega^2$ and then integrate with respect to $\omega$. But this term cannot be a problem here because $\int_{\R^3} V^4$ is not a gauge invariant quantity. This is where the Pauli-Villars scheme helps us. Not only these integrals will become well-defined, but also their sum will simply vanish because the regularization was precisely designed to preserve gauge invariance.

But before we explain all this in details, let us indicate how to handle the multiple commutators that we get when commuting the operators involving $p$. We start with $T_{4, 1}$, for instance. Following the general strategy explained above, we write 
\begin{align*}
T_{4, 1}(\bA) = & \frac{1}{2 \pi} \int_\R \sum_{j = 0}^2 c_j \, \omega^2 \, d\omega \bigg( \frac{1}{p^2 + m_j^2 + \omega^2} \, (\boW_2)^2 \frac{1}{(p^2 + m_j^2 + \omega^2)^2}\\
& + \frac{1}{p^2 + m_j^2 + \omega^2} \, \boW_2 \, \Big[ \frac{1}{p^2 + m_j^2 + \omega^2}, \boW_2 \Big] \frac{1}{p^2 + m_j^2 + \omega^2}\bigg),
\end{align*}
where, as usual, $[S, T] : = S T - T S$. We notice that
$$\Big[ \boW_2, \frac{1}{p^2 + m_j^2 + \omega^2} \Big] = \frac{1}{p^2 + m_j^2 + \omega^2} \, \big[ p^2, \boW_2 \big] \, \frac{1}{p^2 + m_j^2 + \omega^2},$$
while
$$\big[ p^2, \boW_2 \big] = p \big[ p, \boW_2 \big] + \big[ p, \boW_2 \big] p = - i \big\{ p, \nabla \boW_2 \big\}_{\R^3}.$$
Hence, we have
\begin{multline}
\label{eq:T41}
T_{4, 1}(\bA) = \frac{1}{2 \pi} \int_\R \sum_{j = 0}^2 c_j \, \omega^2 \, d\omega \bigg( \frac{1}{p^2 + m_j^2 + \omega^2} \, (\boW_2)^2 \frac{1}{(p^2 + m_j^2 + \omega^2)^2}\\
+i \frac{1}{p^2 + m_j^2 + \omega^2} \, \boW_2 \, \frac{1}{p^2 + m_j^2 + \omega^2} \, \big\{ p, \nabla \boW_2 \big\}_{\R^3} \, \frac{1}{(p^2 + m_j^2 + \omega^2)^2} \bigg),
\end{multline}
where
$$\nabla \boW_2 = 2 A \cdot \nabla A + 2 V \, \nabla V - 2 V \, \big( \bsalpha \cdot \nabla A \big) - 2 \nabla V \, \big( \bsalpha \cdot A \big).$$
We then argue as in the proof of Proposition~\ref{prop:trace-class}. We use that $\boW_2\in L^3(\R^3)$, with
$$\| \boW_2 \|_{L^3} \leq K \| \bA \|_{L^6}^2 \leq K \| \bA \|_{\Hdiv}^2,$$
and that $\nabla \boW_2 \in L^\frac{3}{2}(\R^3)$, with
$$\| \nabla\boW_2 \|_{L^\frac{3}{2}} \leq K \| \bA \|_{L^6} \| \bA \|_{\Hdiv} \leq K \| \bA \|_{\Hdiv}^2,$$
by the Sobolev inequality. By the Kato-Seiler-Simon inequality~\eqref{eq:KSS}, we obtain for the term involving $p\cdot \nabla\boW_2$,
\begin{align*}
&\Big\| \frac{1}{p^2 + m_j^2 + \omega^2} \, \boW_2 \, \frac{1}{p^2 + m_j^2 + \omega^2}\,p \cdot \nabla \boW_2 \, \frac{1}{(p^2 + m_j^2 + \omega^2)^2}\Big\|_{\gS_1}\\
& \quad \leq K \Big\| \frac{1}{p^2 + m_j^2 + \omega^2} \, \boW_2 \Big\|_{\gS_3}\, \Big\| \frac{|p|}{p^2 + m_j^2 + \omega^2} |\nabla\boW_2|^\frac{1}{4} \Big\|_{\gS_6}\times\\
& \quad \quad \times \Big\| |\nabla\boW_2|^\frac{3}{4}\frac{1}{(p^2 + m_j^2 + \omega^2)^2} \Big\|_{\gS_2}\\
& \quad \leq \frac{K}{(m_j^2 + \omega^2)^2} \big\| \bA \big\|_{\Hdiv}^4,
\end{align*}
for some universal constant $K$. The argument is exactly the same for the term involving $\nabla \boW_2 \cdot p$ instead of $p \cdot \nabla \boW_2$. Therefore, we obtain the bound 
\begin{align*}
\int_\R \sum_{j = 0}^2 |c_j| \, \Big\| \frac{1}{(p^2 + m_j^2 + \omega^2)^3} \, \big\{ p, \nabla \boW_2 \big\}_{\R^3} \, & \frac{1}{p^2 + m_j^2 + \omega^2} \, \boW_2 \Big\|_{\gS_1} \omega^2 \, d\omega\\
\leq & K \bigg( \sum_{j = 0}^2 \frac{|c_j|}{m_j} \bigg) \big\|\bA \big\|_{\Hdiv}^4.
\end{align*}
In particular, we have shown that the operator $T_{4, 1}(\bA)$ can be written in the form
\begin{equation}
\label{eq:T41end}
T_{4, 1}(\bA) = \boT_{4, 1}(\bA) + \boS_{4, 1}(\bA),
\end{equation}
with
\begin{equation}
\label{eq:boT41}
\big\| \boT_{4, 1}(\bA) \big\|_{\gS_1} \leq K \bigg( \sum_{j = 0}^2 \frac{|c_j|}{m_j} \bigg) \big\|\bA \big\|_{\Hdiv}^4,
\end{equation}
and 
\begin{equation}
\label{def:boS41}
\boS_{4, 1}(\bA) := \frac{1}{2 \pi} \int_\R \sum_{j = 0}^2 c_j \, \frac{1}{p^2 + m_j^2 + \omega^2} \, (\boW_2)^2 \, \frac{1}{(p^2 + m_j^2 + \omega^2)^2} \, \omega^2 \, d\omega.
\end{equation}
By the Kato-Seiler-Simon inequality, this term is trace-class when $\bA \in L^4(\R^3, \R^4)$ and conditions~\eqref{cond:PV} are fulfilled. On the other hand, there is no evidence that the trace-class norm of $\boS_{4, 1}(\bA)$ can be bounded using only the norm $\| \nabla \bA \|_{L^2}$. Fortunately, this term will cancel with the other ones of the same type, as we will explain later.

Our strategy to handle the operators $T_{4, 2}(\bA)$ and $T_{4, 3}(\bA)$ follows exactly the same lines. We first simplify the expressions of $T_{4, 2}(\bA)$ and $T_{4, 3}(\bA)$ by discarding the terms containing the operator $B \cdot \bsSigma$. Concerning $T_{4, 3}(\bA)$, we can compute
\begin{align*}
& \Big\| \frac{1}{p^2 + m_j^2 + \omega^2} \Big( \boW_1 \, \frac{1}{p^2 + m_j^2 + \omega^2} \Big)^3 B \cdot \bsSigma \, \frac{1}{p^2 + m_j^2 + \omega^2} \Big\|_{\gS_1}\\
\leq & \frac{K}{(m_j^2 + \omega^2)^2} \, \| B \|_{L^2} \Big( \| \nabla \bA \|_{L^2}^3 + \frac{m_j^3}{(m_j^2 + \omega^2)^\frac{3}{2}} \| \nabla V \|_{L^2}^3 + \| B \|_{L^2}^3 \Big),
\end{align*}
so that
$$T_{4, 3}(\bA) = \boT_{4, 3}^a(\bA) + \frac{1}{2 \pi} \int_\R \sum_{j = 0}^2 \frac{c_j \, \omega^2 \, d\omega}{p^2 + m_j^2 + \omega^2} \bigg( \big( \boW_1 - B \cdot \bsSigma \big) \, \frac{1}{p^2 + m_j^2 + \omega^2} \bigg)^4,$$
with
\begin{equation}
\label{eq:boT43a}
\big\| \boT_{4, 3}^a(\bA) \big\|_{\gS_1} \leq K \bigg( \sum_{j = 0}^2 \frac{|c_j|}{m_j} \bigg) \big\|\bA \big\|_{\Hdiv}^4.
\end{equation}
We next commute, as above, the operator $\boW_1 - B \cdot \bsSigma$ with the operator $1/(p^2 + m_j^2 + \omega^2)$ in order to establish that
\begin{align*}
& \frac{1}{2 \pi} \int_\R \omega^2 \, d\omega \sum_{j = 0}^2 c_j \frac{1}{p^2 + m_j^2 + \omega^2} \bigg( \big( \boW_1 - B \cdot \bsSigma \big) \, \frac{1}{p^2 + m_j^2 + \omega^2} \bigg)^4\\
& = \boT_{4, 3}^b(\bA) +\frac{1}{2 \pi} \int_\R \omega^2 \, d\omega \sum_{j = 0}^2 c_j \times\\
& \times \frac{1}{(p^2 + m_j^2 + \omega^2)^4} \Big( \big\{ p, A - V \bsalpha \big\}_{\R^3} - 2 m_j V \bsbeta \Big)^4\frac{1}{p^2 + m_j^2 + \omega^2},
\end{align*}
where $\boT_{4, 3}^b(\bA)$ also satisfies~\eqref{eq:boT43a}. Finally, we use that 
$$\big\{ p, A - \bsalpha \cdot V \big\}_{\R^3} = 2 p \cdot \big( A - \bsalpha \cdot V \big) - i \bsalpha \cdot \nabla V,$$
as well as the anti-commutation formulas for the matrices $\bsalpha_k$ and $\bsbeta$, to obtain the formula
\begin{align*}
\frac{1}{2 \pi} \int_\R & \omega^2 \, d\omega \sum_{j = 0}^2 c_j \frac{1}{(p^2 + m_j^2 + \omega^2)^4} \times\\
& \times \Big( \big\{ p, A - V \bsalpha \big\}_{\R^3} - 2 m_j V \bsbeta \Big)^4\frac{1}{p^2 + m_j^2 + \omega^2} = \boT_{4, 3}^c(\bA) + \boS_{4, 3}(\bA),
\end{align*}
with $\boT_{4, 3}^c(\bA)$ satisfying again~\eqref{eq:boT43a}, and
\begin{equation}
\label{def:boS43}
\begin{split}
& \boS_{4, 3}(\bA) := \frac{8}{\pi} \int_\R\omega^2 \, d\omega \sum_{j = 0}^2 c_j\,\frac{1}{(p^2 + m_j^2 + \omega^2)^4} \bigg( (p^2 + m_j^2)^2 V^4\\
& \quad - 4 (p^2 + m_j^2) (m_j \bsbeta + p \cdot \bsalpha) (p \cdot A) V^3 \bsbeta + 6 (p^2 + m_j^2) \sum_{l = 1}^3 p_l \big( p \cdot A) A_\ell V^2\\
& \quad - 4 m_j \sum_{l = 1}^3 \sum_{m = 1}^3 p_l p_m (m_j \bsbeta + p \cdot \bsalpha) (p \cdot A) A_l A_m V\\
& \quad + \sum_{l = 1}^3 \sum_{m = 1}^3 \sum_{n = 1}^3 p_l p_m p_n (p \cdot A) A_l A_m A_n \bigg)\frac{1}{p^2 + m_j^2 + \omega^2}.
\end{split}
\end{equation}
The computation leading to this formula is tedious but elementary. 
In conclusion, setting $\boT_{4, 3}(\bA) := \boT_{4, 3}^a(\bA) + \boT_{4, 3}^b(\bA) + \boT_{4, 3}^c(\bA)$, we have established that
\begin{equation}
\label{eq:T43end}
T_{4, 3}(\bA) = \boT_{4, 3}(\bA) + \boS_{4, 3}(\bA), 
\end{equation}
where $\boT_{4, 3}(\bA)$ satisfies~\eqref{eq:boT43a}. Similarly, one can check that
\begin{equation}
\label{eq:T42end}
\begin{split}
T_{4, 2}(\bA) = \boT_{4, 2}(\bA) + \boS_{4, 2}(\bA),
\end{split}
\end{equation}
with
\begin{equation}
\label{eq:boT42}
\big\| \boT_{4, 2}(\bA) \big\|_{\gS_1} \leq K \bigg( \sum_{j = 0}^2 \frac{|c_j|}{m_j} \bigg) \big\|\bA \big\|_{\Hdiv}^4,
\end{equation}
and
\begin{equation}
\label{def:boS42}
\begin{split}
\boS_{4, 2}(\bA) & := \frac{2}{\pi} \int_\R\omega^2 \, d\omega \sum_{j = 0}^2 c_j \frac{1}{(p^2 + m_j^2 + \omega^2)^3} \bigg( (p^2 + m_j^2) \big( 3 |A|^2 + 3 V^2\\
& - 2 (\bsalpha \cdot A) \big) V^2 - 2 (p \cdot \bsalpha + m_j \bsbeta) (p \cdot A) \big( 3 |A|^2 + 5 V^2 \big) V\\
& + 3 \sum_{l = 1}^3 p_l (p \cdot A) A_l \big( |A|^2 + 5 V^2 - 2 (\bsalpha \cdot A) V \big) \bigg)\frac{1}{p^2 + m_j^2 + \omega^2}.
\end{split}
\end{equation}
Notice here again that the Kato-Seiler-Simon inequality implies that $\boS_{4, 2}(\bA)$ and $\boS_{4, 3}(\bA)$ are trace-class when $\bA \in L^4(\R^3, \R^4)$ and conditions~\eqref{cond:PV} are satisfied. Therefore, we always assume that $\bA \in L^4(\R^3, \R^4)$ to make our calculations meaningful.

The last step in the proof is to compute the traces of the singular operators $\boS_{4, 1}(\bA)$, $\boS_{4, 2}(\bA)$ and $\boS_{4, 3}(\bA)$ for $\bA \in L^4(\R^3,\R^4)$. As announced before we claim that
\begin{equation}
\label{eq:key4}
\tr \boS_{4, 1}(\bA) - \tr \boS_{4, 2}(\bA) + \tr \boS_{4, 3}(\bA) = 0,
\end{equation}
an identity which is enough to complete the proof of Lemma~\ref{lem:4thterm}. To prove this we could make up an abstract argument based on gauge invariance. However we have to be careful with the fact that even if we can freely exchange the trace with the integration over $\omega$, these only make sense \emph{after} we have taken the sum over the coefficients $c_j$. The order matters and this complicates the mathematical analysis. Instead, we calculate the sum explicitly and verify that it is equal to $0$.

A simple computation in Fourier space shows that
\begin{equation}
\label{eq:trS41}
\begin{split}
\tr \boS_{4, 1}(\bA) = \frac{2}{\pi (2 \pi)^\frac{3}{2}} \int_{\R^3} \, dp \int_\R \omega^2 \, d\omega & \bigg( \sum_{j = 0}^2 c_j \, \frac{1}{(p^2 + m_j^2 + \omega^2)^3} \bigg) \times\\
\times & \int_{\R^3} \big( |A|^4 + 6 |A|^2 V^2 + V^4 \big).
\end{split}
\end{equation}
Similarly, one can check that
\begin{align*}
\tr \boS_{4, 2}(\bA) = & \frac{8}{\pi (2 \pi)^\frac{3}{2}} \int_{\R^3} dp \int_\R \omega^2 \, d\omega \bigg( \sum_{j = 0}^2 c_j \, \frac{1}{(p^2 + m_j^2 + \omega^2)^4} \bigg) \times\\
& \times \bigg( 3 \big( p^2 + m_j^2 \big) \int_{\R^3} \big( |A|^2 V^2 + V^4 \big)\\
& \quad + 3 \sum_{l = 1}^3 \sum_{m = 1}^3 p_l p_m \int_{\R^3} \big( A_l A_m |A|^2 + 5 A_l A_m V^2 \big) \bigg).
\end{align*}
An integration by parts shows that 
$$\int_{\R^3} \sum_{j = 0}^2 c_j \, \frac{p_l \, p_m \, dp}{(p^2 + m_j^2 + \omega^2)^4} = \frac{\delta_{l, m}}{6} \int_{\R^3} \sum_{j = 0}^2 c_j \, \frac{dp}{(p^2 + m_j^2 + \omega^2)^3},$$
and we obtain
\begin{equation}
\label{eq:trS42}
\begin{split}
& \tr \boS_{4, 2}(\bA)\\
& = \frac{24}{\pi (2 \pi)^\frac{3}{2}} \int_{\R^3} dp \int_\R \omega^2 \, d\omega \, \bigg( \sum_{j = 0}^2 c_j \, \frac{p^2 + m_j^2}{(p^2 + m_j^2 + \omega^2)^4} \bigg) \int_{\R^3} \big( |A|^2 V^2 + V^4 \big)\\
& + \frac{4}{\pi (2 \pi)^\frac{3}{2}} \int_{\R^3} dp \int_\R \omega^2 \, d\omega \, \bigg( \sum_{j = 0}^2 c_j \, \frac{1}{(p^2 + m_j^2 + \omega^2)^3} \bigg) \int_{\R^3} \big( |A|^4 + 5 |A|^2 V^2 \big).
\end{split}
\end{equation}
Similar computations lead to the expression
\begin{equation}
\label{eq:trS43}
\begin{split}
\tr \boS_{4, 3}(\bA) & := \frac{32}{\pi (2 \pi)^\frac{3}{2}} \bigg( \int_\R \int_{\R^3} \sum_{j = 0}^2 c_j \, \frac{(p^2 + m_j^2)^2 \, dp \, \omega^2 \, d\omega}{(p^2 + m_j^2 + \omega^2)^5} \bigg) \int_{\R^3} V^4\\
+ \frac{8}{\pi (2 \pi)^\frac{3}{2}} & \bigg( \int_\R \int_{\R^3} \sum_{j = 0}^2 c_j \, \frac{dp \, \omega^2 \, d\omega}{(p^2 + m_j^2 + \omega^2)^3} \Big( 1 + \frac{3 (p^2 + m_j^2)}{p^2 + m_j^2 + \omega^2} \Big) \bigg) \int_{\R^3} |A|^2 V^2\\
+ \frac{2}{\pi (2 \pi)^\frac{3}{2}} & \bigg( \int_\R \int_{\R^3} \sum_{j = 0}^2 c_j \, \frac{dp \, \omega^2 \, d\omega}{(p^2 + m_j^2 + \omega^2)^3} \bigg) \int_{\R^3} |A|^4.
\end{split}
\end{equation}
In view of~\eqref{eq:trS41} and~\eqref{eq:trS42}, we obtain
\begin{align*}
\tr \boS_{4, 1}(\bA) - \tr \boS_{4, 2}(\bA) + \tr \boS_{4, 3}(\bA) = \frac{2}{\pi (2 \pi)^\frac{3}{2}} \bigg( \int_{\R^3} V^4 \bigg) \int_\R \omega^2 \, d\omega \int_{\R^3} dp \times\\
\times \sum_{j = 0}^2 c_j \, \bigg( \frac{1}{(p^2 + m_j^2 + \omega^2)^3} - 12 \frac{p^2 + m_j^2}{(p^2 + m_j^2 + \omega^2)^4} + 16 \frac{(p^2 + m_j^2)^2}{(p^2 + m_j^2 + \omega^2)^5} \bigg).
\end{align*}
A direct computation then shows that
$$\int_\R \bigg( \frac{1}{(1 + \omega^2)^3} - \frac{12}{(1 + \omega^2)^4} + \frac{16}{(1 + \omega^2)^5} \bigg) \, \omega^2 \, d\omega = 0.$$
This is enough to deduce~\eqref{eq:key4}, and complete the proof of Lemma~\ref{lem:4thterm}.
\end{proof}

%%%%%%%%%%%%%%%%%%%%%%%%%%%%%%%%%%%%%%%%%%%%%%%
\subsection{Regularity of the sixth order term}
\label{sub:6th-smooth}
%%%%%%%%%%%%%%%%%%%%%%%%%%%%%%%%%%%%%%%%%%%%%%%

In this section, we come back to the sixth order term studied in the proof of Proposition~\ref{prop:trace-class}. The sixth order term is defined as
\begin{equation}
\label{def:boR6-2}
\boR_6(\bA) = \frac{1}{4 \pi} \int_\R \tr \big( R_6'(\omega, \bA) + R_6'(- \omega, \bA) \big) \, d\omega,
\end{equation}
where
$$R_6'(\omega, \bA) := \sum_{j = 0}^2 c_j \, \frac{i \omega}{D_{m_j, \bA} + i \omega} \Big( \big( \bsalpha \cdot A - V \big) \frac{1}{D_{m_j, 0} + i \omega} \Big)^6.$$
We have shown that it is trace-class when $\bA \in \Hdiv$. We can indeed write estimate~\eqref{eq:6th-final} as
\begin{equation}
\label{eq:6th-final_bis}
\big\| \boR_6(\bA) \big\|_{\gS_1} \leq \frac1{2\pi}\int_\R \big\| R_6'(\omega, \bA) \big\|_{\gS_1} \, d\omega \leq K \bigg( \sum_{j = 0}^2 \frac{|c_j|}{m_j^2} \bigg) \big\| \bA \big\|_{\Hdiv}^6 .
\end{equation}
Here we want to prove that $\boR_6$ is actually smooth, under suitable assumptions on $\bA$. We first establish the continuity of $\boR_6$ through

\begin{lemma}[Continuity of the sixth order term]
\label{lem:6th-cont}
The functional $\boR_6$ is locally $\theta$--H\"older continuous on the space $\Hdiv$ for any $0 < \theta <1$.
\end{lemma}

\begin{proof}
We consider the difference $R_6'(\omega, \bA) - R_6'(\omega, \bA')$ for four-potentials $\bA$ and $\bA'$ in a given ball of $\Hdiv$, which we write as
\begin{equation}
\label{eq:dec-R6}
\begin{split}
R_6' & (\omega, \bA) - R_6'(\omega, \bA')\\
& = \sum_{j = 0}^2 c_j \bigg( i \omega \Big( \frac{1}{D_{m_j, \bA} + i \omega} - \frac{1}{D_{m_j, \bA'} + i \omega} \Big) \Big( \big( \bsalpha \cdot A - V \big) \frac{1}{D_{m_j, 0} + i \omega} \Big)^6\\
& + \sum_{k = 0}^5 \frac{i \omega}{D_{m_j, \bA'} + i \omega} \Big( \big( \bsalpha \cdot A' - V' \big) \frac{1}{D_{m_j, 0} + i \omega} \Big)^k \times\\
& \times \big( \bsalpha \cdot (A - A') - V + V' \big) \frac{1}{D_{m_j, 0} + i \omega} \Big( \big( \bsalpha \cdot A - V \big) \frac{1}{D_{m_j, 0} + i \omega}\Big)^{5 - k} \bigg).
\end{split}
\end{equation}
The five terms in the sum over the index $k$ can be estimated similarly as in the proof of~\eqref{eq:6th-final}. Their $\gS_1$--norms are bounded by a universal constant $K$ times
$$\sum_{j = 0}^2 \frac{|c_j|}{(m_j^2 + \omega^2)^\frac{3}{2}} \Big( \| \bA \|_{\Hdiv}^5 + \| \bA' \|_{\Hdiv}^5 \Big) \big\| \bA - \bA' \big\|_{\Hdiv}.$$
If we follow the same proof for the first term, we need an estimate on the operator norm
$$\bigg\| i \omega \Big( \frac{1}{D_{m_j, \bA} + i \omega} - \frac{1}{D_{m_j, \bA'} + i \omega} \Big) \bigg\|.$$
On one hand, we remark that $\| (D_{m_j, \bA} + i \omega)^{- 1} \| \leq 1/\omega$, so that
\begin{equation}
\label{eq:estim_difference_1}
\bigg\| i \omega \Big( \frac{1}{D_{m_j, \bA} + i \omega} - \frac{1}{D_{m_j, \bA'} + i \omega} \Big) \bigg\| \leq 2.
\end{equation}
On the other hand, we can use the resolvent formula to write
\begin{equation}
\label{eq:resolvent_expansion_A_A'}
\begin{split}
\frac{1}{D_{m_j, \bA} + i \omega} & - \frac{1}{D_{m_j, \bA'} + i \omega}\\
& = \frac{1}{D_{m_j, \bA} + i \omega} \big( \bsalpha \cdot (\bA - \bA') + V' - V \big) \frac{1}{D_{m_j, \bA'} + i \omega}.
\end{split}
\end{equation}
For small $\bA$, or small $\bA'$, we have no problem in estimating this term using that the spectrum of $D_{m_j, \bA}$ stays away from $0$ by Lemma~\ref{lem:spectre}, and that $(D_{m_j, \bA} + i \omega)^{- 1} (D_{m_j, 0} + i \omega)$ is bounded uniformly. The argument is essentially the same in the general case.
We decompose the expression in the right-hand side of~\eqref{eq:resolvent_expansion_A_A'} as
\begin{equation}
\label{eq:dec-resolvent}
\begin{split}
i \omega \bigg( & \frac{1}{D_{m_j, \bA} + i \omega} \big( \bsalpha \cdot (\bA - \bA') + V' - V \big) \frac{1}{D_{m_j, \bA'} + i \omega} \bigg)\\
& = \Big( \frac{1}{D_{m_j, \bA} + i \omega} \big( D_{m_j, \bA} + i \mu \big) \Big) \times \Big( \frac{1}{D_{m_j, \bA} + i \mu} \big( D_{m_j, 0} + i \mu \big) \Big) \times\\
& \times \Big( \frac{1}{D_{m_j, 0} + i \mu} \big( \bsalpha \cdot (\bA - \bA') + V' - V \big) \Big) \times \frac{i \omega}{D_{m_j, \bA'} + i \omega},
\end{split}
\end{equation}
for some positive number $\mu$. We check that
\begin{equation}
\label{eq:estim-resolvent1}
\Big\| \frac{1}{D_{m_j, \bA} + i \omega} \big( D_{m_j, \bA} + i \mu \big) \Big\| \leq \Big\| \frac{D_{m_j, \bA}}{D_{m_j, \bA} + i \omega} \Big\| + \Big\| \frac{i \mu}{D_{m_j, \bA} + i \omega} \Big\| \leq 1 + \frac{\mu}{|\omega|}.
\end{equation}
Setting $\mu := 4 K^2 \| \bA \|_{\Hdiv}^2$, we also remark that
\begin{equation}
\label{eq:estim_resolvent_A_free}
\begin{split}
\Big\| \frac{1}{D_{m_j, \bA} + i \mu} (D_{m_j, 0} + i \mu) \Big\| & = \Big\| \Big( 1 + \frac{1}{D_{m_j, 0} + i \mu} \, \big( V - \bsalpha \cdot A \big) \Big)^{-1} \Big\|\\
& \leq \sum_{n = 0}^\infty \Big\| \frac{1}{D_{m_j, 0} + i \mu} \, \big( V - \bsalpha \cdot A \big) \Big\|^n\\
& \leq \sum_{n = 0}^\infty \frac{K^n}{(m_j^2 + \mu^2)^\frac{n}{4}} \, \| \bA \|_{\Hdiv}^n \leq 2.
\end{split}
\end{equation}
Recalling that $(i \omega) \| (D_{m_j, \bA'} + i \omega)^{- 1} \| \leq 1$, we infer from~\eqref{eq:resolvent_expansion_A_A'},~\eqref{eq:dec-resolvent},~\eqref{eq:estim-resolvent1} and~\eqref{eq:estim_resolvent_A_free} that
\begin{equation}
\label{eq:estim_difference_2}
\begin{split}
\bigg\| (i \omega) \Big( \frac{1}{D_{m_j, \bA} + i \omega} & - \frac{1}{D_{m_j, \bA'} + i \omega} \Big) \bigg\|\\
\leq & \frac{K}{\big( m_j^2 + \mu^2 \big)^\frac{1}{4}} \Big( 1 + \frac{\mu}{|\omega|} \Big) \| \bA - \bA' \|_{\Hdiv}\\
\leq & K \Big( \frac{1}{\sqrt{m_j}} + \frac{1}{|\omega|} \| \bA \|_{\Hdiv} \Big) \big\| \bA - \bA' \big\|_{\Hdiv}.
\end{split}
\end{equation}
In this bound, we can replace $\bA$ by $\bA'$ by symmetry. Recall that we are in a given ball in $\Hdiv$, so that $\| \bA \|_{\Hdiv}$ is bounded by some constant.

Collecting estimates~\eqref{eq:estim_difference_1} and~\eqref{eq:estim_difference_2}, we have shown that
\begin{equation}
\label{eq:diff-holder}
\begin{split}
\int_\R & \| R_6'(\omega, \bA) - R_6'(\omega, \bA') \|_{\gS_1} \, d\omega\\
& \leq K \bigg( \sum_{j = 0}^2 \frac{|c_j|}{m_j^2} \bigg) \Big( \| \bA \|_{\Hdiv}^5 + \| \bA' \|_{\Hdiv}^5 \Big) \big\| \bA - \bA' \big\|_{\Hdiv}\\
& + K \bigg( \sum_{j = 0}^2 \frac{|c_j|}{m_j^2} \bigg) \Big( \| \bA \|_{\Hdiv}^6 + \| \bA' \|_{\Hdiv}^6 \Big) \boI(m_j, \bA, \bA'),
\end{split}
\end{equation}
where
\begin{align*}
\boI(m_j, \bA, \bA') := & \int_0^\infty \frac{d\omega}{(1 + \omega^2)^\frac{3}{2}} \times\\
\times & \min \Big\{ 2 , \, \Big( \frac{1}{\sqrt{m_j}} + \frac{1}{m_j \omega} \| \bA \|_{\Hdiv} \Big) \big\| \bA - \bA' \big\|_{\Hdiv} \Big\}.
\end{align*}
Assuming that $\| \bA - \bA' \|_{\Hdiv} \leq m_j/(\sqrt{m_j} + \| \bA \|_{\Hdiv})$ for any $j = 0, 1, 2$, we can estimate the integral $\boI(m_j, \bA, \bA')$ as
\begin{align*}
\big| \boI(m_j, \bA, \bA') \big| \leq & \frac{1}{\sqrt{m_j}} \bigg( \int_0^\infty \frac{d\omega}{(1 + \omega^2)^\frac{3}{2}} \bigg) \big\| \bA - \bA' \big\|_{\Hdiv}\\
& + 2 \boJ \Big( \frac{1}{2 m_j} \| \bA \|_{\Hdiv} \| \bA - \bA' \|_{\Hdiv}\Big), 
\end{align*}
with
$$\boJ(t) := \int_0^t \frac{d\omega}{(1 + \omega^2)^\frac{3}{2}} + t \int_t^\infty \frac{d\omega}{\omega (1 + \omega^2)^\frac{3}{2}}.$$
It remains to observe that
$$\boJ(t) \leq K t \big( 1 + |\log t| \big),$$
and to combine with~\eqref{def:boR6-2} and~\eqref{eq:diff-holder}, to conclude that the functional $\boR_6$ is locally $\theta$--H\"older for any $0 < \theta < 1$.
\end{proof}

We next turn to the differentiability of $\boR_6$.

\begin{lemma}[Regularity of the sixth order term]
\label{lem:6th-reg}
The functional $\boR_6$ is of class $\boC^\infty$ on the open subset $\boH$ of $\Hdiv$ containing all the four-potentials $\bA$ such that $0 \notin \sigma(D_{m_j, \bA})$ for each $j= 0, 1, 2$. Moreover, there exists a universal constant $K$ such that
\begin{equation}
\label{eq:borne_d2_R6}
\big\| {\rm d}^2 \boR_6(\bA) \big\| \leq K \sum_{j = 0}^2 \frac{|c_j|}{m_j^2} \Big( 1 + \frac{L_\bA^2}{m_j} \| \bA \|_{\Hdiv}^2 \Big) \big\| \bA \big\|_{\Hdiv}^4,
\end{equation}
where 
$$L_\bA := \max \big\{ \| (D_{m_j,\bA} + i \omega)^{- 1} (D_{m_j, 0} + i \omega) \|, \ \omega \in \R, \ j = 0, 1, 2 \big\} < \infty.$$
\end{lemma}

\begin{proof}
The proof relies on elements in the proof of Lemma~\ref{lem:6th-cont}. When $0$ is not an eigenvalue of $D_{m_j, \bA}$ for each $j = 0, 1, 2$, we can deduce from Lemma~\ref{lem:spectre} the existence of a positive constant $K_\bA$ such that
\begin{equation}
\label{eq:bounded-inverse}
\Big\| \frac{1}{D_{m_j, \bA'} + i \omega} \Big\| \leq \min \Big\{ K_\bA, \frac{1}{|\omega|} \Big\},
\end{equation}
for any $\bA' \in \Hdiv$, with $\| \bA' - \bA \|_{\Hdiv}$ small enough. As a consequence, we can replace estimate~\eqref{eq:estim-resolvent1} by the inequality
$$\Big\| \frac{1}{D_{m_j, \bA} + i \omega} \big( D_{m_j, \bA} + i \mu \big) \Big\| \leq 1 + K_\bA \mu.$$
Since
\begin{align*}
\Big\| \frac{1}{D_{m_j, \bA} + i \omega} & \big( D_{m_j, 0} + i \omega \big) \Big\|\\
& \leq \Big\| \frac{1}{D_{m_j, \bA} + i \omega} \big( D_{m_j, 0} + i \mu \big) \Big\| + |\omega - \mu| \min \Big\{ K_\bA, \frac{1}{|\omega|} \Big\},
\end{align*}
it follows that
\begin{equation}
\label{eq:thequantity}
\begin{split}
\Big\| & \frac{1}{D_{m_j, \bA} + i \omega} \big( D_{m_j, 0} + i \omega \big) \Big\|\\
& \leq \big( 1 + K_\bA \mu \big) \Big\| \frac{1}{D_{m_j, \bA} + i \mu} \big( D_{m_j, 0} + i \mu \big) \Big\| + |\omega - \mu| \min \Big\{ K_\bA, \frac{1}{|\omega|} \Big\}.
\end{split}
\end{equation}
Following the lines of the proof of~\eqref{eq:estim_resolvent_A_free}, we deduce that the quantity in the right-hand side of~\eqref{eq:thequantity} is bounded independently on $\omega$ by a positive constant $L_\bA$, depending only on the four-potential $\bA$ and the mass $m_j$. Actually, we can claim, up to a possible larger choice of $L_\bA$, that
$$\Big\| \frac{1}{D_{m_j, \bA'} + i \omega} \big( D_{m_j, 0} + i \omega \big) \Big\| \leq L_\bA,$$
for any $\omega \in \R$, $j = 0, 1, 2$, and $\bA' \in \Hdiv$, with $\| \bA' - \bA \|_{\Hdiv}$ small enough.

As a result, we can upgrade~\eqref{eq:estim_difference_2} into
$$\bigg\| i \omega \Big( \frac{1}{D_{m_j, \bA} + i \omega} - \frac{1}{D_{m_j, \bA'} + i \omega} \Big) \bigg\|\leq \frac{K L_\bA}{(m_j^2 + \omega^2)^\frac{1}{4}} \big\| \bA - \bA' \big\|_{\Hdiv}.$$
Similarly, we can compute
\begin{equation}
\label{eq:est-first-diff}
\begin{split}
\bigg\| \frac{i \omega}{D_{m_j, \bA} + i \omega} \big( \bsalpha \cdot (\bA - \bA') & + V' - V \big) \frac{1}{D_{m_j, \bA} + i \omega} \bigg\|\\
& \leq \frac{K L_\bA}{(m_j^2 + \omega^2)^\frac{1}{4}} \big\| \bA - \bA' \big\|_{\Hdiv}.
\end{split}
\end{equation}
At this stage, we can iterate the resolvent expansion in~\eqref{eq:resolvent_expansion_A_A'} to obtain
\begin{align*}
\frac{1}{D_{m_j, \bA} + i \omega} & - \frac{1}{D_{m_j, \bA'} + i \omega}\\
& = \frac{1}{D_{m_j, \bA} + i \omega} \big( \bsalpha \cdot (\bA - \bA') + V' - V \big) \frac{1}{D_{m_j, \bA} + i \omega}\\
& + \Big( \frac{1}{D_{m_j, \bA} + i \omega} \big( \bsalpha \cdot (\bA - \bA') + V' - V \big) \Big)^2 \frac{1}{D_{m_j, \bA'} + i \omega}.
\end{align*}
Inserting this identity into~\eqref{eq:dec-R6}, we can write
\begin{equation}
\label{eq:diff-R6}
R_6'(\omega, \bA) - R_6'(\omega, \bA') = {\rm d}_\bA R_6'(\omega, \bA)(\bA -\bA') + r_6' \big( \omega, \bA, \bA' \big).
\end{equation}
Here, ${\rm d}_\bA R_6'(\omega, \bA)$ refers to the continuous linear mapping from $\Hdiv$ to $\gS_1(\R^3, \R^4)$ given by
\begin{equation}
\label{eq:dR6'}
\begin{split}
{\rm d}_\bA & R_6'(\omega, \bA)(\gv, \ga)\\
= & \sum_{j = 0}^2 c_j \bigg( \frac{i \omega}{D_{m_j, \bA} + i \omega} \big( \bsalpha \cdot \ga - \gv \big) \frac{1}{D_{m_j, \bA} + i \omega} \Big( \big( \bsalpha \cdot A - V \big) \frac{1}{D_{m_j, 0} + i \omega} \Big)^6\\
+ & \sum_{k = 0}^5 \frac{i \omega}{D_{m_j, \bA} + i \omega} \Big( \big( \bsalpha \cdot A - V \big) \frac{1}{D_{m_j, 0} + i \omega} \Big)^{5 - k} \Big( \big( \bsalpha \cdot \ga - \gv \big) \frac{1}{D_{m_j, 0} + i \omega} \Big) \times\\
& \quad \times \Big( \big( \bsalpha \cdot A - V \big) \frac{1}{D_{m_j, 0} + i \omega}\Big)^k \bigg).
\end{split}
\end{equation}
In view of~\eqref{eq:est-first-diff}, and again the computations in the proof of estimate~\eqref{eq:6th-final}, the operator norm of ${\rm d}_\bA R_6'(\omega, \bA)$ is indeed bounded by
\begin{equation}
\label{eq:est-diff-R6}
\begin{split}
\big\| {\rm d}_\bA & R_6'(\omega, \bA) \big\|\\
& \leq K \sum_{j = 0}^2 \frac{|c_j|}{(m_j^2 + \omega^2)^\frac{3}{2}} \| \bA \|_{\Hdiv}^5 \Big( 1 + \frac{L_\bA}{(m_j^2 + \omega^2)^\frac{1}{4}} \| \bA \|_{\Hdiv} \Big).
\end{split}
\end{equation}
Similarly, the remainder $r_6' \big( \omega, \bA, \bA' \big)$ in~\eqref{eq:diff-R6} may be estimated as
\begin{equation}
\label{eq:estim-r6}
\begin{split}
\Big\| & r_6' \big( \omega, \bA, \bA' \big) \Big\|_{\gS_1}\\
& \leq K \| \bA - \bA' \|_{\Hdiv}^2 \sum_{j = 0}^2 \frac{|c_j|}{(m_j^2 + \omega^2)^\frac{3}{2}} \big( \| \bA \|_{\Hdiv}^4 + \| \bA' \|_{\Hdiv}^4 \big) \times\\
& \quad \times \Big( 1 + \frac{L_\bA^2}{(m_j^2 + \omega^2)^\frac{1}{2}} \big( \| \bA \|_{\Hdiv}^2 + \| \bA' \|_{\Hdiv} ^2 \big) \Big).
\end{split}
\end{equation}
Collecting~\eqref{def:boR6-2},~\eqref{eq:diff-R6},~\eqref{eq:est-diff-R6} and~\eqref{eq:estim-r6} is enough to establish the continuous differentiability of the function $\boR_6$ on a neighborhood of $\bA$, with a differential given by
\begin{equation}
\label{eq:dboR6}
{\rm d} \boR_6(\bA)(\gv, \ga) = \frac{1}{4 \pi} \int_\R \tr \Big( {\rm d}_\bA R_6'(\omega, \bA)(\gv, \ga) + {\rm d}_\bA R_6'(- \omega, \bA)(\gv, \ga) \Big) \, d\omega.
\end{equation}

Finally, we can extend the previous arguments for the continuous differentiability of $\boR_6$ to the proof that it is actually of class $\boC^\infty$. In particular, we can check that the norm of the quadratic form ${\rm d}_\bA^2 R_6'$ is bounded by
\begin{align*}
\Big\| & {\rm d}_\bA^2 R_6' \big( \omega, \bA \big) \Big\|\\
& \leq K \sum_{j = 0}^2 \frac{|c_j|}{(m_j^2 + \omega^2)^\frac{3}{2}} \| \bA \|_{\Hdiv}^4 \Big( 1 + \frac{L_\bA^2}{(m_j^2 + \omega^2)^\frac{1}{2}} \| \bA \|_{\Hdiv}^2 \Big).
\end{align*}
Estimate~\eqref{eq:borne_d2_R6} follows integrating with respect to $\omega$.
\end{proof}

When $\| \bA \|_{\Hdiv}$ is small enough, we can prove that the constant $L_\bA$ does not depend on $\bA$.

\begin{cor}[Estimate in a neighborhood of zero]
Assume that $c_j$ and $m_j$ satisfy~\eqref{cond:c_i_thm}. There exists a universal constant $\eta$ such that, given any $\bA \in \Hdiv$ with $\| \bA \|_{\Hdiv} \leq \eta \sqrt{m_0}$, the functional $\boR_6$ is of class $\boC^\infty$ on the ball 
$$\boB_\bA(\eta \sqrt{m_0}) = \big\{ \bA \in\Hdiv \ : \ \| \bA \|_{\Hdiv} \leq \eta \sqrt{m_0} \big\},$$
and satisfies the estimate
\begin{equation}
\label{eq:borne_d2_R6_0}
\big\| {\rm d}^2 \boR_6(\bA) \big\| \leq K \bigg( \sum_{j = 0}^2 \frac{|c_j|}{m_j^2} \bigg) \, \big\| \bA \big\|_{\Hdiv}^4.
\end{equation}
\end{cor}

\begin{proof}
When $\bA$ is small enough, the spectrum of $D_{m_j, \bA}$ does not contain $0$ by Lemma~\ref{lem:spectre}. Moreover, when
$$\big\| \bA \big\|_{\Hdiv} \leq \eta \min \big\{ \sqrt{m_0}, \sqrt{m_1}, \sqrt{m_2} \big\} = \eta \sqrt{m_0},$$
for $\eta$ small enough, we can infer from~\eqref{eq:estim_resolvent_A_free} that
$$L_\bA \leq 2.$$
Inserting in~\eqref{eq:borne_d2_R6}, and using the inequality $\| \bA \|_{\Hdiv} \leq \eta \sqrt{m_j}$, gives estimate~\eqref{eq:borne_d2_R6_0}.
\end{proof}

%%%%%%%%%%%%%%%%%%%%%%%%%%%%%%%%%%%%%%%%%
%%%%%%%%%%%%%%%%%%%%%%%%%%%%%%%%%%%%%%%%%
%%%%%%%%%%%%%%%%%%%%%%%%%%%%%%%%%%%%%%%%%
\section{Proof of Theorem~\ref{thm:defF}}
\label{sub:endproof}
%%%%%%%%%%%%%%%%%%%%%%%%%%%%%%%%%%%%%%%%%
%%%%%%%%%%%%%%%%%%%%%%%%%%%%%%%%%%%%%%%%%
%%%%%%%%%%%%%%%%%%%%%%%%%%%%%%%%%%%%%%%%%

With the results of the previous section at hand, the proof of Theorem~\ref{thm:defF} is only a few lines. As a matter of fact, given any $\bA \in L^1(\R^3, \R^4) \cap \Hdiv$, we have shown that the functional $\boF_{\rm PV}(\bA)$ is well-defined by the expression
\begin{equation}
\label{eq:develop}
\boF_{\rm PV}(\bA) = \boF_2(\bF) + \boR(\bA),
\end{equation}
where 
$$\boF_2(\bF) := \tr \big( \tr_{\C^4} T_2(\bA)\big),$$
and
$$\boR(\bA) := \tr \big( \tr_{\C^4} T_4(\bA)\big) + \tr \big( \tr_{\C^4} T'_6(\bA)\big),$$
are defined in~\eqref{eq:dev-TA}. By Lemma~\ref{lem:2ndform}, the function $\boF_2$ is given by~\eqref{eq:formF2} and it is quadratic with respect to $\bF$. Since $M$ is bounded, we deduce that $\boF_2$ is smooth on $L^2(\R^3, \R^6)$. On the other hand, the function $\bA \mapsto \boF_4(\bA) := \tr(\tr_{\C^4} T_4(\bA))$ is quartic and satisfies~\eqref{eq:trT4A}. Hence, it is a smooth function on $\Hdiv$. We have proved separately in Lemma~\ref{lem:6th-cont} above that $\boR_6(\bA) = \tr(\tr_{\C^4} T'_6(\bA))$ is an H\"older continuous function on $\Hdiv$, which satisfies~\eqref{eq:6th-final}. We deduce from all this that $\boF_{\rm PV}$ has a unique continuous extension to $\Hdiv$, which is given by~\eqref{eq:develop}, and that $\boR$ satisfies estimate~\eqref{est:R}. The properties of $M$ can be found in Lemma~\ref{lem:prop_M}.\qed

%%%%%%%%%%%%%%%%%%%%%%%%%%%%%%%%%%%%%%%%%%%%%%%%%%%%%%
%%%%%%%%%%%%%%%%%%%%%%%%%%%%%%%%%%%%%%%%%%%%%%%%%%%%%%
%%%%%%%%%%%%%%%%%%%%%%%%%%%%%%%%%%%%%%%%%%%%%%%%%%%%%%
\section{Proof of Theorem~\ref{thm:differentiability}}
\label{sec:endproof_diff}
%%%%%%%%%%%%%%%%%%%%%%%%%%%%%%%%%%%%%%%%%%%%%%%%%%%%%%
%%%%%%%%%%%%%%%%%%%%%%%%%%%%%%%%%%%%%%%%%%%%%%%%%%%%%%
%%%%%%%%%%%%%%%%%%%%%%%%%%%%%%%%%%%%%%%%%%%%%%%%%%%%%%

In view of the results in Sections~\ref{sec:atrace} and~\ref{sec:estimates}, the functional $\boF_{\rm PV}$ is smooth on the open subset $\boH$ of $\Hdiv$ containing all the four-potentials $\bA$ such that $0 \notin \sigma(D_{m_j, \bA})$ for each $j= 0, 1, 2$. Indeed, the function $\bA \mapsto \boF_4(\bA) := \tr(\tr_{\C^4} T_4(\bA))$ is quartic and satisfies~\eqref{eq:trT4A}. Hence, it is of class $\boC^\infty$ on $\Hdiv$. Similarly, in view of Lemmas~\ref{lem:2ndform} and~\ref{lem:prop_M}, the quadratic map $\boF_2$ is smooth on $L^2(\R^3, \R^6)$. On the other hand, we have shown in Section~\ref{sub:6th-smooth} that $\boR_6$ is smooth when $0$ is not an eigenvalue of $D_{m_j, \bA}$ for $j = 0, 1, 2$. We deduce that $\boF_{\rm PV}$ is smooth on the set $\boH$.

In order to complete the proof of Theorem~\ref{thm:differentiability}, it remains to identify ${\rm d} \boF_{\rm PV}(\bA)$. As mentioned in Formulas~\eqref{eq:diff-boF} and~\eqref{eq:def-jA-rhoA}, this differential is related to the operator
$$Q_\bA := \sum_{j = 0}^2 c_j \, \1_{(- \infty, 0)} \big( D_{m_j, \bA} \big).$$
Concerning the properties of the operator $Q_\bA$, we can establish the following

\begin{lemma}[Properties of $\rho_\bA$ and $j_\bA$]
\label{lem:prop-QA}
Assume that $c_j$ and $m_j$ satisfy conditions~\eqref{cond:c_i_thm}.

\smallskip

\noindent $(i)$ Let $\bA \in \Hdiv$ be a four-potential such that $0$ is not an eigenvalue of $D_{m_j, \bA}$ for $j = 0, 1 , 2$. Then the operators $\tr_{\C^4} \, Q_\bA$ and $\tr_{\C^4} \, \bsalpha Q_\bA$ are locally trace-class on $L^2(\R^3, \R^4)$. More precisely,
given any function $\chi\in L_c^\infty(\R^3)$ (that is, bounded with compact support), the maps
$$\bA \in \boH \mapsto \chi \big( \tr_{\C^4} \, Q_\bA \big) \chi \in \gS_1$$
and
$$\bA \in \boH \mapsto \chi \big( \tr_{\C^4} \bsalpha \, Q_\bA \big) \chi \in \gS_1$$
are continuous from $\boH$ to $\gS_1$. In particular, the density $\rho_\bA$ and the current $j_\bA$, given by
$$\rho_\bA(x) := \big[ \tr_{\C^4} Q_\bA \big](x, x) \quad {\rm and} \quad j_\bA(x) := \big[\tr_{\C^4} \bsalpha \, Q_\bA \big](x, x),$$
are well-defined and locally integrable on $\R^3$. Moreover, the maps $\bA \mapsto \rho_\bA \, \chi^2$ and $\bA \mapsto j_\bA \, \chi^2$ are continuous from $\boH$ to $L^1(\R^3)$. Finally, for $\bA \equiv 0$, we have 
$$\rho_0 \equiv 0 \quad {\rm and} \quad j_0 \equiv 0.$$

\medskip

\noindent $(ii)$ If moreover $\bA \in L^1(\R^3, \R^4)$, then, the operators $\tr_{\C^4} \, (Q_\bA - Q_0)$ and $\tr_{\C^4} \, \bsalpha (Q_\bA - Q_0)$ are trace-class on $L^2(\R^3, \R^4)$, and the density $\rho_\bA$ and the current $j_\bA$ are in $L^1(\R^3)$.
\end{lemma}

\begin{proof}
We split the proof into three steps. First, we consider the special case $\bA \equiv 0$.

%%%%%%%%%%%%%%%%%%%%%%%%%%%%%%%%%%%%%%%%%%%%%%%%%%%%%%%%%%%%%%%%%%%%%%%%%%%%%%%%%%%%%%%%
\addtocontents{toc}{\SkipTocEntry}
\subsection*{The operators $\tr_{\C^4} Q_0$ and $\tr_{\C^4} \bsalpha Q_0$ are locally trace-class}
%%%%%%%%%%%%%%%%%%%%%%%%%%%%%%%%%%%%%%%%%%%%%%%%%%%%%%%%%%%%%%%%%%%%%%%%%%%%%%%%%%%%%%%%

Using that $\sum_{j=0}^2c_j=0$, we can write
$$Q_0 = \sum_{j = 0}^2 c_j \, \Big( \1_{(- \infty, 0)} \big( D_{m_j, 0} \big) - \frac{1}{2} \Big) = -\frac{1}{2} \tr_{\C^4} \sum_{j = 0}^2 c_j \, \frac{D_{m_j, 0}}{|D_{m_j, 0}|}.$$
As a consequence, we obtain
$$\tr_{\C^4} Q_0 = 0.$$
In particular, the density $\rho_0$ is well-defined and it identically vanishes on $\R^3$.
Similarly, we have
$$\tr_{\C^4} \bsalpha Q_0 = - 2 \sum_{j = 0}^2 c_j \, \frac{p}{(p^2 + m_j^2)^\frac{1}{2}}.$$
Due to conditions~\eqref{cond:PV}, the latter function behaves like
$$\sum_{j = 0}^2 c_j \, \frac{p}{(p^2 + m_j^2)^\frac{1}{2}} \sim \frac{3}{8} \sum_{j = 0}^2 c_j \, m_j^4 \, \frac{p}{|p|^5},$$
as $|p| \to \infty$, hence it is in $L^1(\R^3)$. By the Kato-Seiler-Simon inequality~\eqref{eq:KSS}, we deduce that the operator $\tr_{\C^4} \, \bsalpha \chi Q_0 \chi$ is trace-class for any $\chi \in L^2(\R^3)$. Hence, $\tr_{\C^4} \, \bsalpha Q_0$ is locally trace-class. In particular, the current $j_0$ is well-defined and locally integrable on $\R^3$. Moreover, we can compute
\begin{align*}
\int_{\R^3} j_0 \, \chi^2 = & \tr \big( \tr_{\C^4} \, \bsalpha \chi Q_0 \chi \big)\\
= & - \frac{1}{4 \pi^3} \int_{\R^3} \int_{\R^3} \sum_{j = 0}^2 c_j \, \frac{p}{(p^2 + m_j^2)^\frac{1}{2}} |\widehat{\chi}(q - p)|^2 \, dp \, dq,
\end{align*}
which shows that
$$j_0=- \frac{1}{4 \pi^3}\int_{\R^3}\sum_{j = 0}^2 c_j \, \frac{p}{(p^2 + m_j^2)^\frac{1}{2}}\,dp=0$$
by rotational symmetry.

We next consider the general case.

%%%%%%%%%%%%%%%%%%%%%%%%%%%%%%%%%%%%%%%%%%%%%%%%%%%%%%%%%%%%%%%%%%%%%%%%%%%%%%%%%%%%%%%%%%%%%%
\addtocontents{toc}{\SkipTocEntry}
\subsection*{The operators $\tr_{\C^4} Q_\bA$ and $\tr_{\C^4} \bsalpha Q_\bA$ are (locally) trace-class}
%%%%%%%%%%%%%%%%%%%%%%%%%%%%%%%%%%%%%%%%%%%%%%%%%%%%%%%%%%%%%%%%%%%%%%%%%%%%%%%%%%%%%%%%%%%%%%

From the previous discussion, we conclude that it is sufficient to prove that the operators $\tr_{\C^4} (Q_\bA - Q_0)$ and $\tr_{\C^4} \bsalpha (Q_\bA - Q_0)$ are locally trace-class. The corresponding charge and current densities will be the same as that of $Q_\bA$.

Concerning the (local) trace-class nature of the operator $Q_\bA - Q_0$, we follow the proof of Proposition~\ref{prop:trace-class}. Our starting point is the integral formula
\begin{equation}
\label{eq:sign}
\sign x = \frac{2}{\pi} \int_{\R} \frac{x \,\omega^2}{(x^2 + \omega^2)^2} \, d\omega = \frac{1}{2 \pi} \int_\R \Big( \frac{i \omega}{(x + i \omega)^2} - \frac{i \omega}{(x - i \omega)^2} \Big) \, d\omega.
\end{equation}
When $T$ is a self-adjoint operator on $L^2(\R^3, \R^4)$ with domain $D(T)$, we deduce that the sign of $T$ is given by
\begin{equation}
\label{eq:intsign}
\sign T = \frac{1}{2 \pi} \int_\R \Big( \frac{i \omega}{(T + i \omega)^2} - \frac{i \omega}{(T - i \omega)^2} \Big) \, d\omega,
\end{equation}
the integral in the right-hand side of~\eqref{eq:intsign} being convergent as an operator from $D(T)$ to $L^2(\R^3, \C^4)$.

In particular, the operator
$$Q_\bA - Q_0 = - \frac{1}{2} \sum_{j = 0}^2 c_j \, \big( \sign D_{m_j, \bA} - \sign D_{m_j, 0} \big),$$
is given by the expression
\begin{equation}
\label{eq:devel-Q}
\begin{split}
Q_\bA - Q_0 = -\frac{1}{4 \pi} \int_\R \sum_{j = 0}^2 c_j \, \Big( & \frac{i \omega}{(D_{m_j, \bA} + i \omega)^2} - \frac{i \omega}{(D_{m_j, 0} + i \omega)^2}\\
& - \frac{i \omega}{(D_{m_j, \bA} - i \omega)^2} + \frac{i \omega}{(D_{m_j, 0} - i \omega)^2} \Big) \, d\omega,
\end{split}
\end{equation}
on $H^1(\R^3, \C^4)$. In order to establish statements $(ii)$ and $(iii)$ of Lemma~\ref{lem:prop-QA}, we will prove that
\begin{equation}
\label{int:omega-Q}
\begin{split}
\int_\R \bigg\| \sum_{j = 0}^2 c_j \, \tr_{\C^4} \gm \, \chi \, \Big(& \frac{i \omega}{(D_{m_j, \bA} + i \omega)^2} - \frac{i \omega}{(D_{m_j, 0} + i \omega)^2}\\
& - \frac{i \omega}{(D_{m_j, \bA} - i \omega)^2} + \frac{i \omega}{(D_{m_j, 0} - i \omega)^2} \Big) \, \chi \bigg\|_{\gS_1} \, d\omega < \infty,
\end{split}
\end{equation}
for any of the matrices $\gm = I_4, \bsalpha_1, \bsalpha_2, \bsalpha_3$, and either when $\bA \in L^1(\R^3, \R^4) \cap H^1(\R^3, \R^4)$ and $\chi \equiv 1$, or when $\bA \in \Hdiv$ and $\chi \in L_c^\infty(\R^3, \R)$. In the different cases, the $\C^4$--traces of the operators $\gm (Q_\bA - Q_0)$, respectively $\gm \chi (Q_\bA - Q_0) \chi$, will define trace-class operators on $L^2(\R^3, \C^4)$. Then the operators $\tr_{\C^4} \gm Q_\bA$ will be locally trace-class and the density $\rho_\bA$ and the current $j_\bA$ will be well-defined and locally integrable on $\R^3$. Moreover, they will be integrable on $\R^3$ for any $\bA \in L^1(\R^3, \R^4) \cap H^1(\R^3, \R^4)$.

In order to prove~\eqref{int:omega-Q}, we use the expansion
\begin{equation}
\label{eq:resolvent-Q}
\begin{split}
\frac{i \omega}{(D_{m_j, \bA} + i \omega)^2} & - \frac{i \omega}{(D_{m_j, 0} + i \omega)^2} - \frac{i \omega}{(D_{m_j, \bA} - i \omega)^2} + \frac{i \omega}{(D_{m_j, 0} - i \omega)^2}\\
:= \sum_{n = 1}^5 & \Big( Q_n(\omega, \bA) + Q_n(-\omega, \bA) \Big) + Q_6'(\omega, \bA) + Q_6'(- \omega, \bA)\\
& - Q_7'(\omega, \bA) - Q_7'(- \omega, \bA),
\end{split}
\end{equation}
with
\begin{align*}
Q_n(\omega, \bA) & := (n + 1) \sum_{j = 0}^2 c_j \, \frac{i \omega}{(D_{m_j, 0} + i \omega)^2} \Big( \big( \bsalpha \cdot A - V \big) \frac{1}{D_{m_j, 0} + i \omega} \Big)^n,\\
Q_6'(\omega, \bA) & := 7 \sum_{j = 0}^2 c_j \, \frac{i \omega}{(D_{m_j, \bA} + i \omega)^2} \Big( \big( \bsalpha \cdot A - V \big) \frac{1}{D_{m_j, 0} + i \omega} \Big)^6,
\end{align*}
and
$$Q_7'(\omega, \bA) := 6 \sum_{j = 0}^2 c_j \, \frac{i \omega}{(D_{m_j, \bA} + i \omega)^2} \Big( \big( \bsalpha \cdot A - V \big) \frac{1}{D_{m_j, 0} + i \omega} \Big)^7.$$
We next estimate the terms related to the operators $Q_n(\omega, \bA)$, $Q_6'(\omega, \bA)$ and $Q_7'(\omega, \bA)$, as we have previously done for the operators $R_n(\omega, \bA)$ and $R_6'(\omega, \bA)$ in Section~\ref{sec:atrace}.

Concerning $Q_6'(\omega, \bA)$ and $Q_7'(\omega, \bA)$, we recall that $0$ is not an eigenvalue of $D_{m_j, \bA}$ for each $j = 0, 1, 2$. Hence, there exists a positive constant $K$ such that
\begin{equation}
\label{eq:no-spectrum}
\Big\| \frac{1}{D_{m_j, \bA} + i \omega} \Big\| \leq K,
\end{equation}
for all $\omega \in \R$ and $j = 0, 1, 2$. Following the proof of~\eqref{eq:6th-final}, we deduce that
\begin{align*}
\int_\R \big( \| Q_6'(\omega, \bA) \|_{\gS_1} + & \| Q_7'(\omega, \bA) \|_{\gS_1} \big) \, d\omega\\
\leq & K \sum_{j = 0}^2 \frac{|c_j|}{m_j^2} \, \| \nabla \bA \|_{L^2}^6 \, \bigg( 1 + \frac{1}{\sqrt{m_j}} \| \nabla \bA \|_{L^2} \bigg).
\end{align*}
As a consequence, the integrals
$$\boQ_6'(\bA) := \frac{1}{4 \pi} \int_\R \big( Q_6'(\omega, \bA) + Q_6'(- \omega, \bA) \big) \, d\omega,$$
and
$$\boQ_7'(\bA) := \frac{1}{4 \pi} \int_\R \big( Q_7'(\omega, \bA) + Q_7'(- \omega, \bA) \big) \, d\omega,$$
define trace-class operators on $L^2(\R^3, \R^4)$ when $\bA \in \Hdiv$. The related densities $\rho_6'(\bA)$ and $\rho_7'(\bA)$, and currents $j_6'(\bA)$ and $j_7'(\bA)$, are well-defined and integrable on $\R^3$. Moreover, in view of~\eqref{eq:no-spectrum}, we can repeat the arguments in the proof of Lemma~\ref{lem:6th-reg} in order to establish the smoothness of the maps $\bA \mapsto \boQ_6'(\bA)$ and $\bA \mapsto \boQ_7'(\bA)$ from $\boH$ onto $\gS_1$.

For $3 \leq n \leq 5$, the operators $Q_n(\omega, \bA)$ satisfy the estimates
\begin{equation}
\label{eq:Qn-Ln}
\int_\R \big\| Q_n(\omega, \bA) \big\|_{\gS_1} \, d\omega \leq K_n \big\| \bA \big\|_{L^n}^n,
\end{equation}
and
\begin{equation}
\label{eq:Qn-L6}
\int_\R \big\| \chi Q_n(\omega, \bA) \chi \big\|_{\gS_1} \, d\omega \leq K_n \big\| \bA \big\|_{L^6}^n \big\| \chi \big\|_{L^\frac{12}{6 - n}}^2,
\end{equation}
for any function $\chi \in L_c^\infty(\R^3)$. Here, $K_n$ refers to a positive constant depending only on the coefficients $c_j$ and the masses $m_j$. For $n = 4$ and $n = 5$, we can indeed use the Kato-Seiler-Simon inequality~\eqref{eq:KSS} to write
$$\big\| Q_n(\omega, \bA) \big\|_{\gS_1} \leq K \sum_{j = 0}^2 |c_j| \, \big\| \bA \big\|_{L^n}^n \, \int_{\R^3} \frac{|\omega| \, dp}{(p^2 + m_j^2 + \omega^2)^\frac{n + 2}{2}}.$$
Integrating with respect to $\omega$, we obtain inequality~\eqref{eq:Qn-Ln} with
$$K_n := K \sum_{j = 0}^2 \frac{|c_j|}{m_j^{n - 3}}.$$ 
For $n = 3$, we rely on the identity $c_0 + c_1 + c_2 = 0$ to write
\begin{align*}
Q_3(\omega, \bA) & = 4 \sum_{j = 0}^2 c_j \bigg( \sum_{k = 1}^3 \frac{i \omega}{D_{m_0, 0} + i \omega} \Big( \frac{1}{D_{m_0, 0} + i \omega} \big( \bsalpha \cdot A - V \big) \Big)^k \times\\
& \times \Big( \frac{i \omega}{D_{m_j, 0} + i \omega} - \frac{i \omega}{D_{m_0, 0} + i \omega} \Big) \Big( \big( \bsalpha \cdot A - V \big) \frac{1}{D_{m_j, 0} + i \omega} \Big)^{3 - k}\\
& + \Big( \frac{i \omega}{(D_{m_j, 0} + i \omega)^2} - \frac{i \omega}{(D_{m_0, 0} + i \omega)^2} \Big) \Big( \big( \bsalpha \cdot A - V \big) \frac{1}{D_{m_j, 0} + i \omega} \Big)^3.
\end{align*}
Using inequality~\eqref{eq:norm-diff-Dj}, we deduce that
$$\big\| Q_3(\omega, \bA) \big\|_{\gS_1} \leq K \bigg( \sum_{j = 0}^2 |c_j| (m_j - m_0) \bigg) \frac{|\omega|}{(m_0^2 + \omega^2)^\frac{3}{2}} \, \big\| \bA \big\|_{L^3}^3,$$
which provides estimate~\eqref{eq:Qn-Ln} with
$$K_3 := \sum_{j = 0}^2 |c_j| \frac{m_j - m_0}{m_0}.$$
Inequalities~\eqref{eq:Qn-L6} follow similarly. Applying the Sobolev inequality~\eqref{eq:Sobolev} to~\eqref{eq:Qn-L6}, we deduce that the integrals
$$\boQ_n(\bA) := \frac{1}{4 \pi} \int_\R \big( Q_n(\omega, \bA) + Q_n(- \omega, \bA) \big) \, d\omega,$$
define locally trace-class operators on $L^2(\R^3, \R^4)$ for $3 \leq n \leq 5$, as soon as $\bA \in \Hdiv$. The related densities $\rho_n(\bA)$ and currents $j_n(\bA)$ are well-defined and locally integrable on $\R^3$. When $\bA$ is moreover in $L^n(\R^3)$, inequality~\eqref{eq:Qn-Ln} guarantees that the operators $\boQ_n(\bA)$ are trace-class, while the functions $\rho_n(\bA)$ and $j_n(\bA)$ are integrable on $\R^3$. The continuity in these spaces follows from multi-linearity.

For $n = 1$, we refine our estimates using the cancellations provided by conditions~\eqref{cond:PV}. Following the lines of the analysis of the operator $R_1(\omega, \bA)$, we start by writing
\begin{equation}
\label{eq:dec-Q1}
Q_1(\omega, \bA) + Q_1(- \omega, \bA) = Q_{1, 1}(\omega, \bA) - Q_{1, 2}(\omega, \bA),
\end{equation}
where
\begin{equation}
\label{eq:def-Q11}
Q_{1, 1}(\omega, \bA) := 8 \sum_{j = 0}^2 c_j \, \omega^2 \frac{D_{m_j, 0}}{(D_{m_j, 0}^2 + \omega^2)^2} \big\{ \bsalpha \cdot A - V, D_{m_j, 0} \big\}_{\R^3} \frac{1}{D_{m_j, 0}^2 + \omega^2},
\end{equation}
and
$$Q_{1, 2}(\omega, \bA) := 4 \sum_{j = 0}^2 c_j \, \omega^2 \frac{1}{D_{m_j, 0}^2 + \omega^2} \big( \bsalpha \cdot A - V \big) \frac{1}{D_{m_j, 0}^2 + \omega^2}.$$
As for the operator $Q_{1, 2}(\omega, \bA)$, we combine conditions~\eqref{cond:PV} with identities~\eqref{eq:decomp_carres} to estimate
\begin{equation}
\label{eq:est-Q12}
\big\| Q_{1, 2}(\omega, \bA) \big\|_{\gS_1} \leq K \bigg( \sum_{j = 0}^2 |c_j| \big( m_j^2 - m_0^2 \big)^ 2 \bigg) \frac{\omega^2}{(m_0^2 + \omega^2)^\frac{5}{2}} \, \big\| \bA \big\|_{L^1}.
\end{equation}
In order to estimate the operator $Q_{1, 1}(\omega, \bA)$, we eliminate the odd powers of the masses $m_j$ in the numerator of the right-hand side of~\eqref{eq:def-Q11} by taking the $\C^4$--trace. Recall that
$$\big\{ \bsalpha \cdot A - V, D_{m_j, 0} \big\}_{\R^3} = \big\{ p , A - V \bsalpha \big\}_{\R^3} + B \cdot \bsSigma - 2 m_j V \bsbeta.$$
Since
\begin{equation}
\label{eq:trace-beta}
\tr_{\C^4} \bigg( \bsbeta^d \prod_{k = 1}^3 \bsalpha_k^{d_k} \bigg) = 0,
\end{equation}
when $d$ is odd, we obtain
\begin{align*}
\tr_{\C^4} \big( \gm \, Q_{1, 1}(\omega, & \bA) \big) := 8 \sum_{j = 0}^2 c_j \, \omega^2 \frac{1}{(D_{m_j, 0}^2 + \omega^2)^2} \bigg( - 2 m_j^2 V \tr_{\C^4} \big( \gm \big)\\
& + \tr_{\C^4} \Big( \gm \big( \bsalpha \cdot p \big) \big( \big\{ p , A - V \bsalpha \big\}_{\R^3} + B \cdot \bsSigma \big) \Big) \bigg) \frac{1}{D_{m_j, 0}^2 + \omega^2},
\end{align*}
for any of the matrices $\gm = I_4, \bsalpha_1, \bsalpha_2, \bsalpha_3$. On the other hand, we can compute
\begin{equation}
\label{eq:decomp-carres2}
\begin{split}
\frac{1}{(p^2 + m_j^2 + \omega^2)^2} & = \frac{1}{(p^2 + m_0^2 + \omega^2)^2} + \frac{m_0^2 - m_j^2}{(p^2 + m_j^2 + \omega^2) (p^2 + m_0^2 + \omega^2)^2}\\
& + \frac{m_0^2 - m_j^2}{(p^2 + m_j^2 + \omega^2)^2 (p^2 + m_0^2 + \omega^2)},
\end{split}
\end{equation}
as well as
\begin{align*}
\frac{1}{(p^2 + m_j^2 + \omega^2)^2} = \frac{1}{(p^2 + m_0^2 + \omega^2)^2} + & \frac{2 (m_0^2 - m_j^2)}{(p^2 + m_0^2 + \omega^2)^3}\\
+ \frac{2 (m_0^2 - m_j^2)^2}{(p^2 + m_j^2 + \omega^2) (p^2 + m_0^2 + \omega^2)^3} + & \frac{(m_0^2 - m_j^2)^2}{(p^2 + m_j^2 + \omega^2)^2 (p^2 + m_0^2 + \omega^2)^2}.
\end{align*}
Combining again with conditions~\eqref{cond:PV} and identities~\eqref{eq:decomp_carres}, we obtain the estimate
\begin{align*}
\Big\| \tr_{\C^4} \big( \gm \, Q_{1, 1}(\omega, \bA) \big) \Big\|_{\gS_1} \leq K & \sum_{j = 0}^2 |c_j| \big( m_j^2 - m_0^2 \big) \, \frac{\omega^2}{(m_0^2 + \omega^2)^\frac{5}{2}} \times\\
& \times \Big( m_j^2 \big\| V \big\|_{L^1} + \big( m_j^2 - m_0^2 \big) \big\| \bA \big\|_{L^1}\Big).
\end{align*}
In view of~\eqref{eq:dec-Q1} and~\eqref{eq:est-Q12}, we have
\begin{equation}
\label{eq:Q1-L1}
\begin{split}
\int_\R & \Big\| \tr_{\C^4} \Big( \gm \, \big( Q_1(\omega, \bA) + Q_1(- \omega, \bA) \big) \Big) \Big\|_{\gS_1} \, d\omega\\
& \leq K \sum_{j = 0}^2 |c_j| \, \big( m_j^2 - m_0^2 \big) \, \bigg( \frac{m_j^2}{m_0^2} \, \big\| V \big\|_{L^1} + \frac{m_j^2 - m_0^2}{m_0^2} \, \big\| \bA \big\|_{L^1} \bigg).
\end{split}
\end{equation}
Similarly, we can check that
\begin{equation}
\label{eq:Q1-L6}
\begin{split}
\int_\R & \Big\| \tr_{\C^4} \Big( \gm \, \chi \big( Q_1(\omega, \bA) + Q_1(- \omega, \bA) \big) \chi \Big) \Big\|_{\gS_1} \, d\omega\\
& \leq K \sum_{j = 0}^2 |c_j| \, \big( m_j^2 - m_0^2 \big) \, \bigg( \frac{m_j^2}{m_0^2} \, \big\| V \big\|_{L^6} + \frac{m_j^2 - m_0^2}{m_0^2} \, \big\| \bA \big\|_{L^6} \bigg) \, \big\| \chi \big\|_{L^\frac{12}{5}}^2.
\end{split}
\end{equation}

For $n = 2$, the analysis is identical. We compute
\begin{equation}
\label{eq:dec-Q2}
Q_2(\omega, \bA) + Q_2(- \omega, \bA) = Q_{2, 1}(\omega, \bA) - Q_{2, 2}(\omega, \bA),
\end{equation}
where
\begin{align*}
& Q_{2, 1}(\omega, \bA) := 12 \sum_{j = 0}^2 c_j \, \omega^2 \times\\
& \times \bigg( \frac{D_{m_j, 0}}{(D_{m_j, 0}^2 + \omega^2)^2} \big\{ \bsalpha \cdot A - V, D_{m_j, 0} \big\} \frac{1}{D_{m_j, 0}^2 + \omega^2} (\bsalpha \cdot A - V) \frac{D_{m_j, 0}}{D_{m_j, 0}^2 + \omega^2}\\
& + \frac{D_{m_j, 0}^2}{(D_{m_j, 0}^2 + \omega^2)^2} \big\{ \bsalpha \cdot A - V, D_{m_j, 0} \big\} \frac{1}{D_{m_j, 0}^2 + \omega^2} (\bsalpha \cdot A - V) \frac{1}{D_{m_j, 0}^2 + \omega^2} \bigg), 
\end{align*}
and
\begin{align*}
Q_{2, 2}(\omega, \bA) & := 6 \sum_{j = 0}^2 c_j \, \omega^2 \bigg( \frac{2 D_{m_j, 0}}{D_{m_j, 0}^2 + \omega^2} \Big( \big( \bsalpha \cdot A - V \big) \frac{1}{D_{m_j, 0}^2 + \omega^2} \Big)^2\\
& + \frac{1}{D_{m_j, 0}^2 + \omega^2} \big( \bsalpha \cdot A - V \big) \frac{D_{m_j, 0}}{D_{m_j, 0}^2 + \omega^2} \big( \bsalpha \cdot A - V \big) \frac{1}{D_{m_j, 0}^2 + \omega^2}\\
& + \Big( \frac{1}{D_{m_j, 0}^2 + \omega^2} \big( \bsalpha \cdot A - V \big) \Big)^2 \frac{D_{m_j, 0}}{D_{m_j, 0}^2 + \omega^2} \bigg).
\end{align*}
In order to estimate the operators $Q_{2, 1}(\omega, \bA)$ and $Q_{2, 2}(\omega, \bA)$, we again take the $\C^4$--trace. For $\gm = I_4$ or $\gm = \bsalpha_k$, we derive from~\eqref{eq:trace-beta} that
\begin{align*}
& \tr_{\C^4} \big( \gm \, Q_{2, 1}(\omega, \bA) \big) = 12 \sum_{j = 0}^2 c_j \, \omega^2 \times\\
& \times \tr_{\C^4} \bigg( \frac{\gm (\bsalpha \cdot p)}{(D_{m_j, 0}^2 + \omega^2)^2} \big\{ \bsalpha \cdot A - V, \bsalpha \cdot p \big\} \frac{1}{D_{m_j, 0}^2 + \omega^2}(\bsalpha \cdot A - V)\frac{1}{D_{m_j, 0}^2 + \omega^2}\\
& - m_j^2 \, \gm \frac{1}{(D_{m_j, 0}^2 + \omega^2)^2} \big\{ \bsalpha \cdot A + V, \bsalpha \cdot p \big\} \frac{1}{D_{m_j, 0}^2 + \omega^2}(\bsalpha \cdot A + V)\frac{1}{D_{m_j, 0}^2 + \omega^2}\\
& - 2 m_j^2 \, \gm \frac{1}{(D_{m_j, 0}^2 + \omega^2)^2} V \frac{1}{D_{m_j, 0}^2 + \omega^2}(\bsalpha \cdot A - V)\frac{\bsalpha\cdot p}{D_{m_j, 0}^2 + \omega^2}\\
& + 2 m_j^2 \, \gm \frac{\bsalpha \cdot p}{(D_{m_j, 0}^2 + \omega^2)^2} V \frac{1}{D_{m_j, 0}^2 + \omega^2}(\bsalpha \cdot A + V)\frac{1}{D_{m_j, 0}^2 + \omega^2}\\
& + \gm \frac{p^2+m_j^2}{(D_{m_j, 0}^2 + \omega^2)^2} \big\{ \bsalpha \cdot A - V, \bsalpha \cdot p \big\} \frac{1}{D_{m_j, 0}^2 + \omega^2}(\bsalpha \cdot A - V)\frac{1}{D_{m_j, 0}^2 + \omega^2}\bigg),
\end{align*}
while
\begin{align*}
& \tr_{\C^4} \big( \gm \, Q_{2, 2}(\omega, \bA) \big) = 6 \sum_{j = 0}^2 c_j \, \omega^2 \tr_{\C^4} \bigg( 2 \gm \frac{\bsalpha \cdot p}{D_{m_j, 0}^2 + \omega^2} \times\\
& \times \Big( \big( \bsalpha \cdot A - V \big) \frac{1}{D_{m_j, 0}^2 + \omega^2} \Big)^2 + \gm \frac{1}{D_{m_j, 0}^2 + \omega^2} \big( \bsalpha \cdot A - V \big) \frac{\bsalpha \cdot p}{D_{m_j, 0}^2 + \omega^2} \times\\
& \times \big( \bsalpha \cdot A - V \big) \frac{1}{D_{m_j, 0}^2 + \omega^2} + \gm \Big( \frac{1}{D_{m_j, 0}^2 + \omega^2} \big( \bsalpha \cdot A - V \big) \Big)^2 \frac{\bsalpha \cdot p}{D_{m_j, 0}^2 + \omega^2} \bigg).
\end{align*}
Invoking conditions~\eqref{cond:PV}, as well as identities~\eqref{eq:decomp_carres} and~\eqref{eq:decomp-carres2}, we deduce that
\begin{align*}
\Big\| \tr_{\C^4} \big( \gm \, Q_{2, 1}(\omega, \bA) \big) \Big\|_{\gS_1} \leq & K \sum_{j = 0}^2 |c_j| \, \bigg( \big( m_j^2 - m_0^2 \big) \, \big\| \bA \big\|_{L^2}^2 \, \frac{\omega^2}{(m_0^2 + \omega^2)^2}\\
& + m_j^2 \, \big\| \bA \big\|_{L^2} \, \big( \big\| \bA \big\|_{L^2} + \big\| V \big\|_{L^2} \big) \, \frac{\omega^2}{(m_j^2 + \omega^2)^2} \bigg),
\end{align*}
and
$$\Big\| \tr_{\C^4} \big( \gm \, Q_{2, 2}(\omega, \bA) \big) \Big\|_{\gS_1} \leq K \bigg( \sum_{j = 0}^2 |c_j| \, \big( m_j^2 - m_0^2 \big) \bigg) \big\| \bA \big\|_{L^2}^2 \, \frac{\omega^2}{(m_0^2 + \omega^2)^2}.$$
It follows that
\begin{align*}
\int_\R & \Big\| \tr_{\C^4} \Big( \gm \, \big( Q_2(\omega, \bA) + Q_2(- \omega, \bA) \big) \Big) \Big\|_{\gS_1} \, d\omega\\
& \leq K \sum_{j = 0}^2 |c_j| \bigg( \frac{m_j^2 - m_0^2}{m_0} \, \big\| \bA \big\|_{L^2}^2 + m_j \, \big\| \bA \big\|_{L^2} \, \big( \big\| \bA \big\|_{L^2} + \big\| V \big\|_{L^2} \big) \bigg).
\end{align*}
Similarly, we have
\begin{equation}
\label{eq:Q2-L6}
\begin{split}
\int_\R & \Big\| \tr_{\C^4} \Big( \gm \, \chi \big( Q_2(\omega, \bA) + Q_2(- \omega, \bA) \big) \chi \Big) \Big\|_{\gS_1} \, d\omega\\
& \leq K \sum_{j = 0}^2 |c_j| \bigg( \frac{m_j^2 - m_0^2}{m_0} \, \big\| \bA \big\|_{L^6}^2 + m_j \, \big\| \bA \big\|_{L^6} \, \big( \big\| \bA \big\|_{L^6} + \big\| V \big\|_{L^6} \big) \bigg) \, \big\| \chi \big\|_{L^3}^2.
\end{split}
\end{equation}
In view of~\eqref{eq:Q1-L1} and~\eqref{eq:Q1-L6}, we conclude that the integrals
$$\tr_{\C^4} \big( \gm \, \boQ_n \big) := \frac{1}{4 \pi} \int_\R \tr_{\C^4} \Big( \gm \, \big( Q_n(\omega, \bA) + Q_n(- \omega, \bA) \big) \Big) \, d\omega,$$
also define local trace-class operators on $L^2(\R^3, \R^4)$ for $n = 1, 2$, as soon as $\bA \in \Hdiv$. The operators are trace-class when $\bA$ is in $L^n(\R^3)$. Concerning the related densities $\rho_n(\bA)$ and currents $j_n(\bA)$, they are well-defined and locally integrable on $\R^3$ for $\bA \in \Hdiv$, and integrable on $\R^3$ for $\bA \in L^n(\R^3)$. Their continuity follows again by multi-linearity.

At this stage, it remains to recall Formulas~\eqref{eq:devel-Q} and~\eqref{eq:resolvent-Q} to complete the proof of Lemma~\ref{lem:prop-QA}.
\end{proof}

We are now in position to complete the proof of Theorem~\ref{thm:differentiability}.

\begin{proof}[End of the proof of Theorem~\ref{thm:differentiability}]
We have shown that the functional $\boF_{\rm PV}$ is smooth on the open subset $\boH$ of four-potentials $\bA$ such that $0$ is not an eigenvalue of $D_{m_j, \bA}$ for each $j = 0, 1, 2$. In particular, the differential ${\rm d} \boF_{\rm PV}(\bA)$ is a bounded form on $\Hdiv$. By duality, it can be identified with a couple of functions $(\rho_*, j_*)$ in the Coulomb space $\boC$ defined in~\eqref{eq:def_Coulomb}. Our task reduces to verify that $\rho_* = \rho_\bA$ and $j_* = - j_\bA$.

We first restrict our attention to four-potentials $\bA$ which are moreover integrable on $\R^3$. In this case, the functional $\boF_{\rm PV}(\bA)$ is given by Formula~\eqref{eq:proper-def-F}, which may be written in view of~\eqref{eq:dev-TA} as
$$\boF_{\rm PV}(\bA) = \sum_{n = 1}^5 \boF_n(\bA) + \boR_6(\bA),$$
where we recall that
$$\boF_n(\bA) := \frac{1}{4 \pi} \int_\R \tr \Big( \tr_{\C^4} \big( R_n(\omega, \bA) + R_n(- \omega, \bA) \big) \Big) \, d\omega,$$
and
$$\boR_6(\bA) := \frac{1}{4 \pi} \int_\R \tr \Big( \tr_{\C^4} \big( R_6'(\omega, \bA) + R_6'(- \omega, \bA) \big) \, d\omega.$$
We have computed the differential of ${\rm d} \boR_6(\bA)$ in~\eqref{eq:dboR6}. On the other hand, the functionals $\boF_n$ are $n$-linear with respect to $\bA$, so that their differentials are given by
\begin{align*}
{\rm d} & \boF_n(\bA)(\gv, \ga)\\
& = \frac{1}{4 \pi} \int_\R \tr \Big( \tr_{\C^4} \big( {\rm d}_\bA R_n(\omega, \bA)(\gv, \ga) + {\rm d}_\bA R_n(- \omega, \bA)(\gv, \ga) \big) \Big) \, d\omega,
\end{align*}
with
\begin{equation}
\label{eq:dRn}
\begin{split}
{\rm d}_\bA R_n(\omega, \bA)(\gv, \ga) & = \sum_{j = 0}^2 c_j \, \frac{i \omega}{D_{m_j, 0} + i \omega} \sum_{k = 0}^{n - 1} \Big( \big( \bsalpha \cdot A - V \big) \frac{1}{D_{m_j, 0} + i \omega} \Big)^k \times\\
& \times \big( \bsalpha \cdot \ga - \gv \big) \frac{1}{D_{m_j, 0} + i \omega} \Big( \big( \bsalpha \cdot A - V \big) \frac{1}{D_{m_j, 0} + i \omega} \Big)^{n - 1 - k},
\end{split}
\end{equation}
for any $(\gv, \ga) \in L^1(\R^3, \R^4) \cap \Hdiv$. It follows that the differential ${\rm d} \boF_{\rm PV}(\bA)$ is equal to
$${\rm d} \boF_{\rm PV}(\bA)(\gv, \ga) = \frac{1}{4 \pi} \int_\R \Xi(\omega, \bA)(\gv, \ga) \, d\omega,$$
with
\begin{align*}
\Xi(\omega, \bA)(\gv, \ga) := & \tr \bigg( \sum_{n = 1}^5 \tr\Big(\tr_{\C^4} \big( {\rm d}_\bA R_n(\omega, \bA)(\gv, \ga) + {\rm d}_\bA R_n(- \omega, \bA)(\gv, \ga) \big)\Big)\\
& + \tr\Big(\tr_{\C^4} \big( {\rm d}_\bA R_6'(\omega, \bA)(\gv, \ga) + {\rm d}_\bA R_6'(- \omega, \bA)(\gv, \ga) \big) \Big)\bigg).
\end{align*}
At this stage, we make use of Formulas~\eqref{eq:dR6'} and \eqref{eq:dRn} to check that
\begin{equation}
\label{eq:Phi}
\Xi(\omega, \bA)(\gv, \ga) = \tr \bigg( \tr_{\C^4} \bigg( \sum_{j = 0}^2 c_j \, \frac{i \omega}{(D_{m_j, \bA} + i \omega)^2} \big( \bsalpha \cdot \ga - \gv \big) \bigg) \bigg).
\end{equation}
Indeed, we have established in the course of Lemma~\ref{lem:6th-reg} that each term in the decomposition of ${\rm d}_\bA R_6'(\omega, \bA)(\gv, \ga)$ which is provided by Formula~\eqref{eq:dR6'} is trace-class. As a consequence, we can write
\begin{align*}
\tr & \Big( \tr_{\C^4} \, {\rm d}_\bA R_6'(\omega, \bA)(\gv, \ga) \Big)\\
& = \sum_{j = 0}^2 c_j \, \tr \, \frac{i \omega}{D_{m_j, \bA} + i \omega} \big( \bsalpha \cdot \ga - \gv \big) \frac{1}{D_{m_j, \bA} + i \omega} \Big( \big( \bsalpha \cdot A - V \big) \frac{1}{D_{m_j, 0} + i \omega} \Big)^6\\
& + \sum_{j = 0}^2 c_j \, \sum_{k = 0}^5 \tr \, \frac{i \omega}{D_{m_j, \bA} + i \omega} \Big( \big( \bsalpha \cdot A - V \big) \frac{1}{D_{m_j, 0} + i \omega} \Big)^{5 - k} \times\\
& \quad \times \big( \bsalpha \cdot \ga - \gv \big) \frac{1}{D_{m_j, 0} + i \omega} \Big( \big( \bsalpha \cdot A - V \big) \frac{1}{D_{m_j, 0} + i \omega} \Big)^k.
\end{align*}
An advantage of this further decomposition is that we are allowed to commute the products in the right-hand side, so as to obtain
\begin{align*}
\tr & \Big( \tr_{\C^4} \, {\rm d}_\bA R_6'(\omega, \bA)(\gv, \ga) \Big)\\
& = \sum_{j = 0}^2 c_j \, \tr \, \frac{i \omega}{D_{m_j, \bA} + i \omega} \Big( \big( \bsalpha \cdot A - V \big) \frac{1}{D_{m_j, 0} + i \omega} \Big)^6 \frac{1}{D_{m_j, \bA} + i \omega} \big( \bsalpha \cdot \ga - \gv \big)\\
& + \sum_{j = 0}^2 c_j \, \sum_{k = 0}^5 \tr \, \frac{i \omega}{D_{m_j, 0} + i \omega} \Big( \big( \bsalpha \cdot A - V \big) \frac{1}{D_{m_j, 0} + i \omega} \Big)^k \frac{1}{D_{m_j, \bA} + i \omega} \times\\
& \quad \times \Big( \big( \bsalpha \cdot A - V \big) \frac{1}{D_{m_j, 0} + i \omega} \Big)^{5 - k} \big( \bsalpha \cdot \ga - \gv \big).
\end{align*}
This follows from the property that the operator $(i \omega)(D_{m_j, \bA} + i \omega)^{-1}$ is bounded, while the operators $(\bsalpha \cdot A - V)(D_{m_j, 0} + i \omega)^{-1}$ and $(\bsalpha \cdot \ga - \gv)(D_{m_j, 0} + i \omega)^{-1}$ belong to suitable Schatten spaces. Using the resolvent expansion~\eqref{eq:resolvent_expansion}, we are led to
\begin{equation}
\label{eq:good-dR6'}
\begin{split}
\tr \Big( \tr_{\C^4} & \, {\rm d}_\bA R_6'(\omega, \bA)(\gv, \ga) \Big) = \tr \bigg( \tr_{\C^4} \sum_{j = 0}^2 c_j \bigg( \frac{i \omega}{(D_{m_j, \bA} + i \omega)^2} \big( \bsalpha \cdot \ga - \gv \big)\\
- & \sum_{k = 0}^4 \, \sum_{l = 0}^{4 - k} \, \frac{1}{D_{m_j, 0} + i \omega} \Big( \big( \bsalpha \cdot A - V \big) \frac{1}{D_{m_j, 0} + i \omega} \Big)^k \frac{i \omega}{D_{m_j, 0} + i \omega} \times\\
& \quad \times \Big( \big( \bsalpha \cdot A - V \big) \frac{1}{D_{m_j, 0} + i \omega} \Big)^l \big( \bsalpha \cdot \ga - \gv \big) \bigg) \bigg).
\end{split}
\end{equation}
Similarly, we can deduce from~\eqref{eq:dRn} that
\begin{align*}
\tr \Big( \tr_{\C^4} & \, {\rm d}_\bA R_n(\omega, \bA)(\gv, \ga) \Big)\\
= & \tr \bigg( \tr_{\C^4} \sum_{j = 0}^2 c_j \, \sum_{k = 0}^{n - 1} \frac{1}{D_{m_j, 0} + i \omega} \Big( \big( \bsalpha \cdot A - V \big) \frac{1}{D_{m_j, 0} + i \omega} \Big)^{n - 1 - k} \times\\
& \quad \times \frac{i \omega}{D_{m_j, 0} + i \omega} \Big( \big( \bsalpha \cdot A - V \big) \frac{1}{D_{m_j, 0} + i \omega} \Big)^k \big( \bsalpha \cdot \ga - \gv \big) \bigg).
\end{align*}
Formula~\eqref{eq:Phi} follows combining with~\eqref{eq:good-dR6'}.

As a conclusion, we have derived the following expression of ${\rm d} \boF_{\rm PV}(\bA)$,
\begin{align*}
{\rm d} \boF_{\rm PV}(\bA)(\gv, \ga) = \frac{1}{4 \pi} \int_\R\tr \bigg( \tr_{\C^4} \bigg( & \sum_{j = 0}^2 c_j \, \Big( \frac{i \omega}{(D_{m_j, \bA} + i \omega)^2}\\
& - \frac{i \omega}{(D_{m_j, \bA} - i \omega)^2} \Big) \big( \bsalpha \cdot \ga - \gv \big) \, d\omega \bigg) \bigg).
\end{align*}
In view of~\eqref{eq:sign} and Lemma~\ref{lem:prop-QA}, we deduce that
$${\rm d} \boF_{\rm PV}(\bA)(\gv, \ga) = \tr \Big( \tr_{\C^4} \Big( Q_\bA \big( \gv -\bsalpha \cdot \ga \big) \Big) \Big) = \int_{\R^3} \big( \rho_\bA \gv - j_\bA \cdot \ga \big),$$
so that $\rho_* = \rho_\bA$ and $j_* = j_\bA$, when $\bA \in L^1(\R^3, \R^4) \cap \Hdiv$.

In the general case where $\bA$ is only in $\Hdiv$, we can construct a sequence of maps $(\bA_n)_{n \in \N}$ in $ L^1(\R^3, \R^4) \cap \Hdiv$, for which $0 \notin \sigma(D_{m_j, \bA_n})$ for any $n \in \N$ and each $j = 0, 1, 2$, and such that
$$\bA_n \to \bA \quad {\rm in} \quad \Hdiv,$$
as $n \to \infty$. The existence of such a sequence follows from the density of $L^1(\R^3, \R^4) \cap \Hdiv$ in $\Hdiv$, and statement $(ii)$ in Lemma~\ref{lem:spectre}. For each integer $n$, we know that
$${\rm d} \boF_{\rm PV}(\bA_n)(\gv, \ga) = \int_{\R^3} \Big( \rho_{\bA_n} \gv - j_{\bA_n} \cdot \ga \Big),$$
for any four-potential $(\gv, \ga) \in \boC_c^\infty(\R^3, \R^4)$. Combining the continuous differentiability of the functional $\boF_{\rm PV}$ with statement $(i)$ in Lemma~\ref{lem:prop-QA}, we obtain, taking the limit $n \to \infty$,
$${\rm d} \boF_{\rm PV}(\bA)(\gv, \ga) = \int_{\R^3} \Big( \rho_{\bA} \gv - j_{\bA} \cdot \ga \Big),$$
which completes the proof of $(ii)$ in Theorem~\ref{thm:differentiability}.

Concerning $(iii)$, recall that the second order differential of $\boF_{\rm PV}$ is equal to
$${\rm d}^2 \boF_{\rm PV}(\bA) = {\rm d}^2 \boF_2(\bF) + {\rm d}^2 \boF_4(\bA) + {\rm d}^2 \boR(\bA).$$ 
Since $\boF_2$ is quadratic and $\boF_4$ is quartic, estimate~\eqref{eq:estim_Hessian} appears as a consequence of Formula~\eqref{eq:formF2}, and inequalities~\eqref{eq:trT4A} and~\eqref{eq:borne_d2_R6_0}. This completes the proof of Theorem~\ref{thm:differentiability}.
\end{proof}

%%%%%
% \bibliographystyle{siam}
% \bibliography{biblio}

\end{document}